\documentclass[11pt]{article}

\usepackage[letterpaper,margin=1in]{geometry}

\usepackage{amsmath}
\usepackage{amsthm}

\usepackage{lineno}
\usepackage{graphicx}
\graphicspath{{./art/}}

\usepackage[plain,noend]{algorithm2e}
\makeatletter
\renewcommand{\algocf@captiontext}[2]{#1\algocf@typo. \AlCapFnt{}#2} % text of caption
\def\@algocf@capt@plain{top}
\renewcommand{\algocf@makecaption}[2]{%
  \addtolength{\hsize}{\algomargin}%
  \sbox\@tempboxa{\algocf@captiontext{#1}{#2}}%
  \ifdim\wd\@tempboxa >\hsize%     % if caption is longer than a line
    \hskip .5\algomargin%
    \parbox[t]{\hsize}{\algocf@captiontext{#1}{#2}}% then caption is not centered
  \else%
    \global\@minipagefalse%
    \hbox to\hsize{\box\@tempboxa}% else caption is centered
  \fi%
  \addtolength{\hsize}{-\algomargin}%
}
\makeatother

\newtheorem{theorem}{Theorem}
\newtheorem{lemma}[theorem]{Lemma}
\newtheorem{proposition}[theorem]{Proposition}
\newtheorem{corollary}[theorem]{Corollary}
\newtheorem{assumption}[theorem]{Assumption}
\newtheorem{definition}[theorem]{Definition}
\newtheorem{example}[theorem]{Example}

\makeatletter
\renewenvironment{proof}[1][]{\par
  \pushQED{\qed}%
  \normalfont \topsep6\p@\@plus6\p@\relax
  \def\@tempproofopt{#1}%
  \trivlist
  \item[\hskip\labelsep\itshape
    \proofname\ifx\@tempproofopt\@empty\else\space#1\fi\@addpunct{.}]\ignorespaces
}{\popQED\endtrivlist\@endpefalse}
\makeatother

\usepackage{booktabs}
\usepackage{placeins}
\usepackage{multirow, dsfont, amsfonts, amssymb, color}
\usepackage{mathtools}
\usepackage{pifont}   % for \ding
\usepackage{enumitem} % for custom enumerations

\usepackage{authblk}

\providecommand{\email}[1]{\texttt{#1}}

\usepackage{fancyhdr}
\usepackage[round]{natbib}

\usepackage{xr-hyper}
\usepackage[hidelinks]{hyperref}

\newcommand{\cmark}{\ding{51}}
\newcommand{\xmark}{\ding{55}}
\newcommand{\cmarks}{\;\,\cmark\textsuperscript{$\dagger$}}
\newcommand{\cmarku}{\;\,\cmark\textsuperscript{$\uparrow$}}

\newcommand{\I}{\mathds{1}}
\newenvironment{keywords}%
  {\par\vspace{0.5em}\noindent\small\textit{Keywords}: }%
  {\par\vspace{0.5em}}

\begin{document}

\title{Conformal Prediction for Regression with Clipped Outcomes}

\author[1]{M. Sesia\thanks{\email{sesia@marshall.usc.edu}}}
\author[2]{V. Svetnik\thanks{\email{vladimir\_svetnik@merck.com}}}
\affil[1]{Departments of Data Sciences and Operations, and of Computer Science,
University of Southern California, Los Angeles, California, U.S.A.}
\affil[2]{Biostatistics and Research Decision Sciences, Merck \& Co., Inc.,
Rahway, New Jersey, U.S.A.}

\date{}
\maketitle

\begin{abstract}
We study conformal prediction for regression using calibration data with outcomes that are doubly censored (clipped) at known fixed thresholds. We show that existing methods are unsatisfactory in this setting, as they yield intervals that may have higher marginal coverage than desired and yet lose conditional coverage precisely for the easier-to-predict cases whose outcomes are typically fully observed. This reveals that marginal coverage, the usual target of conformal prediction, may not be the ideal goal under clipping. We address this challenge by introducing a new nonconformity score and calibration methods at both ends of this trade-off: one for tight marginal coverage, and a two-step method that prioritizes conditional coverage. We characterize their finite-sample coverage and oracle-like asymptotic behavior under suitable consistency of the underlying model, and we compare them to more direct adaptations of existing approaches.

%%% Local Variables:
%%% mode: latex
%%% TeX-master: "main_biometrika"
%%% End:

\end{abstract}

\begin{keywords}
Censored data, machine learning, predictive inference, regression, uncertainty quantification.
\end{keywords}

\section{Introduction}

\subsection{Motivation and Preview of Contributions}
In early-stage drug discovery, machine-learning models are used to predict the biological activity of molecules from their chemical structure \citep{vamathevan2019applications}, and the relevant outcome data are often doubly censored, or \emph{clipped}, due to the design limitations of chemical assays \citep{svensson2025enhancing}: a molecule whose activity exceeds the upper limit of the assay, or falls below its lower limit, is recorded merely as lying beyond that limit. Drug discovery is just one of many domains where such censoring arises \citep{helsel2005nondetects}. This paper focuses on censoring at fixed thresholds because this is the common single-assay case and already makes the problem interesting; we discuss possible future extension in Section~\ref{sec:discussion}.

There are many models that can learn to predict latent outcomes from doubly censored data, including
parametric Tobit regression \citep{tobin1958estimation}, gradient-boosted trees \citep{sigrist2019grabit},
and neural networks \citep{danaila2024deep}. This paper studies how conformal prediction
\citep{vovk2005algorithmic}, specifically in its inductive form,
should be combined with any such model to provide distribution-free predictive uncertainty estimation, using calibration data
clipped at known, fixed thresholds. While conformal prediction cannot
resolve outcomes outside the range of the observable data, it can
produce useful prediction intervals that either lie within the observed range or indicate that
a new outcome may fall beyond it, and in that case on which side or sides.

Conformal prediction is well understood in the standard uncensored regression setting \citep{lei2018distribution}, where existing methods target exact marginal coverage \citep{romano2019cqr} and also provide asymptotic oracle-like conditional coverage under suitable model-consistency assumptions \citep{sesia2020comparison}. Under clipping, however, these two goals come into conflict. As we show, the most direct adaptations of existing methods may be marginally conservative and yet unsatisfactory, producing prediction intervals that under-cover precisely in the easier-to-predict region of the feature space where the outcome is typically observed. We address this challenge by introducing a new nonconformity score and calibration methods at both ends of the trade-off: one for tight marginal coverage, and one that prioritizes conditional coverage. We study the finite-sample and asymptotic behavior of these methods, relate them to the more direct adaptations of existing approaches, and validate them empirically.

\subsection{Related Works}

While conformal prediction has mostly focused on fully observable outcomes, there are already
several extensions to partially observed outcomes; see \citet{sesia2026elements} for a recent review.
 For example, in survival analysis a line of research originating from \citet{candes2023conformalized} studies how 
to obtain approximate marginal coverage for time-to-event outcomes under non-informative
right-censoring.
We study a different problem where censoring depends on $Y$ rather than $X$. Our setting is closer
to that of \citet{liu2025prediction}, who construct prediction sets under general double censoring; 
their approach, however, is unsatisfactory here, as
explained in Appendix~\ref{sec:app-relation-double}.
More distantly related works study conformal prediction under
label noise \citep{einbinder2024label,sesia2025adaptive}, and other forms of imperfect
supervision \citep{stutz2023conformal,cauchois2024predictive}.% To our knowledge, this is the first
%work to study the regression with clipping problem in depth.

\section{Methodology}\label{sec:method}

\subsection{Notation and Problem Statement}\label{sec:setup}

Let $X \in \mathcal{X}$ denote the feature vector and $Y \in [Y^{\min}, Y^{\max}] \subseteq \mathbb{R}$ the
latent outcome of interest, whose support has finite range $R := Y^{\max} - Y^{\min} > 0$.
Given two fixed and known censoring thresholds $Y^{\min} \leq C^{\mathrm{L}} < C^{\mathrm{R}} \leq Y^{\max}$, instead of $Y$ we observe
the clipped outcome
\begin{align}\label{eq:clip}
    & \tilde{Y} = \Pi(Y),
    & \Pi(y) := \max\{C^{\mathrm{L}}, \min\{y, C^{\mathrm{R}}\} \}
\end{align}
Therefore, seeing $\tilde{Y} = C^{\mathrm{L}}$ only tells us that $Y \in [Y^{\min}, C^{\mathrm{L}}]$,
and analogously for $\tilde{Y} = C^{\mathrm{R}}$.

We assume throughout that $\{(X_i, Y_i, \tilde{Y}_i) \}_{i=1}^{n+1}$ are exchangeable random samples from some distribution $\mathcal{P} = \mathcal{P}_{X,Y} \times \mathcal{P}_{\tilde{Y} \mid Y}$ with unknown and arbitrary $\mathcal{P}_{X,Y}$.
The goal is to predict $Y_{n+1}$, given $X_{n+1}$ and the censored calibration data $\mathcal{D}_{\mathrm{cal}} := \{(X_i, \tilde{Y}_i )\}_{i=1}^n$ indexed by $[n] := \{1,\ldots,n\}$.

For any $x \in \mathcal{X}$, let $\hat{f}(x) = [\hat{f}^{\mathrm{lo}}(x), \hat{f}^{\mathrm{up}}(x)]$ denote
a prediction interval produced by any fixed model, pre-trained on data independent of $\mathcal{D}_{\mathrm{cal}}$, with $\hat{f}^{\mathrm{lo}}(X) \leq \hat{f}^{\mathrm{up}}(X)$ almost surely (a.s.). We make no other assumptions on this model nor on the accuracy of its predictions.
A reasonable choice is to use $\hat{f}^{\mathrm{lo}}(x) = \hat{q}_{\alpha/2}(x)$ and
$\hat{f}^{\mathrm{up}}(x) = \hat{q}_{1-\alpha/2}(x)$,
where $\hat{q}_{\alpha}$ denotes a conditional quantile estimate at level $\alpha$ for the distribution of $Y \mid X$, although this is not required for now.

Given a miscoverage level $\alpha \in (0,1)$, the typical goal is to construct an informative
prediction interval $\widehat{\mathcal{C}}(X_{n+1})$ with finite-sample marginal coverage, 
$\mathbb{P}[Y_{n+1} \in \widehat{\mathcal{C}}(X_{n+1})] \geq 1-\alpha$,
taking the probability over both $(X_{n+1},Y_{n+1})$ and $\mathcal{D}_{\mathrm{cal}}$. Marginal coverage is not as strong and 
informative as \emph{conditional} coverage,
$\mathbb{P}[Y_{n+1} \in \widehat{\mathcal{C}}(X_{n+1}) \mid X_{n+1}=x] \geq 1-\alpha$ for almost all
$x$, but the latter cannot be guaranteed in finite samples \citep{foygel2021limits}. Marginal coverage is
therefore the main target in conformal prediction, with conditional coverage pursued only approximately: well-designed
marginally calibrated methods such as conformalized quantile regression (CQR; \citealp{romano2019cqr}) often achieve high
conditional coverage in practice and can be shown to attain it asymptotically, under suitable assumptions including model consistency. 
This paper studies how to preserve the desirable properties of CQR---finite-sample marginal coverage and asymptotic conditional coverage---in the more
challenging clipped setting.

\subsection{Expanded Conformalized Quantile Regression} \label{sec:eCQR}
Existing conformal methods such as CQR (reviewed in Appendix~\ref{sec:cqr-review}) 
construct prediction intervals for fully observable outcomes, like our clipped $\tilde{Y}_{n+1}$. We show
that the natural extension of CQR to our setting over-covers $Y_{n+1}$ marginally, even though (as explained later) it
tends to under-cover it conditionally for \emph{easier} cases. Recall that the CQR
interval for $\tilde{Y}_{n+1}$ takes the form
\begin{align}\label{eq:cqr-set}
  \tilde{\mathcal{C}}(X_{n+1}) = \left[ \hat{f}^{\mathrm{lo}}(X_{n+1}) - \tilde\tau(\alpha),
  \hat{f}^{\mathrm{up}}(X_{n+1}) + \tilde\tau(\alpha) \right],
\end{align}
with the calibrated threshold $\tilde\tau(\alpha)$ equal to the
$\lceil (n+1)(1-\alpha)\rceil$-th smallest value among the non-conformity scores
$\{\tilde{S}^{\mathrm{CQR}}_i\}_{i=1}^{n} \cup \{R\}$ given by 
\begin{align}\label{eq:cqr-score-closed}
    & \tilde{S}_i^{\mathrm{CQR}} = s^{\mathrm{CQR}}(X_i, \tilde{Y}_i)
    & s^{\mathrm{CQR}}(x,y) := \max\bigl\{\hat{f}^{\mathrm{lo}}(x) - y, \; y - \hat{f}^{\mathrm{up}}(x)\bigr\},
  && \forall i \in [n].
\end{align}
In~\eqref{eq:cqr-set}, we adopt the convention that an interval $[a, b]$ with $a > b$ corresponds to the empty set.

\begin{theorem}[\cite{romano2019cqr}]\label{thm:cqr-validity}
If $(X_i, \tilde{Y}_i)_{i=1}^{n+1}$ are exchangeable,  
  $\mathbb{P}[\tilde{Y}_{n+1} \in \tilde{\mathcal{C}}(X_{n+1})] \geq 1-\alpha$ and, if $\{\tilde{S}^{\mathrm{CQR}}\}_{i=1}^{n+1}$ are a.s.\ distinct, $\mathbb{P}[\tilde{Y}_{n+1} \in \tilde{\mathcal{C}}(X_{n+1})] \leq 1-\alpha + 1/(n+1)$.
\end{theorem} 

This $\mathcal{O}(1/n)$-tight coverage upper bound hinges on the score distinctness assumption, which
in practice is usually ensured by jittering the base predictions $\hat{f}$; see Appendix~\ref{sec:cqr-review}.

The interval $\tilde{\mathcal{C}}(X_{n+1})$ for $\tilde{Y}_{n+1}$ can be naturally extended to cover
the latent outcome $Y_{n+1}$. Since observing $\tilde Y = C^{\mathrm{L}}$ reveals only $Y \in [Y^{\min},
C^{\mathrm{L}}]$, and $\tilde Y = C^{\mathrm{R}}$ only $Y \in [C^{\mathrm{R}}, Y^{\max}]$, with no
information on the position of $Y$ outside $[C^{\mathrm{L}}, C^{\mathrm{R}}]$, an endpoint should be sent all the way to the support boundary once it crosses a threshold. Formally, define the
\emph{snapping function} $\psi : \mathbb{R}^2 \to 2^{\mathbb{R}}$ element-wise such that, $\psi(\ell,u) = [\underline{\psi}(\ell), \overline{\psi}(u)]$ for any $\ell \leq u$, with
\begin{equation}\label{eq:psi-def}
    \underline{\psi}(\ell) = \begin{cases} Y^{\min}, & \ell \leq C^{\mathrm{L}}, \\
        \min\{\ell, C^{\mathrm{R}}\}, & \ell > C^{\mathrm{L}}, \end{cases}
    \qquad
    \overline{\psi}(u) = \begin{cases} \max\{u, C^{\mathrm{L}}\}, & u < C^{\mathrm{R}}, \\
        Y^{\max}, & u \geq C^{\mathrm{R}}. \end{cases}
    \end{equation}
  %   \psi(\ell, u)
%     = \begin{cases}
%         [Y^{\min}, Y^{\max}], & \text{if } \ell \leq C^{\mathrm{L}},\ u \geq C^{\mathrm{R}}, \\
%         [Y^{\min}, \max\{C^{\mathrm{L}}, u\}], & \text{if } \ell \leq C^{\mathrm{L}},\ \ell \leq u < C^{\mathrm{R}}, \\
%         [\ell, u], & \text{if } C^{\mathrm{L}} < \ell \leq u < C^{\mathrm{R}}, \\
%         [\min\{C^{\mathrm{R}}, \ell\}, Y^{\max}], & \text{if } C^{\mathrm{L}} < \ell \leq u,\ u \geq C^{\mathrm{R}}, \\
%         \emptyset, & \text{if } \ell > u.
%     \end{cases}
Snapping the CQR interval with $\psi$ yields what we call \emph{expanded CQR} (eCQR):
$\widehat{\mathcal{C}}^{\mathrm{eCQR}}(X_{n+1}) = \psi \bigl( \tilde{\mathcal{C}}(X_{n+1}) \bigr)$,
summarized with optional base jittering by Algorithm~\ref{alg:ecqr} in Appendix~\ref{sec:appendix-algorithms}.

While it is unsurprising that eCQR provides valid marginal coverage for $Y_{n+1}$, it is interesting to
quantify the margin by which it does so conservatively. All proofs are in Appendix~\ref{sec:appendix-proofs}.
\begin{theorem}\label{thm:ecqr-validity}
If $(X_i, Y_i)_{i=1}^{n+1}$ are exchangeable,  
$\mathbb{P}[Y_{n+1} \in \widehat{\mathcal{C}}^{\mathrm{eCQR}}(X_{n+1})] = p_{\mathrm{slack}} + \mathbb{P}[\tilde{Y}_{n+1} \in \tilde{\mathcal{C}}(X_{n+1}) ] \geq 1-\alpha + p_{\mathrm{slack}}$, with
$ p_{\mathrm{slack}}  =
  \mathbb{P}[ \tilde{\mathcal{C}}(X) < C^{\mathrm{L}},\; Y \leq C^{\mathrm{L}}] + \mathbb{P}[ \tilde{\mathcal{C}}(X) > C^{\mathrm{R}},\; Y \geq C^{\mathrm{R}}]$.
\end{theorem}
Thus, by Theorems~\ref{thm:cqr-validity} and~\ref{thm:ecqr-validity}, the coverage slack of eCQR is controlled by the probability that the CQR interval is non-empty and lies entirely beyond the observable range $[C^{\mathrm{L}}, C^{\mathrm{R}}]$ on the side where the outcome is censored. 
This suggests that excess marginal coverage may be reduced by avoiding extrapolation, projecting the base predictor $\hat{f}$ onto $[C^{\mathrm{L}}, C^{\mathrm{R}}]$ before applying eCQR. Algorithm~\ref{alg:pecqr} summarizes this approach, which we call \emph{projected-expanded CQR} (peCQR). Intuitive as it is, however, even peCQR is marginally over-conservative: while the projection removes one source of slack, it introduces another by creating ties in the calibration scores which cannot be broken by jittering without undoing the projection itself. We refer to Appendix~\ref{sec:app-ecqr} for further details on peCQR and its finite-sample marginal coverage.

\subsection{Tightening Marginal Coverage: ClipCQR} \label{sec:method-clipcqr}

We introduce a non-conformity score function designed for clipped outcomes.
This enables tight marginal coverage for the latent outcome, clarifies why eCQR is marginally conservative, and serves as the stepping
stone for the conditional calibration method of the next section.

In general, conformal prediction calibrates a nested family of intervals indexed by a scalar parameter \citep{gupta2022nested}.
For CQR, this family is $\phi^{\mathrm{CQR}}(x; \tau) := [\,\hat{f}^{\mathrm{lo}}(x) - \tau,\;\; \hat{f}^{\mathrm{up}}(x) + \tau]$, with $\tau \in \mathbb{R}$ (see Appendix~\ref{sec:cqr-review}).
The natural analogue of the CQR family for the clipped setting is
\begin{equation}\label{eq:phi-def}
    \phi(x; \tau) := \psi\bigl( \phi^{\mathrm{CQR}}(x; \tau) \bigr),
    \qquad \tau \in \mathbb{R}.
\end{equation}
with $\psi$ from~\eqref{eq:psi-def}.
This is nested by construction, $\phi(x; \tau_1) \subseteq \phi(x; \tau_2)$ whenever $\tau_1 \leq \tau_2$,
and can cover the entire support,
$\phi(x; \tau) \to [Y^{\min}, Y^{\max}]$ as $\tau \to R$ and $\phi(x; \tau) \to \emptyset$ as $\tau \to -\infty$.
The eCQR prediction interval belongs to~\eqref{eq:phi-def} since $\widehat{\mathcal{C}}^{\mathrm{eCQR}}(X_{n+1}) = \phi(X_{n+1}; \hat\tau^{\mathrm{eCQR}}(\alpha))$.

It is now clear that eCQR is marginally conservative because it takes the clipped outcomes at face value while calibrating $\hat\tau^{\mathrm{eCQR}}(\alpha)$ as if $\widehat{\mathcal{C}}^{\mathrm{eCQR}}(X_{n+1})$ belonged to the $\phi^{\mathrm{CQR}}$ family.
To achieve tight marginal coverage, we can calibrate~\eqref{eq:phi-def} directly, using the following score function.

\begin{lemma} \label{lemma:score-formula}
The score function $s(x, y) := \inf\bigl\{\tau \in \mathbb{R} : y \in \phi(x; \tau)\bigr\}$ corresponding to~\eqref{eq:phi-def} is:
\begin{equation}\label{eq:score-formula}
    s(x, y) = \begin{cases}
        \max\bigl\{\hat{f}^{\mathrm{lo}}(x) - C^{\mathrm{L}},\, \tfrac{1}{2}(\hat{f}^{\mathrm{lo}}(x) - \hat{f}^{\mathrm{up}}(x))\bigr\}, & \text{if } y \leq C^{\mathrm{L}}, \\
        \max\bigl\{\hat{f}^{\mathrm{lo}}(x) - y,\, y - \hat{f}^{\mathrm{up}}(x)\bigr\}, & \text{if } C^{\mathrm{L}} < y < C^{\mathrm{R}}, \\
        \max\bigl\{C^{\mathrm{R}} - \hat{f}^{\mathrm{up}}(x),\, \tfrac{1}{2}(\hat{f}^{\mathrm{lo}}(x) - \hat{f}^{\mathrm{up}}(x))\bigr\}, & \text{if } y \geq C^{\mathrm{R}}.
    \end{cases}
\end{equation}
\end{lemma}

This is closely related to the CQR score~\eqref{eq:cqr-score-closed}: for uncensored observations the two coincide, as they do for censored observations unless the base interval overlaps with the unobservable region on the censored side. In that case~\eqref{eq:score-formula} is the more negative of the two; the contrast is sharpest when the base interval lies entirely inside the unobservable region, where~\eqref{eq:score-formula} turns negative while the CQR score, taking the clipped outcome at face value, stays positive.

After evaluating this score~\eqref{eq:score-formula} on the clipped data, computing
$S_i = s(X_i, \tilde{Y}_i)$ for $i \in [n]$, the threshold
$\hat{\tau}(\alpha)$ is calibrated as the $\lceil(n+1)(1-\alpha)\rceil$-th smallest of
$\{S_i\}_{i=1}^n \cup \{R\}$, and the prediction interval is $\widehat{\mathcal{C}}(X_{n+1}) = \phi(X_{n+1}; \hat{\tau}(\alpha))$;
we call this method \emph{ClipCQR}.
Unlike peCQR, it accommodates tie-breaking among nonconformity scores
with the usual base jittering strategy (Appendix~\ref{sec:cqr-review}).
The ClipCQR method including jittering is summarized by Algorithm~\ref{alg:clipcqr}.

ClipCQR always tightens eCQR; see Proposition~\ref{prop:ecqr-ccqr} in Appendix~\ref{app:theory}.
Moreover, the next result shows the marginal coverage of ClipCQR is tight, within $\mathcal{O}(1/n)$ of the nominal level.

\begin{theorem}\label{thm:ccqr-validity}
If $(X_i, Y_i)_{i=1}^{n+1}$ are exchangeable, the ClipCQR prediction interval satisfies
\begin{align*}
    1 - \alpha \leq \mathbb{P}\bigl[Y_{n+1} \in \widehat{\mathcal{C}}(X_{n+1})\bigr] \leq 1 - \alpha + \tfrac{1}{n+1},
\end{align*}
with the second inequality holding if ClipCQR is applied with jittering.
\end{theorem}

Marginal coverage is usually the primary goal in conformal prediction, largely because it ordinarily does not conflict with approximate conditional coverage \citep{sesia2020comparison,romano2020classification}. Under clipping, however, it does. As discussed in Section~\ref{sec:theory}, both ClipCQR and eCQR tend to over-contract the prediction intervals in the \emph{interior} region of the feature space, where the model does not extrapolate. This motivates the alternative method proposed next.

\subsection{Prioritising Conditional Coverage: ClipCQR+}

We propose ClipCQR+, which prioritizes conditional coverage by calibrating two thresholds: one for points whose base interval
stays within the observable range, and one for those whose base interval extrapolates beyond it.
Let $\widehat{\mathcal I} \subseteq [n]$ denote the (random) subset of \emph{internal} calibration points, for which the (jittered) base prediction interval lies strictly inside the observable range:
\begin{align} \label{eq:I-hat}
  & \widehat{\mathcal I} := \bigl\{i \in [n] : X_i \in \mathrm{int}(\hat{f}) \bigr\},
  & \mathrm{int}(\hat{f}) := \left\{ x \in \mathcal{X} : C^{\mathrm L} < \hat f^{\mathrm{lo}}(x) \le \hat f^{\mathrm{up}}(x) < C^{\mathrm R} \right\}.
\end{align}
Let $\hat\tau(\alpha; \widehat{\mathcal I})$ be the threshold obtained by applying ClipCQR (Algorithm~\ref{alg:clipcqr}) using only the
internal calibration points indexed by $\widehat{\mathcal I}$, with $\hat\tau(\alpha; \widehat{\mathcal I}) = R$ if $\widehat{\mathcal I} = \emptyset$.
Similarly, let $\hat\tau(\alpha; \widehat{\mathcal I}^{\mathrm{c}})$ be the corresponding threshold calibrated using only the
\textit{extrapolation} points indexed by $\widehat{\mathcal I}^{\mathrm{c}} = [n] \setminus \widehat{\mathcal I}$.
ClipCQR+ replaces the $\hat\tau(\alpha)$ threshold of ClipCQR with
\begin{align} \label{eq:-threshold-clipcqrt+}
  \hat\tau^+(\alpha, X_{n+1}) := \hat\tau(\alpha;\widehat{\mathcal I})\,\I\{X_{n+1}\in\mathrm{int}(\hat f)\} + \max\{0,\hat\tau(\alpha;\widehat{\mathcal I}^{\mathrm c})\}\,\I\{X_{n+1}\notin\mathrm{int}(\hat f)\},
\end{align}
and then outputs $\widehat{\mathcal{C}}^+(X_{n+1}) = \phi\bigl(X_{n+1};\, \hat{\tau}^+(\alpha, X_{n+1})\bigr)$.
See Algorithm~\ref{alg:clipcqr_plus} for a summary.

The ClipCQR+ prediction intervals guarantee valid (though possibly conservative) marginal coverage and, in addition, coverage conditional on $X_{n+1} \in \mathrm{int}(\hat{f})$. It also enjoys desirable asymptotic conditional properties (Section~\ref{sec:theory}) and performs relatively well in practice (Section~\ref{sec:experiments}).

\begin{theorem}\label{thm:ccqr_plus-validity}
If $(X_i, Y_i)_{i=1}^{n+1}$ are exchangeable,  the ClipCQR+ prediction interval satisfies
\begin{align*}
    & \mathbb{P}\bigl[Y_{n+1} \in \widehat{\mathcal{C}}^+(X_{n+1})\bigr] \geq 1 - \alpha,
  & \mathbb{P}\bigl[Y_{n+1} \in \widehat{\mathcal{C}}^+(X_{n+1}) \mid X_{n+1} \in \mathrm{int}(\hat{f}) \bigr] \geq 1 - \alpha.
\end{align*}
\end{theorem}

All four methods---eCQR, peCQR, ClipCQR, and ClipCQR+---coincide if the base prediction intervals never extrapolate (Corollary~\ref{cor:ecqr-pecqr-ccqr} in Appendix~\ref{app:theory}). In general, ClipCQR+ and eCQR belong to the same family~\eqref{eq:phi-def}, but the ordering of their thresholds is data-dependent, so neither prediction interval need contain the other. The relation between ClipCQR+ and peCQR is likewise data-dependent, though these two are more closely connected (see Section~\ref{sec:theory} and Appendix~\ref{app:theory}).

\section{Theoretical Analysis} \label{sec:theory}

We analyse the large-sample behaviour of ClipCQR+, showing that it approximates an ideal oracle
under a suitable consistency assumption for the base predictor. The analyses of ClipCQR, eCQR and peCQR are
summarised at the end of this section, with details deferred to Appendix~\ref{app:theory}.

Let $F_{Y \mid X}(y \mid x) = \mathbb{P}[Y \leq y \mid X=x]$ and
$q_\beta(x) := \inf\{y : F_{Y \mid X}(y \mid x) \ge \beta\}$ respectively denote the conditional 
cumulative distribution and quantile functions of the population. The
\emph{oracle} conditional quantiles $q^{\mathrm{lo}}(x) := q_{\alpha/2}(x)$ and
$q^{\mathrm{up}}(x) := q_{1-\alpha/2}(x)$ give the ideal $\alpha$-level equal-tailed prediction oracle
$\mathcal{O}_{\mathrm{uncens}}(x) := [q^{\mathrm{lo}}(x), q^{\mathrm{up}}(x)]$ for the uncensored
setting \citep{sesia2020comparison}. Under censoring, predictions outside the observable range
can never be resolved in practice, and so it is intuitive to snap this oracle to $[C^{\mathrm{L}},C^{\mathrm{R}}]$:
 writing $\phi^\star(x;\tau) := \psi(q^{\mathrm{lo}}(x) - \tau,\, q^{\mathrm{up}}(x) + \tau)$
for the resulting \emph{oracle family}, the \emph{snapped conditional oracle} is its member with $\tau = 0$,
\[
  \mathcal{O}_0(x) := \phi^\star(x; 0) = \psi\bigl(q^{\mathrm{lo}}(x), q^{\mathrm{up}}(x)\bigr).
\]
Intuitively, $\mathcal{O}_0$ is the ideal target for a conditionally valid censoring-aware method: it returns the boundaries of
$\mathcal{O}_{\mathrm{uncens}}$ where observable, and collapses them onto the
relevant threshold where they are not.
The following result characterizes its marginal and conditional coverage.

\begin{assumption}[Regularity]\label{ass:regularity}
(i) $q^{\mathrm{lo}}(X)$ and $q^{\mathrm{up}}(X)$ have continuous distributions without point masses;
(ii) $Y \mid X = x$ has a density that is continuous, strictly positive on $[Y^{\min}, Y^{\max}]$, and uniformly bounded by some $M < \infty$, for almost every $x$.
\end{assumption}

\begin{proposition}\label{prop:O0-coverage}
Let $\pi^{\mathrm L} := \mathbb{P}(q^{\mathrm{lo}}(X) \le C^{\mathrm L})$ and
$\pi^{\mathrm R} := \mathbb{P}(q^{\mathrm{up}}(X) \ge C^{\mathrm R})$. Under Assumption~\ref{ass:regularity},
$\mathcal{O}_0$ is conservative, both conditionally $\mathbb{P}(Y \in \mathcal{O}_0(X) \mid X)\ge 1-\alpha$, and marginally 
$\mathbb{P}(Y \in \mathcal{O}_0(X)) \ge 1-\alpha + \tfrac{\alpha}{2}(\pi^{\mathrm L} + \pi^{\mathrm R})$, with equality 
if and only if $[q^{\mathrm{lo}}(X),q^{\mathrm{up}}(X)]$ a.s.\ overlaps with $[C^{\mathrm L}, C^{\mathrm R}]$.
% \[
%   \mathbb{P}(Y \in \mathcal{O}_0(X) \mid X)
%   = \begin{cases}
%     F_{Y\mid X}(C^{\mathrm L}\mid X) \geq 1-\alpha, & \text{if }q^{\mathrm{up}}(X)<C^{\mathrm L},\\
%     1-F_{Y\mid X}(C^{\mathrm R}\mid X) \geq 1-\alpha, & \text{if } q^{\mathrm{lo}}(X)>C^{\mathrm R}\\
%     1-\alpha + \tfrac{\alpha}{2}\bigl(\I\{q^{\mathrm{lo}}(X) \le C^{\mathrm L}\} + \I\{q^{\mathrm{up}}(X) \ge C^{\mathrm R}\}\bigr),
%       & \text{otherwise}.
%   \end{cases}
% \]
If $C^{\mathrm L} < q^{\mathrm{lo}}(X) \le q^{\mathrm{up}}(X) < C^{\mathrm R}$ a.s., then $\mathcal{O}_0 = \mathcal{O}_{\mathrm{uncens}}$.
\end{proposition}

We now show ClipCQR+ asymptotically recovers $\mathcal{O}_0$, if the base quantiles are
consistently estimated within the observable region.
To state this result rigorously, let $\{\hat f_n\}_{n \geq 1}$ denote a fixed sequence of base predictors,
indexed by the calibration sample size $n$ but trained independently of the calibration data.
All limits below are for $n \to \infty$.
For simplicity, we assume no jittering ($\delta=0$).
All probabilities are over the calibration data and an independent point $(X,Y)$. We write $\lambda$ for
Lebesgue measure and $\triangle$ for symmetric set difference.

\begin{definition}[Convergence of prediction bands]\label{def:band-convergence}
Let $\{\widehat{\mathcal{C}}_n\}_{n \geq 1}$ be data-dependent prediction bands and
$\mathcal{C}_\infty$ a fixed band, each mapping $\mathcal{X}$ to a subinterval of $[Y^{\min}, Y^{\max}]$. We write
$\widehat{\mathcal{C}}_n \to \mathcal{C}_\infty$ if
$\lambda(\widehat{\mathcal{C}}_n(X) \triangle \mathcal{C}_\infty(X)) \xrightarrow{p} 0$, 
with probability over the data and an independent $X \sim P_X$.
\end{definition}

\begin{definition}\label{def:cond-cov}
$\{\widehat{\mathcal{C}}_n\}_{n \geq 1}$ has asymptotic conditional coverage matching $\mathcal{C}_\infty$ if there
exist $\Lambda_n \subseteq \mathcal{X}$ with $\mathbb{P}(X \in \Lambda_n) \to 1$ and
$\sup_{x \in \Lambda_n} | \mathbb{P}(Y \in \widehat{\mathcal{C}}_n(x) \mid X = x)
  - \mathbb{P}(Y \in \mathcal{C}_\infty(x) \mid X = x) | \xrightarrow{p} 0$.
\end{definition}

\begin{assumption}[Consistency]\label{ass:consistency}
$\mathbb E \bigl( \bigl|\Pi\hat f^{\mathrm{lo}}_n(X)-\Pi q^{\mathrm{lo}}(X)\bigr| \bigr) \to 0$, $\mathbb E \bigl( \bigl|\Pi\hat f^{\mathrm{up}}_n(X)-\Pi q^{\mathrm{up}}(X)\bigr| \bigr) \to0$.
\end{assumption}

\begin{assumption}[Base regularity]\label{ass:base-regularity}
The base endpoints do not concentrate at the thresholds:
\[
  \lim_{\eta\downarrow0}\ \limsup_{n\to\infty}\ \mathbb P\bigl(|\hat f^{\mathrm{lo}}_n(X) - C^{\mathrm L}| \leq \eta\bigr)=0,
  \qquad
  \lim_{\eta\downarrow0}\ \limsup_{n\to\infty}\ \mathbb P\bigl(|\hat f^{\mathrm{up}}_n(X)-C^{\mathrm R}| \leq \eta \bigr)=0.
\]
\end{assumption}

Assumption~\ref{ass:consistency} asks only that the base predictor estimate the
conditional quantiles accurately within the observable range.
Assumption~\ref{ass:base-regularity} holds whenever the distributions of $\hat f^{\mathrm{lo}}_n(X)$ and $\hat f^{\mathrm{up}}_n(X)$ have, on neighborhoods of $C^{\mathrm L}$ and $C^{\mathrm R}$, a density bounded uniformly in $n$. This rules out models trained to predict the clipped outcomes, which may have  point masses at $C^{\mathrm L}$ and $C^{\mathrm R}$, and is more naturally satisfied when the base predictor targets the \emph{latent} conditional quantiles \citep{tobin1958estimation}.

\begin{assumption}[Non-extreme censoring]\label{ass:non-extreme-cens}
$\mathbb P(C^{\mathrm{L}} < q^{\mathrm{lo}}(X) \le q^{\mathrm{up}}(X) < C^{\mathrm{R}}) > 0$.
\end{assumption}

\begin{theorem}\label{thm:ccqr-plus-oracle}
Assume $(X_i, Y_i)_{i=1}^{n+1}$ are i.i.d. Under Assumptions~\ref{ass:regularity}--\ref{ass:non-extreme-cens}, letting $\widehat{\mathcal{C}}^{+}_n$ denote the ClipCQR+ prediction interval:
(i) $\widehat{\mathcal{C}}^{+}_n \to \mathcal{O}_0$; (ii) $\mathbb{P}(Y \in \widehat{\mathcal{C}}^{+}_n(X)) \to
\mathbb{P}(Y \in \mathcal{O}_0)$; and (iii) $\widehat{\mathcal{C}}^{+}_n$ has asymptotic conditional coverage matching
$\mathcal{O}_0$ (Definition~\ref{def:cond-cov}).
\end{theorem}

The other methods are studied in Appendix~\ref{app:theory}. Under similar assumptions, ClipCQR and eCQR approach distinct oracles (resp.~$\mathcal{O}_{\mathrm{m}}$ and $\mathcal{O}_{\mathrm{e}}$) that lack conditional coverage (Theorems~\ref{thm:ccqr-oracle} and~\ref{thm:ecqr-oracle}).
By contrast, peCQR converges to $\mathcal{O}_0$ similar to ClipCQR+ (Theorem~\ref{thm:pecqr-oracle}), although it does not enjoy a similar finite-sample conditional guarantee.
Table~\ref{tab:method-summary} summarizes these results.

\begin{table}[!htb]
\centering
\small
\setlength{\tabcolsep}{5pt}
\renewcommand{\arraystretch}{1.25}
\caption{Theoretical comparison of conformal methods to predict a latent regression outcome $Y$ using clipped calibration data.
\cmark\ valid, \cmarks\ valid in a weaker sense, \cmarku\ conservative.}
\label{tab:method-summary}
\begin{tabular}{@{}l ccc l ccc@{}}
\toprule
& \multicolumn{3}{c}{Finite-sample} & \multicolumn{4}{c}{Asymptotic limit (under base consistency)} \\
\cmidrule(lr){2-4}\cmidrule(l){5-8}
Method & Marg.\ cov. & Cond.\ cov. & Ref. & Oracle & Marg.\ cov. & Cond.\ cov. & Ref. \\
\midrule
%CQR     & \xmark      & \xmark  & Theorem~\ref{thm:cqr-validity}       & unsnapped                                                    & $<1{-}\alpha$    & \xmark  & Theorem~\ref{thm:cqr-oracle} \\
eCQR     & \cmarku & \xmark  & Theorem~\ref{thm:ecqr-validity}      & $\mathcal{O}_{\mathrm{e}} \supseteq \mathcal{O}_{\mathrm{m}}$ & \cmarku  & \xmark  & Theorem~\ref{thm:ecqr-oracle} \\
peCQR    & \cmarku & \xmark  & Theorem~\ref{thm:pecqr-validity}    & $\mathcal{O}_0$                                              & \cmarku  & \cmark  & Theorem~\ref{thm:pecqr-oracle} \\
ClipCQR  & \cmark & \xmark & Theorem~\ref{thm:ccqr-validity} & $\mathcal{O}_{\mathrm{m}}\subseteq \mathcal{O}_0$            & \cmark     & \xmark  & Theorem~\ref{thm:ccqr-oracle} \\
ClipCQR+ & \cmarku & \cmarks & Theorem~\ref{thm:ccqr_plus-validity} & $\mathcal{O}_0$                                              & \cmarku  & \cmark  & Theorem~\ref{thm:ccqr-plus-oracle} \\
\bottomrule
\end{tabular}
\end{table}

\section{Numerical Experiments with Synthetic Data} \label{sec:experiments}

We simulate heteroskedastic synthetic data $X \sim \mathrm{Unif}[0,3]^p$ and
 $Y = m(X) + \sigma(X)\varepsilon$, with
$m(x) = 0.5\,x_1\sin(2\pi x_1)$, $\sigma(x) = 0.1\,x_1^2\mathbf{1}\{x_1>0\}$, and $\varepsilon \sim N(0,1)$.
The training and calibration data contain clipped outcomes $\tilde{Y} = \Pi(Y)$, with known clipping thresholds $[C^{\mathrm{L}},C^{\mathrm{R}}]$. The base predictor is a gradient-boosting heteroscedastic regression model, implemented by the \texttt{R} package \texttt{xgboost}, and trained using a Tobit-like loss function that allows it to estimate the data generating distribution well. 
Results with different base models are in Appendix~\ref{app:empirical}.

We apply ClipCQR+ at level $\alpha = 0.1$ and compare it to ClipCQR and eCQR, as well as to the base prediction intervals extracted
from the conditional quantiles estimated by the model, both with and without snapping. 
We also compare to the \emph{LdPT} conformal method of \citet{liu2025prediction}, discussed in Appendix~\ref{sec:app-relation-double}. Results with peCQR are reported in Appendix~\ref{app:empirical}, to avoid cluttering.
All methods are applied assuming $[Y^{\min}, Y^{\max}] = [-4, 4]$.
As an ideal reference we
include the snapped oracle $\mathcal{O}_0$, computed from the true
data-generating distribution. We report four metrics for each approach:
marginal coverage of $Y$, estimated \emph{worst-slab} conditional coverage of $Y$ (ref.~Appendix~\ref{app:worstcase}), the
proportion of closed intervals (not crossing a censoring endpoint), and the
average \textit{width} of the prediction intervals projected on the observable range $[C^{\mathrm{L}},C^{\mathrm{R}}]$,
so that the maximum width is one.
Results are averaged over $100$ independent experiments.

Figure~\ref{fig:fig1_synthetic_ncal} varies the training and calibration sample size
$n$ (kept equal), fixing $C^{\mathrm{R}} = 0.5 = -C^{\mathrm{L}}$ and the
number of features $p = 1$. Qualitatively similar results for different censoring thresholds, feature dimensions, and base models
are presented in Appendix~\ref{app:empirical}.
Overall, the results match the theory.
All conformal methods attain marginal coverage at or
above the nominal level. ClipCQR+ usually achieves the highest conditional coverage,
and as $n$ grows it behaves like the ideal oracle $\mathcal{O}_0$. 
By contrast, ClipCQR and eCQR tend to have lower conditional coverage, 
not only when $n$ is small but also in very large samples where the base model estimates the population distribution well.
LdPT is very conservative and produces uninformative prediction intervals for all $n$.

\begin{figure}[!htb]
  \centering
  \includegraphics[width=\linewidth]{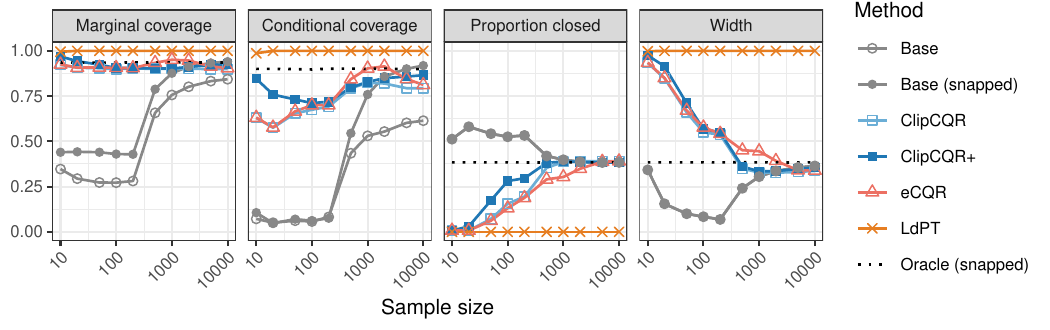} \vspace{-2em}
  \caption{Performance of conformal prediction intervals on doubly censored synthetic data with univariate features, as a function of the training and calibration sample size.
The censoring thresholds are $C^{\mathrm{R}}=0.5=-C^{\mathrm{L}}$.
All methods use the same gradient-boosting base model.
The ClipCQR+ intervals approach the oracle ones as the sample size grows.}
  \label{fig:fig1_synthetic_ncal}
\end{figure}

\section{Discussion} \label{sec:discussion}
A key takeaway is that marginal coverage is not always the best target in conformal prediction: under clipping it conflicts with conditional coverage. We study this trade-off and offer methods at both ends: ClipCQR prioritizes marginal coverage and ClipCQR+ conditional coverage.
A natural extension of this work would study how to handle heterogeneous censoring thresholds,
potentially drawing inspiration from \citet{liu2025prediction}.
A simple but likely suboptimal option is to clip all data to the most restrictive thresholds; alternatively, one
could split the sample non-randomly, similar to \citet{xie2025}, assigning observations with atypical thresholds to the training set to keep the calibration set homogeneous. If the thresholds depend on the features, however, this induces a distribution shift that must be corrected via weighting \citep{tibshirani2019conformal}.

\section*{Supplementary Material}

The Supplementary Material contains proofs, additional methodological details,
further simulation results, and an application to quantitative structure--activity relationship (QSAR) data.

%%% Local Variables:
%%% mode: latex
%%% TeX-master: "main_biometrika"
%%% End:

\section*{Declaration of the use of generative AI and AI-assisted technologies}
Claude Opus 4.8 edited language and drafted some proofs. The authors take full responsibility.

\bibliographystyle{plainnat}
\bibliography{bibliography}

\clearpage
\appendix

% Special numbering for appendix (Use "S" instead of "A" in Supplement)
\renewcommand{\thesection}{A\arabic{section}}
\renewcommand{\theequation}{A\arabic{equation}}
\renewcommand{\thetheorem}{A\arabic{theorem}}
\renewcommand{\theassumption}{A\arabic{assumption}}
\renewcommand{\thecorollary}{A\arabic{corollary}}
\renewcommand{\theproposition}{A\arabic{proposition}}
\renewcommand{\thelemma}{A\arabic{lemma}}
\renewcommand{\thetable}{A\arabic{table}}
\renewcommand{\thefigure}{A\arabic{figure}}
\renewcommand{\thealgocf}{A\arabic{algocf}}

\def\theHalgorithm{A\arabic{algorithm}}

\section{Review of Conformalized Quantile Regression (CQR)}\label{sec:cqr-review}

We briefly review split conformalized quantile regression (CQR;
\citealp{romano2019cqr}) in the standard setting where the goal is to predict 
a fully observed outcome $\tilde Y$. We present CQR
in the \emph{nested-set} formulation \citep{gupta2022nested}.

Starting from a pre-trained base prediction interval
$\hat{f}(x) = [\hat{f}^{\mathrm{lo}}(x), \hat{f}^{\mathrm{up}}(x)]$, CQR adjusts its
width by a single scalar $\tau \in \mathbb{R}$, inflating it when $\tau > 0$ and
contracting it when $\tau < 0$. At any feature value $x \in \mathcal{X}$, the CQR
family of prediction intervals indexed by $\tau \in \mathbb{R}$ is
\begin{equation}\label{eq:cqr-family}
    \phi^{\mathrm{CQR}}(x; \tau)
    \;:=\; \bigl[\,\hat{f}^{\mathrm{lo}}(x) - \tau,\;\; \hat{f}^{\mathrm{up}}(x) + \tau\,\bigr],
\end{equation}
with the convention that an interval $[a, b]$ with $a > b$
corresponds to the empty set.

These intervals are nested:
$\phi^{\mathrm{CQR}}(x; \tau_1) \subseteq \phi^{\mathrm{CQR}}(x; \tau_2)$ whenever
$\tau_1 \leq \tau_2$. At $\tau = 0$ the family returns the base interval,
$\phi^{\mathrm{CQR}}(x; 0) = [\hat{f}^{\mathrm{lo}}(x), \hat{f}^{\mathrm{up}}(x)]$;
as $\tau \to R$ it grows to cover all of $\mathbb{R}$; and as $\tau$
decreases it narrows, collapsing to the singleton midpoint
$\hat{m}(x) := (\hat{f}^{\mathrm{lo}}(x) + \hat{f}^{\mathrm{up}}(x))/2$ at
$\tau = \tfrac12(\hat{f}^{\mathrm{lo}}(x) - \hat{f}^{\mathrm{up}}(x))$ and becoming
empty for smaller $\tau$.

The nonconformity score of a pair $(x, \tilde y)$ is the smallest $\tau$ at which the
family first captures $y$:
\begin{equation}\label{eq:cqr-score}
    s^{\mathrm{CQR}}(x, \tilde y) \;:=\; \inf\bigl\{\tau \in \mathbb{R} :
                              \tilde y \in \phi^{\mathrm{CQR}}(x; \tau)\bigr\}.
\end{equation}
Solving $\tilde y \in [\hat{f}^{\mathrm{lo}}(x) - \tau,\, \hat{f}^{\mathrm{up}}(x) + \tau]$
for the smallest admissible $\tau$ gives the familiar closed form~\eqref{eq:cqr-score-closed},
which is the signed distance of $\tilde y$ from the base interval, negative for points inside.
This identity holds at every $\tilde y$, including the midpoint $\hat{m}(x)$, where it
returns $\tfrac12(\hat{f}^{\mathrm{lo}}(x) - \hat{f}^{\mathrm{up}}(x))$.

Using a calibration data set $\{(X_i, \tilde Y_i)\}_{i=1}^n$ independent of the training
data, CQR evaluates the scores $s^{\mathrm{CQR}}(X_i, \tilde Y_i)$ for all $i \in [n]$ and
computes $\tilde\tau(\alpha)$ as the
$\lceil (n+1)(1-\alpha)\rceil$-th smallest value in the augmented multi-set
$\{s^{\mathrm{CQR}}(X_i, \tilde Y_i)\}_{i=1}^{n} \cup \{R\}$. The CQR prediction interval
at $X_{n+1}$ is
\begin{equation*}
    \tilde{\mathcal{C}}(X_{n+1})
    \;=\; \phi^{\mathrm{CQR}}\bigl(X_{n+1};\, \tilde\tau(\alpha)\bigr)
    \;=\; \bigl[\,\hat{f}^{\mathrm{lo}}(X_{n+1}) - \tilde\tau(\alpha),\;
                  \hat{f}^{\mathrm{up}}(X_{n+1}) + \tilde\tau(\alpha)\,\bigr].
\end{equation*}
See Algorithm~\ref{alg:cqr} in Appendix~\ref{sec:appendix-algorithms} for a
summary of this procedure. By nestedness of $\phi^{\mathrm{CQR}}(x; \cdot)$ and the
definition of $s^{\mathrm{CQR}}$ as an infimum,
$\tilde Y_{n+1} \in \phi^{\mathrm{CQR}}(X_{n+1}; \tau)$ if and only if
$s^{\mathrm{CQR}}(X_{n+1}, \tilde Y_{n+1}) \leq \tau$. This leads to the elegant marginal
coverage guarantee of Theorem~\ref{thm:cqr-validity}.

The assumption that the CQR scores $\{s^{\mathrm{CQR}}(X_i, \tilde Y_i)\}_{i=1}^{n+1}$ are almost-surely
distinct typically holds if $\tilde Y$ has a continuous distribution without
point masses.
In general, it can be ensured by adding a small amount of independent noise to the base
predictions. Draw $\eta_1, \ldots, \eta_{n+1}$ i.i.d.~from $\mathrm{Uniform}(0, \delta)$ for a small
$\delta > 0$, independent of the data, and perturb the base endpoints as:
\begin{equation*}%\label{eq:cqr-tiebreak}
    \hat{f}^{\mathrm{lo}}(X_i) \;\mapsto\; \hat{f}^{\mathrm{lo}}(X_i) - \eta_i,
    \qquad
    \hat{f}^{\mathrm{up}}(X_i) \;\mapsto\; \hat{f}^{\mathrm{up}}(X_i) + \eta_i,
    \qquad i \in [n+1].
\end{equation*}
This perturbation induces scores $S'_i
    \;=\; \max\{\hat{f}^{\mathrm{lo}}(X_i) - \tilde Y_i,\;
               \tilde Y_i - \hat{f}^{\mathrm{up}}(X_i)\} - \eta_i
    \;=\; S_i - \eta_i$,
which are deterministic functions of the augmented data tuples $(X_i, \tilde Y_i, \eta_i)$, still exchangeable, and almost-surely
distinct. Thus, applying CQR with $\tilde\tau(\alpha)$ computed from $\{S'_i\}_{i=1}^n$ instead of $\{S_i\}_{i=1}^n$
leads to a tight marginal coverage guarantee within $[1-\alpha, 1-\alpha+1/(n+1)]$.

While Theorem~\ref{thm:cqr-validity} guarantees marginal coverage, CQR additionally
enjoys asymptotic conditional coverage and optimality properties under base consistency \citep{sesia2020comparison}. 
Writing $q_\beta(\cdot)$ for the population conditional $\beta$-quantile of
$\tilde Y \mid X$, suppose \emph{both} base quantile estimators are consistent, i.e.,
$\hat{f}^{\mathrm{lo}}$ and $\hat{f}^{\mathrm{up}}$ converge to $q_{\alpha/2}$ and
$q_{1-\alpha/2}$ in an appropriate sense, respectively. Then, as $n \to \infty$, the conformal adjustment vanishes
($\tilde\tau(\alpha) \to 0$) and $\tilde{\mathcal{C}}(X_{n+1})$ converges to the
oracle prediction interval $[q_{\alpha/2}(X_{n+1}),\, q_{1-\alpha/2}(X_{n+1})]$,
which is the shortest possible interval achieving $(1-\alpha)$ coverage conditional on $X_{n+1}$,
with equal miscoverage in each tail.

These nice finite-sample and asymptotic properties of CQR, combined with its intuitive nature,
explain its popularity in practice. This makes it a natural baseline approach to also consider under
censoring, when $\tilde Y$ may not be observed outside the visible range $[C^{\mathrm{L}},C^{\mathrm{R}}]$.

\section{Algorithms}\label{sec:appendix-algorithms}

\begin{algorithm}[!htb]
  \caption{Conformalized Quantile Regression (CQR)}
  \label{alg:cqr}
  \SetKwInOut{Input}{Input}
  \SetKwInOut{Output}{Output}
  \Input{%
    Calibration data $\{(X_i, \tilde Y_i)\}_{i=1}^{n}$, fully observed and independent of training data; \\
    Pre-trained base interval predictor $\hat{f} = [\hat{f}^{\mathrm{lo}}, \hat{f}^{\mathrm{up}}]$; \\
    Unlabeled test point with features $X_{n+1}$; \\
    Desired miscoverage probability $\alpha \in (0,1)$; \\
    Jitter scale $\delta \geq 0$ (set $\delta = 0$ to disable).
  }
  Draw $\eta_1, \dots, \eta_{n+1} \stackrel{\mathrm{i.i.d.}}{\sim} \mathrm{Uniform}(0, \delta)$, independent of the data, and widen the base predictor into $\hat{f}_{\eta_i} := [\hat{f}^{\mathrm{lo}} - \eta_i,\; \hat{f}^{\mathrm{up}} + \eta_i] =: [\hat{f}^{\mathrm{lo}}_{\eta_i}, \hat{f}^{\mathrm{up}}_{\eta_i}]$ for each $i \in [n+1]$\;
  Compute scores $S_i \gets \max\bigl\{\hat{f}^{\mathrm{lo}}_{\eta_i}(X_i) - \tilde Y_i,\; \tilde Y_i - \hat{f}^{\mathrm{up}}_{\eta_i}(X_i)\bigr\}$ for each $i \in [n]$\;
  Compute threshold $\tilde\tau(\alpha)$ as the $\lceil (n+1)(1-\alpha) \rceil$-th smallest value in $\{S_i\}_{i=1}^{n} \cup \{R\}$\;
  Evaluate the interval $\tilde{\mathcal{C}}(X_{n+1}) \gets \phi^{\mathrm{CQR}}\bigl(X_{n+1};\, \tilde\tau(\alpha)\bigr)$ from the jittered base $\hat{f}_{\eta_{n+1}}$, via~\eqref{eq:cqr-family}\;
  \Return{$\tilde{\mathcal{C}}(X_{n+1}) \subseteq \mathbb{R}$}.
\end{algorithm}

\begin{algorithm}[!htb]
  \caption{Expanded Conformalized Quantile Regression (eCQR)}
  \label{alg:ecqr}
  \SetKwInOut{Input}{Input}
  \SetKwInOut{Output}{Output}
  \Input{%
    Censored calibration data $\{(X_i, \tilde{Y}_i)\}_{i=1}^{n}$, independent of the training data; \\
    Pre-trained base interval predictor $\hat{f} = [\hat{f}^{\mathrm{lo}}, \hat{f}^{\mathrm{up}}]$; \\
    Censoring thresholds $C^{\mathrm{L}} < C^{\mathrm{R}}$ and support $[Y^{\min}, Y^{\max}]$; \\
    Unlabeled test point with features $X_{n+1}$; \\
    Desired miscoverage probability $\alpha \in (0,1)$; \\
    Jitter scale $\delta \geq 0$ (set $\delta = 0$ to disable).
  }
  Draw $\eta_1, \dots, \eta_{n+1} \stackrel{\mathrm{i.i.d.}}{\sim} \mathrm{Uniform}(0, \delta)$, independent of the data, and widen the base predictor into $\hat{f}_{\eta_i} := [\hat{f}^{\mathrm{lo}} - \eta_i,\; \hat{f}^{\mathrm{up}} + \eta_i] =: [\hat{f}^{\mathrm{lo}}_{\eta_i}, \hat{f}^{\mathrm{up}}_{\eta_i}]$ for each $i \in [n+1]$\;
  Compute scores $\tilde{S}^{\mathrm{CQR}}_i \gets \max\bigl\{\hat{f}^{\mathrm{lo}}(X_i) - \tilde{Y}_i,\; \tilde{Y}_i - \hat{f}^{\mathrm{up}}(X_i)\bigr\}$ for each $i \in [n]$\;
  Compute $\hat\tau^{\mathrm{eCQR}}(\alpha)$ as the $\lceil (n+1)(1-\alpha) \rceil$-th smallest value in $\{\tilde{S}^{\mathrm{CQR}}_i\}_{i=1}^{n} \cup \{R\}$\;
  Form $[\tilde{L}(X_{n+1}),\, \tilde{U}(X_{n+1})] \gets \phi^{\mathrm{CQR}}\bigl(X_{n+1};\, \hat\tau^{\mathrm{eCQR}}(\alpha)\bigr)$ via~\eqref{eq:cqr-family}\;
  Evaluate the interval $\widehat{\mathcal{C}}^{\mathrm{eCQR}}(X_{n+1}) \gets \psi\bigl(\tilde{L}(X_{n+1}),\; \tilde{U}(X_{n+1})\bigr)$ via~\eqref{eq:psi-def}\;
  \Return{$\widehat{\mathcal{C}}^{\mathrm{eCQR}}(X_{n+1}) \subseteq [Y^{\min}, Y^{\max}]$}.
\end{algorithm}

\begin{algorithm}[!htb]
  \caption{Projected-Expanded Conformalized Quantile Regression (peCQR)}
  \label{alg:pecqr}
  \SetKwInOut{Input}{Input}
  \SetKwInOut{Output}{Output}
  \Input{%
    Censored calibration data $\{(X_i, \tilde{Y}_i)\}_{i=1}^{n}$, independent of the training data; \\
    Pre-trained base interval predictor $\hat{f} = [\hat{f}^{\mathrm{lo}}, \hat{f}^{\mathrm{up}}]$; \\
    Censoring thresholds $C^{\mathrm{L}} < C^{\mathrm{R}}$ and support $[Y^{\min}, Y^{\max}]$; \\
    Unlabeled test point with features $X_{n+1}$; \\
    Desired miscoverage probability $\alpha \in (0,1)$; \\
    Jitter scale $\delta \geq 0$ (set $\delta = 0$ to disable).
  }
  Draw $\eta_1, \dots, \eta_{n+1} \stackrel{\mathrm{i.i.d.}}{\sim} \mathrm{Uniform}(0, \delta)$, independent of the data, and widen the base predictor into $\hat{f}_{\eta_i} := [\hat{f}^{\mathrm{lo}} - \eta_i,\; \hat{f}^{\mathrm{up}} + \eta_i] =: [\hat{f}^{\mathrm{lo}}_{\eta_i}, \hat{f}^{\mathrm{up}}_{\eta_i}]$ for each $i \in [n+1]$\;
  Project the jittered base onto $[C^{\mathrm{L}}, C^{\mathrm{R}}]$: set $\Pi\hat{f}^{\mathrm{lo}}(X_i) \gets \Pi\hat{f}^{\mathrm{lo}}_{\eta_i}(X_i)$ and $\Pi\hat{f}^{\mathrm{up}}(X_i) \gets \Pi\hat{f}^{\mathrm{up}}_{\eta_i}(X_i)$ for each $i \in [n+1]$\;
  Evaluate prediction interval $\widehat{\mathcal{C}}^{\mathrm{peCQR}}(X_{n+1}) \gets$ output of Algorithm~\ref{alg:ecqr} run with base predictor $\Pi\hat{f} = [\Pi\hat{f}^{\mathrm{lo}}, \Pi\hat{f}^{\mathrm{up}}]$, the same $C^{\mathrm{L}}, C^{\mathrm{R}}, [Y^{\min}, Y^{\max}], \alpha, X_{n+1}$, and jitter scale $0$\;
  \Return{$\widehat{\mathcal{C}}^{\mathrm{peCQR}}(X_{n+1}) \subseteq [Y^{\min}, Y^{\max}]$}.
\end{algorithm}

\begin{algorithm}[!htb]
  \caption{ClipCQR}
  \label{alg:clipcqr}
  \SetKwInOut{Input}{Input}
  \SetKwInOut{Output}{Output}
  \Input{%
    Censored calibration data  $\{(X_i, \tilde{Y}_i)\}_{i=1}^{n}$, independent of the training data; \\
    Pre-trained base interval predictor $\hat{f} = [\hat{f}^{\mathrm{lo}}, \hat{f}^{\mathrm{up}}]$; \\
    Censoring thresholds $C^{\mathrm{L}} < C^{\mathrm{R}}$ and support $[Y^{\min}, Y^{\max}]$; \\
    Unlabeled test point with features $X_{n+1}$; \\
    Desired miscoverage probability $\alpha \in (0,1)$; \\
    Jitter scale $\delta \geq 0$ (set $\delta = 0$ to disable).
  }
  Draw $\eta_1, \dots, \eta_{n+1} \stackrel{\mathrm{i.i.d.}}{\sim} \mathrm{Uniform}(0, \delta)$, independent of the data, and widen the base predictor into $\hat{f}_{\eta_i} := [\hat{f}^{\mathrm{lo}} - \eta_i,\; \hat{f}^{\mathrm{up}} + \eta_i]$ for each $i \in [n+1]$\;
  Compute scores $S_i \gets s(X_i, \tilde{Y}_i)$ from the jittered base $\hat{f}_{\eta_i}$, for each $i \in [n]$, via~\eqref{eq:score-formula}\;
  Compute threshold $\hat{\tau}(\alpha)$ as the $\lceil (n+1)(1-\alpha) \rceil$-th smallest value in $\{S_i\}_{i=1}^{n} \cup \{R\}$\;
  Evaluate interval $\widehat{\mathcal{C}}(X_{n+1}) \gets \phi(X_{n+1};\, \hat{\tau}(\alpha))$ from the jittered base $\hat{f}_{\eta_{n+1}}$, via~\eqref{eq:phi-def} and~\eqref{eq:psi-def}\;
  \Return{$\widehat{\mathcal{C}}(X_{n+1}) \subseteq [Y^{\min}, Y^{\max}]$}.
\end{algorithm}

\begin{algorithm}[!htb]
  \caption{ClipCQR+}
  \label{alg:clipcqr_plus}
  \SetKwInOut{Input}{Input}
  \SetKwInOut{Output}{Output}
  \Input{Same inputs as Algorithm~\ref{alg:clipcqr}.}
  Run Steps 1--2 of Algorithm~\ref{alg:clipcqr} to obtain the jittered base $\hat{f}_{\eta}$ and the scores\;
  Form $\widehat{\mathcal I} \subseteq [n]$ via~\eqref{eq:I-hat} and compute $\hat{\tau}(\alpha;\widehat{\mathcal I})$, the threshold of Step~3 of Algorithm~\ref{alg:clipcqr} applied on the calibration points $i \in \widehat{\mathcal I}$, with the same pre-jittered base $\hat{f}_{\eta}$ and $\delta = 0$\;
  Form $\widehat{\mathcal I}^{\mathrm{c}} = [n] \setminus \widehat{\mathcal I}$ and compute $\hat{\tau}(\alpha;\widehat{\mathcal I}^{\mathrm{c}})$ similarly\;
  Define $\hat\tau^+(\alpha, X_{n+1}) := \hat\tau(\alpha;\widehat{\mathcal I})\,\I\{X_{n+1}\in\mathrm{int}(\hat f)\} + \max\{0,\hat\tau(\alpha;\widehat{\mathcal I}^{\mathrm c})\}\,\I\{X_{n+1}\notin\mathrm{int}(\hat f)\}$\;
  Evaluate interval $\widehat{\mathcal{C}}^+(X_{n+1}) \gets \phi\bigl(X_{n+1};\, \hat{\tau}(\alpha, X_{n+1})\bigr)$ from $\hat{f}_{\eta_{n+1}}$, via~\eqref{eq:phi-def} and~\eqref{eq:psi-def}\;
  \Return{$\widehat{\mathcal{C}}^+(X_{n+1}) \subseteq [Y^{\min}, Y^{\max}]$}.
\end{algorithm}

\clearpage

\section{From Expanded CQR to Projected-Expanded CQR} \label{sec:app-ecqr}

Write the CQR prediction interval~\eqref{eq:cqr-set} as $\tilde{\mathcal{C}}(\cdot) = [\tilde{L}(\cdot), \tilde{U}(\cdot)]$, with
\begin{align*}
  & \tilde{L}(X_{n+1}) = \hat{f}^{\mathrm{lo}}(X_{n+1}) - \tilde\tau(\alpha),
  & \tilde{U}(X_{n+1}) = \hat{f}^{\mathrm{up}}(X_{n+1}) + \tilde\tau(\alpha).
\end{align*}

Recall that Theorem~\ref{thm:ecqr-validity} shows eCQR over-covers on average, with the slack
\begin{align*}
  p_{\mathrm{slack}}
  & :=
    \mathbb{P}\bigl(\tilde{L} \leq \tilde{U} < C^{\mathrm{L}},\; Y_{n+1} < C^{\mathrm{L}}\bigr) + \mathbb{P}\bigl(\tilde{U} \geq \tilde{L} > C^{\mathrm{R}},\; Y_{n+1} > C^{\mathrm{R}}\bigr)
\end{align*}
tending to be larger if censoring is more prevalent and the raw CQR intervals more often lie entirely outside the observable range.
The following corollary upper bounds $p_{\mathrm{slack}}$ by two more easily interpretable quantities.
\begin{corollary}\label{cor:pexp-loose-bound}
Let $p_{\mathrm{cens}} := \mathbb{P}(Y_{n+1} \notin [C^{\mathrm{L}}, C^{\mathrm{R}}])$ denote the
probability that the test point is censored, and let
$p_{\mathrm{ext}} := \mathbb{P}(\tilde{L} \leq \tilde{U} < C^{\mathrm{L}})
+ \mathbb{P}(\tilde{U} \geq \tilde{L} > C^{\mathrm{R}})$ denote the probability
that the raw CQR interval is non-degenerate and lies entirely outside $[C^{\mathrm{L}}, C^{\mathrm{R}}]$. Then $p_{\mathrm{slack}} \leq \min\bigl\{p_{\mathrm{cens}},\; p_{\mathrm{ext}}\bigr\}$,
and consequently, under the assumptions of Theorem~\ref{thm:ecqr-validity},
\[
    1-\alpha + p_{\mathrm{slack}} \leq \mathbb{P}\bigl[Y_{n+1} \in \widehat{\mathcal{C}}^{\mathrm{eCQR}}(X_{n+1})\bigr]
    \leq 1 - \alpha + \frac{1}{n+1} + \min\{p_{\mathrm{cens}},\; p_{\mathrm{ext}}\}.
\]
\end{corollary}
\noindent The proof is deferred to Appendix~\ref{sec:appendix-proofs}.

Of the two terms that bound the over-coverage in
Corollary~\ref{cor:pexp-loose-bound}, the censoring probability
$p_{\mathrm{cens}}$ is a property of the population,
whereas the extrapolation probability $p_{\mathrm{ext}}$
is a property of the procedure which can be controlled.
A natural attempt to remove the coverage slack is therefore to drive
$p_{\mathrm{ext}}$ to zero, by projecting the base quantiles onto $[C^{\mathrm{L}}, C^{\mathrm{R}}]$
before calibration using $\Pi$,
and then run eCQR with $\Pi \hat{f}^{\mathrm{lo}}$ and $\Pi \hat{f}^{\mathrm{up}}$ instead of $\hat{f}^{\mathrm{lo}}$ and $\hat{f}^{\mathrm{up}}$. This modified procedure, which we call Projected-Expanded CQR (peCQR), is summarized by
Algorithm~\ref{alg:pecqr} in Appendix~\ref{sec:appendix-algorithms}.
It is a special case of eCQR with $p_{\mathrm{slack}} = 0$.

\begin{lemma}\label{lem:pecqr-no-slack}
After projecting $\hat{f}$ onto $[C^{\mathrm{L}}, C^{\mathrm{R}}]$ via $\Pi$,  $p_{\mathrm{ext}} = 0$ and hence $p_{\mathrm{slack}} = 0$.
\end{lemma}
\noindent The proof is deferred to Appendix~\ref{sec:appendix-proofs}.

While removing one source of slack, however, peCQR introduces another
because the projection creates ties in the calibration scores. Whenever a left-censored data point has its
lower base quantile projected to the boundary
($\Pi\hat{f}^{\mathrm{lo}}(X_i) = C^{\mathrm{L}} = \tilde{Y}_i$), its score collapses to zero:
\[
    \max\{\Pi\hat{f}^{\mathrm{lo}}(X_i) - \tilde{Y}_i,\;
          \tilde{Y}_i - \Pi\hat{f}^{\mathrm{up}}(X_i)\}
    \;=\; \max\{0,\; C^{\mathrm{L}} - \Pi\hat{f}^{\mathrm{up}}(X_i)\} \;=\; 0,
\]
since $\Pi\hat{f}^{\mathrm{up}}(X_i) \geq C^{\mathrm{L}}$ after projection; symmetrically,
every right-censored point whose upper quantile is projected to $C^{\mathrm{R}}$ also has
score $0$. These ties cannot be broken at random by jittering the projected quantiles,
because that would risk pushing them back outside $[C^{\mathrm{L}}, C^{\mathrm{R}}]$,
undoing the cancellation of the $p_{\mathrm{slack}}$ term. We can nonetheless bound the residual over-coverage of peCQR.

\begin{theorem}\label{thm:pecqr-validity}
Suppose the only ties among the peCQR calibration scores are those at 0 induced by the projection $\Pi$. Then,
\[
    1 - \alpha \leq
    \mathbb{P}\bigl[Y_{n+1} \in \widehat{\mathcal{C}}^{\mathrm{peCQR}}(X_{n+1})\bigr]
    \leq 1 - \alpha + \frac{1}{n+1} + p_{\mathrm{zero}}(\hat{f}),
\]
where $p_{\mathrm{zero}}(\hat{f}) :=
    \mathbb{P}\bigl(Y \leq C^{\mathrm{L}},\; \Pi \hat{f}^{\mathrm{lo}}(X) = C^{\mathrm{L}}\bigr)
    +
    \mathbb{P}\bigl(Y \geq C^{\mathrm{R}},\; \Pi \hat{f}^{\mathrm{up}}(X) \geq C^{\mathrm{R}}\bigr)$
represents the probability of having a score equal to zero due to censoring.
\end{theorem}
\noindent The proof is deferred to Appendix~\ref{sec:appendix-proofs}.

\clearpage

\section{Extended Theoretical Analysis} \label{app:theory}

\subsection{Relation between eCQR, peCQR, ClipCQR, and ClipCQR+}

ClipCQR always tightens eCQR, with equality whenever every calibration point keeps its
base midpoint prediction $\hat{m}(X_i) := (\hat{f}^{\mathrm{lo}}(X_i) + \hat{f}^{\mathrm{up}}(X_i))/2$ inside $[C^{\mathrm{L}}, C^{\mathrm{R}}]$.

\begin{proposition} \label{prop:ecqr-ccqr}
If eCQR (Algorithm~\ref{alg:ecqr}) and ClipCQR (Algorithm~\ref{alg:clipcqr}) are applied with the same data, predictor
$\hat{f}$, and level $\alpha$, and no jittering ($\delta = 0$), then 
$\widehat{\mathcal{C}}(X_{n+1}) \subseteq \widehat{\mathcal{C}}^{\mathrm{eCQR}}(X_{n+1})$ almost surely, 
with equality on the event $\mathcal{E} := \bigl\{\hat{m}(X_i) \in [C^{\mathrm{L}}, C^{\mathrm{R}}] \; \forall i \in [n]\bigr\}$.
\end{proposition}
\noindent The proof is deferred to Appendix~\ref{sec:appendix-proofs}.

If the base prediction intervals always lie within $(C^{\mathrm{L}}, C^{\mathrm{R}})$, all four methods coincide.

\begin{corollary} \label{cor:ecqr-pecqr-ccqr}
Suppose eCQR (Algorithm~\ref{alg:ecqr}), peCQR (Algorithm~\ref{alg:pecqr}),
ClipCQR (Algorithm~\ref{alg:clipcqr}), and ClipCQR+ (Algorithm~\ref{alg:clipcqr_plus}) are applied with the
same data, base predictor $\hat{f}$, and level $\alpha$, and no jittering ($\delta = 0$). Then, on the event
$\mathcal{E}' := \bigl\{C^{\mathrm{L}} < \hat{f}^{\mathrm{lo}}(X_i) \le \hat{f}^{\mathrm{up}}(X_i) < C^{\mathrm{R}}
  \;\; \forall i \in [n+1] \bigr\}$,
\[
    \widehat{\mathcal{C}}(X_{n+1})
    = \widehat{\mathcal{C}}^+(X_{n+1})
    = \widehat{\mathcal{C}}^{\mathrm{peCQR}}(X_{n+1})
    = \widehat{\mathcal{C}}^{\mathrm{eCQR}}(X_{n+1}).
\]
\end{corollary}
\noindent The proof is deferred to Appendix~\ref{sec:appendix-proofs}.

\subsection{Analysis of ClipCQR} \label{app:theory-clipcqr}

Inheriting the definitions and notation from Section~\ref{sec:theory}, 
recall explicitly the oracle family
\begin{align} \label{eq:oracle-family}
\phi^\star(x;\tau) := \psi(q^{\mathrm{lo}}(x) - \tau,\, q^{\mathrm{up}}(x) + \tau).
\end{align}
Consider a different member of this family, the \emph{marginal oracle}:
\begin{align} \label{eq:marginal-oracle}
  & \mathcal{O}_{\mathrm{m}}(X) := \phi^\star(x; \tau^{\mathrm{m}}),
  & \tau^{\mathrm{m}} := \inf\{ t : \mathbb{P}(s^\star(X, Y) \leq t) \geq 1-\alpha \},
\end{align}
where $s^\star(x, y) := \inf\{\tau \in \mathbb{R} : y \in \phi^\star(x; \tau)\}$
is the oracle score function, obtained from~\eqref{eq:score-formula} by replacing
$\hat{f}^{\mathrm{lo}}, \hat{f}^{\mathrm{up}}$ with $q^{\mathrm{lo}}, q^{\mathrm{up}}$:
\begin{align} \label{eq:oracle-score}
s^\star(x, y) =
\begin{cases}
\max\bigl\{\, q^{\mathrm{lo}}(x) - C^{\mathrm{L}},\ \tfrac12(q^{\mathrm{lo}}(x) - q^{\mathrm{up}}(x)) \,\bigr\}
  =: c^{\mathrm L}(x), & y \leq C^{\mathrm{L}}, \\[4pt]
\max\bigl\{\, q^{\mathrm{lo}}(x) - y,\ y - q^{\mathrm{up}}(x) \,\bigr\}, & C^{\mathrm{L}} < y < C^{\mathrm{R}}, \\[4pt]
\max\bigl\{\, C^{\mathrm{R}} - q^{\mathrm{up}}(x),\ \tfrac12(q^{\mathrm{lo}}(x) - q^{\mathrm{up}}(x)) \,\bigr\}
  =: c^{\mathrm R}(x), & y \geq C^{\mathrm{R}}.
\end{cases}
\end{align}

The following result shows that the three oracles --- $\mathcal{O}_{\mathrm{uncens}},\mathcal{O}_0,\mathcal{O}_{\mathrm{m}}$ ---  all coincide under a strict identifiability condition,
and that without it they trade marginal coverage against conditional
coverage: the snapped oracle is conditionally conservative but marginally
over-covers, while the marginal oracle is marginally tightest but conditionally
under-covers.

% \begin{assumption}[Weak regularity]\label{ass:continuity}
% For almost every $x$, $F_{Y \mid X}(\cdot \mid x)$ has a positive
% finite density in a neighbourhood of $q^{\mathrm{lo}}(x)$ and $q^{\mathrm{up}}(x)$,
% and the quantiles place no mass on the censoring thresholds,
% $\mathbb{P}(q^{\mathrm{lo}}(X) = C^{\mathrm{L}}) = \mathbb{P}(q^{\mathrm{up}}(X) = C^{\mathrm{R}}) = 0$.
% (This is weaker than Assumption~\ref{ass:regularity}.)
% \end{assumption}

\begin{assumption}[Strict identifiability]\label{ass:observability}
$C^{\mathrm{L}} < q^{\mathrm{lo}}(X) \leq q^{\mathrm{up}}(X) < C^{\mathrm{R}}$ almost surely.
\end{assumption}

\begin{proposition}\label{prop:oracle-coverage}
With no distributional assumptions:
\begin{enumerate}[label=(\roman*)]
\item $\mathcal{O}_{\mathrm{m}}(X) \subseteq \mathcal{O}_0(X)$ almost surely, and $\mathcal{O}_{\mathrm{m}}$ is the tightest
predictor in the oracle family with $\mathbb{P}(Y \in \mathcal{O}_{\mathrm{m}}(X)) \geq 1-\alpha$.
\end{enumerate}
Define $\pi^{\mathrm L} := \mathbb{P}(q^{\mathrm{lo}}(X) \leq C^{\mathrm{L}})$ and
$\pi^{\mathrm R} := \mathbb{P}(q^{\mathrm{up}}(X) \geq C^{\mathrm{R}})$. Under Assumption~\ref{ass:regularity}:
\begin{enumerate}[label=(\roman*)]\setcounter{enumi}{1}

\item the snapped oracle is conditionally conservative, $\mathbb{P}(Y \in \mathcal{O}_0(X) \mid X)\ge 1-\alpha$, with
\[
  \mathbb{P}(Y \in \mathcal{O}_0(X) \mid X)
  = \begin{cases}
    F_{Y\mid X}(C^{\mathrm L}\mid X), & \text{if }q^{\mathrm{up}}(X)<C^{\mathrm L},\\
    1-F_{Y\mid X}(C^{\mathrm R}\mid X), & \text{if } q^{\mathrm{lo}}(X)>C^{\mathrm R}\\
    1-\alpha + \tfrac{\alpha}{2}\bigl(\I\{q^{\mathrm{lo}}(X) \le C^{\mathrm L}\} + \I\{q^{\mathrm{up}}(X) \ge C^{\mathrm R}\}\bigr),
      & \text{otherwise};
  \end{cases}
\]
consequently $\mathbb{P}(Y \in \mathcal{O}_0(X)) \ge 1-\alpha + \tfrac{\alpha}{2}(\pi^{\mathrm L} + \pi^{\mathrm R})$, with equality iff the first two cases occur with probability zero;

\item the marginal oracle attains exact marginal coverage,
$\mathbb{P}(Y \in \mathcal{O}_{\mathrm{m}}(X)) = 1-\alpha$, but is not conditionally conservative:
$\mathbb{P}(Y \in \mathcal{O}_{\mathrm{m}}(X) \mid X = x) \leq 1-\alpha$ for any $x \in \mathcal{X}$ with
$C^{\mathrm{L}} < q^{\mathrm{lo}}(x) \leq q^{\mathrm{up}}(x) < C^{\mathrm{R}}$, with equality if and only if
Assumption~\ref{ass:observability} holds;
\item $\mathcal{O}_{\mathrm{m}} = \mathcal{O}_0 = \mathcal{O}_{\mathrm{uncens}}$ if and only if
Assumption~\ref{ass:observability} holds.
\end{enumerate}
\end{proposition}

The next theorem establishes that ClipCQR can approach the $\mathcal{O}_{\mathrm{m}}$ oracle.
This requires a consistency assumption (Assumption~\ref{ass:full-consistency}) that is stronger than the base consistency assumption (Assumption~\ref{ass:consistency}) required for 
ClipCQR+ to approach $\mathcal{O}_0$; see Proposition~\ref{prop:consistency-weaker}.

\begin{assumption}[Full consistency]\label{ass:full-consistency}
The base quantile estimates are $L^1$-consistent:
\[
    \mathbb{E} \bigl( |\hat{f}_n^{\mathrm{lo}}(X) - q^{\mathrm{lo}}(X)| \bigr) \to 0
    \quad\text{and}\quad
    \mathbb{E} \bigl( |\hat{f}_n^{\mathrm{up}}(X) - q^{\mathrm{up}}(X)| \bigr) \to 0,
    \qquad \text{ as } n \to \infty.
\]
\end{assumption}

\begin{theorem}\label{thm:ccqr-oracle}
Assume $(X_i, Y_i)_{i=1}^{n+1}$ are i.i.d.
Let $\widehat{\tau}_n(\alpha)$ denote the ClipCQR threshold at sample size $n$, 
with base predictors $\hat{f}_n^{\mathrm{lo}}, \hat{f}_n^{\mathrm{up}}$. 
Under Assumptions~\ref{ass:regularity} and~\ref{ass:full-consistency},
\begin{enumerate}
    \item[(i)] $\widehat{\tau}_n(\alpha) \xrightarrow{p} \tau^{\mathrm{m}} \leq 0$, and
    $\widehat{\mathcal{C}}_n \to \mathcal{O}_{\mathrm{m}}$, where $\mathcal{O}_{\mathrm{m}}$ is the marginal oracle
    \eqref{eq:marginal-oracle}.
    \item[(ii)] $\mathbb{P}(Y \in \widehat{\mathcal{C}}_n(X)) \to \mathbb{P}(Y \in \mathcal{O}_{\mathrm{m}})$.
    \item[(iii)] $\widehat{\mathcal{C}}_n$ has asymptotic conditional
    coverage matching $\mathcal{O}_{\mathrm{m}}$ (Definition~\ref{def:cond-cov}).
\end{enumerate}
\end{theorem}

\subsection{Analysis of eCQR}

We now show that eCQR is consistent with neither $\mathcal{O}_{\mathrm{m}}$ nor $\mathcal{O}_0$. 
Instead, under assumptions similar to those of Theorem~\ref{thm:ccqr-oracle},
eCQR converges to a distinct oracle $\mathcal{O}_{\mathrm{e}}$ that always contains
$\mathcal{O}_{\mathrm{m}}$ but bears no fixed relationship to $\mathcal{O}_0$. 
Moreover, we show that $\mathcal{O}_{\mathrm{e}}$ is less than ideal: it is relatively less interpretable, 
lacks conditional coverage in general, and is sometimes unnecessarily more conservative than
$\mathcal{O}_0$.

Define the oracle CQR score
\begin{equation}\label{eq:cqr-oracle-score}
    s^{\mathrm{CQR},\star}(x,y):=\max\bigl\{q^{\mathrm{lo}}(x)-y,\ y-q^{\mathrm{up}}(x)\bigr\},
    \qquad \tilde{S}^{\mathrm{CQR},\star}:=s^{\mathrm{CQR},\star}(X,\tilde Y),
\end{equation}
so that $s^{\mathrm{CQR},\star}(x,y)\le t\iff y\in[q^{\mathrm{lo}}(x)-t,\,q^{\mathrm{up}}(x)+t]$.
Define also its
$(1-\alpha)$-quantile as:
\[
\tilde\tau^{\mathrm{m}}(\alpha) := \inf\{t:\mathbb P(\tilde{S}^{\mathrm{CQR},\star}\le t)\ge1-\alpha\}.
\]
With this notation, the \emph{expanded oracle} is defined as:
\begin{equation}\label{eq:expanded-oracle}
    \mathcal{O}_{\mathrm{e}}(X):=\phi^\star(X;\tilde\tau^{\mathrm{m}})
    =\psi\bigl(q^{\mathrm{lo}}(X)-\tilde\tau^{\mathrm{m}}(\alpha),\ q^{\mathrm{up}}(X)+\tilde\tau^{\mathrm{m}}(\alpha) \bigr).
\end{equation}
In words, this is the member of the oracle family~\eqref{eq:oracle-family} whose raw pre-snap interval
$[q^{\mathrm{lo}}-\tau,\,q^{\mathrm{up}}+\tau]$ covers the censored $\tilde Y$ at level $1-\alpha$.
The relative difficulty in interpreting $\mathcal{O}_{\mathrm{e}}(X):=\phi^\star(X;\tilde\tau^{\mathrm{m}})$ is that the sign 
of $\tilde\tau^{\mathrm{m}}(\alpha)$ depends on the data distribution and censoring thresholds.
Therefore, $\mathcal O_{\mathrm{e}}$ may either contract or expand $\mathcal O_0$, as illustrated by the following examples.

\begin{example}[Shallow censoring: $\tilde\tau^{\mathrm{m}}<0$ and $\mathcal O_{\mathrm{e}} \subset \mathcal O_0$]\label{ex:ecqr-neg}
Let $Y\sim\mathrm{Unif}[0,10]$ independently of $X$, so $q^{\mathrm{lo}}\equiv1$ and $q^{\mathrm{up}}\equiv9$ at $\alpha = 0.2$, and take
$C^{\mathrm L}=3$, $C^{\mathrm R}=10$ (only shallow left censoring, with $q^{\mathrm{lo}}<C^{\mathrm L}<q^{\mathrm{up}}<C^{\mathrm R}$). Then $\tilde Y=\max(Y,3)$ and the pre-snap oracle band over-covers $\tilde{Y}$:
\[
  \mathbb P\bigl(\tilde Y\in[q^{\mathrm{lo}}(X),\,q^{\mathrm{up}}(X)]\bigr)
  =\mathbb P\bigl(\max(Y,3)\in[1,9]\bigr)=\mathbb P(Y\le9)=0.9>1-\alpha = 0.8.
\]
Contracting to $[q^{\mathrm{lo}}(X)-t,\,q^{\mathrm{up}}(X)+t]$ with $t=-1$ gives exact $1-\alpha$ marginal coverage:
\[
  \mathbb P\bigl(\tilde Y\in[q^{\mathrm{lo}}(X)+1,\,q^{\mathrm{up}}(X)-1]\bigr)
  =\mathbb P\bigl(\max(Y,3)\in[2,8]\bigr)=\mathbb P(Y\le 8)=0.8 = 1-\alpha.
\]
Therefore, $\tilde\tau^{\mathrm{m}}=-1<0$ and $\mathcal O_{\mathrm{e}}=\psi(2,8)=[0,8]\subset [0,9]=\mathcal O_0$.
\end{example}

\begin{example}[Deep censoring: $\tilde\tau^{\mathrm{m}}>0$ and $\mathcal{O}_0\subset\mathcal{O}_{\mathrm{e}}$]\label{ex:ecqr-pos}
Let $C^{\mathrm L}=0$, $C^{\mathrm R}=10$, $\alpha=0.2$, and $X=(X^{(1)},X^{(2)})$ with
$X^{(1)}\sim\mathrm{Bernoulli}(1/2)$ independent of $X^{(2)}\sim\mathrm{Unif}[1,2]$:
\begin{itemize}
\item if $X^{(1)}=0$: $Y\sim\mathrm{Unif}[0,10]$, so $q^{\mathrm{lo}}(X)=1$,
$q^{\mathrm{up}}(X)=9$ are interior and $\tilde Y=Y$;
\item if $X^{(1)}=1$: $Y\sim\mathrm{Unif}[-9-X^{(2)},\,1-X^{(2)}]$, so
$q^{\mathrm{lo}}(X)=-8-X^{(2)}$ and $q^{\mathrm{up}}(X)=-X^{(2)}<C^{\mathrm L}$, and $\tilde Y=0$.
\end{itemize}
The pre-snap oracle band under-covers:
\begin{align*}
  \mathbb P\bigl(\tilde Y\in[q^{\mathrm{lo}}(X),\,q^{\mathrm{up}}(X)]\bigr)
  &= \tfrac12\,\mathbb P\bigl(Y\in[1,9]\mid X^{(1)}=0\bigr)
   + \tfrac12\,\mathbb P\bigl(0\in[q^{\mathrm{lo}}(X),\,q^{\mathrm{up}}(X)]\mid X^{(1)}=1\bigr)\\
  &= \tfrac12\cdot 0.8 + \tfrac12\cdot 0 = 0.4 < 1-\alpha = 0.8 .
\end{align*}
Expanding to $[q^{\mathrm{lo}}(X)-t,\,q^{\mathrm{up}}(X)+t]$: for $t\in[1,2]$ leads to coverage 1 conditional on $X^{(1)}=0$
and, conditional on $X^{(1)}=1$, $\mathbb P\bigl(q^{\mathrm{up}}(X)+t\ge0 \mid X^{(1)}=1\bigr)=\mathbb P(X^{(2)}\le t)=t-1$. Hence
\begin{align*}
  \mathbb P\bigl(\tilde Y\in[q^{\mathrm{lo}}(X) -t ,\,q^{\mathrm{up}}(X) + t]\bigr)
  =\tfrac12\cdot 1 + \tfrac12\,(t-1)=\tfrac{t}{2},
  \qquad t\in[1,2],
\end{align*}
which implies $\tilde\tau^{\mathrm{m}}=1.6>0$. 
Conditional on $X^{(1)}=0$ this
inflates the prediction interval to the full support,
$\mathcal{O}_{\mathrm{e}}=\psi(1-1.6,\,9+1.6)=\psi(-0.6,\,10.6)=[Y^{\min},Y^{\max}]\;\supset\;[1,9]=\mathcal{O}_0$.
On the sub-population with $X^{(1)}=1$ this likewise widens the interval, but more mildly. The snapped oracle is
\[
  \mathcal{O}_0=\psi\bigl(-8-X^{(2)},\,-X^{(2)}\bigr)=[Y^{\min},\,C^{\mathrm L}]=[Y^{\min},0],
\]
which simply reports that the outcome lies below the detection limit. Expanding by $\tilde\tau^{\mathrm{m}}=1.6$,
\[
  \mathcal{O}_{\mathrm{e}}=\psi\bigl(-9.6-X^{(2)},\,1.6-X^{(2)}\bigr),
\]
whose upper endpoint $1.6-X^{(2)}$ rises above $C^{\mathrm L}=0$ precisely for the shallower censored points with
$X^{(2)}<1.6$. There $\mathcal{O}_{\mathrm{e}}=[Y^{\min},\,1.6-X^{(2)}]\supset [Y^{\min},0]=\mathcal{O}_0$, spuriously
extending the interval past the detection limit into the observable range; for the deeper points $X^{(2)}\ge1.6$ the
expansion is too small to reach $C^{\mathrm L}$ and $\mathcal{O}_{\mathrm{e}}=\mathcal{O}_0=[Y^{\min},0]$.
Combining the two cases, we see that $\mathcal{O}_0\subseteq\mathcal{O}_{\mathrm{e}}$ almost surely, with strict inclusion on a set of
positive probability.
\end{example}

Taken together, these examples underscore how difficult $\mathcal{O}_{\mathrm{e}}$ is to interpret. Unlike
$\mathcal{O}_0$, which is the intuitive snapping $\psi$ of the equal-tailed latent oracle
$[q^{\mathrm{lo}},q^{\mathrm{up}}]$, the expanded oracle $\mathcal{O}_{\mathrm{e}}$ applies $\psi$ to a globally shifted band
$[q^{\mathrm{lo}}-\tilde\tau^{\mathrm{m}},\,q^{\mathrm{up}}+\tilde\tau^{\mathrm{m}}]$ whose shift is calibrated against the
wrong target, the censored $\tilde{Y}$. As the examples show, $\tilde\tau^{\mathrm{m}}$ has no fixed
sign and $\mathcal{O}_{\mathrm{e}}$ bears no fixed relationship to $\mathcal{O}_0$, lying strictly inside it in
Example~\ref{ex:ecqr-neg} and strictly outside it in Example~\ref{ex:ecqr-pos}. Example~\ref{ex:ecqr-pos} makes the resulting
inefficiency of $\mathcal{O}_{\mathrm{e}}$ clear: $\mathcal{O}_0$ already attains at least $1-\alpha$ latent coverage on both
sub-populations, yet $\mathcal{O}_{\mathrm{e}}$ inflates the prediction intervals in \emph{both}. The
expanded oracle is therefore uniformly less informative than $\mathcal{O}_0$.

The following proposition characterizes precisely the marginal and conditional coverage of the expanded oracle $\mathcal{O}_{\mathrm{e}}$, as well as its relation to  the marginal oracle $\mathcal{O}_{\mathrm{m}}$.
The result shows that $\mathcal{O}_{\mathrm{e}}(X)$ is never tighter than $\mathcal{O}_{\mathrm{m}}(X)$, and marginally over-covers with a slack term that can be interpreted as
the population analogue of $p_{\mathrm{slack}}$ in Theorem~\ref{thm:ecqr-validity}.
For $t \in \mathbb{R}$, define
\begin{align}\label{eq:oracle-slack}
\begin{split}
    p_{\mathrm{slack}}^\star(t)
    & := \mathbb P\bigl(q^{\mathrm{lo}}(X)-t\le q^{\mathrm{up}}(X)+t<C^{\mathrm L},\,Y \leq C^{\mathrm L}\bigr) \\
    & \qquad + \mathbb P\bigl(q^{\mathrm{up}}(X)+t\ge q^{\mathrm{lo}}(X)-t>C^{\mathrm R},\,Y \geq C^{\mathrm R}\bigr).
  \end{split}
\end{align}

\begin{proposition}\label{prop:ecqr-oracle}
With no distributional assumptions:
\begin{enumerate}[label=(\roman*)]
\item $\tau^{\mathrm{m}}\le\tilde\tau^{\mathrm{m}}$, hence $\mathcal O_{\mathrm{m}}(X)\subseteq\mathcal O_{\mathrm{e}}(X)$, almost surely.
\end{enumerate}
Under Assumption~\ref{ass:regularity}:
\begin{enumerate}[label=(\roman*)]\setcounter{enumi}{1}
\item $\mathbb P(Y\in\mathcal O_{\mathrm{e}}(X)) = 1-\alpha+p_{\mathrm{slack}}^\star(\tilde\tau^{\mathrm{m}})\ge1-\alpha$;
\item the sign of $\tilde\tau^{\mathrm{m}}$ is determined by the censoring probabilities: writing
$\rho^{\mathrm L}:=\mathbb P(q^{\mathrm{up}}(X)<C^{\mathrm L})$ and $\rho^{\mathrm R}:=\mathbb P(q^{\mathrm{lo}}(X)>C^{\mathrm R})$,
then 
\[
  \tilde\tau^{\mathrm{m}}<0 \iff \tfrac\alpha2(\pi^{\mathrm L}+\pi^{\mathrm R}) > (1-\tfrac\alpha2)(\rho^{\mathrm L}+\rho^{\mathrm R}),
\]
with equality giving $\tilde\tau^{\mathrm{m}}=0$; in particular $\tilde\tau^{\mathrm{m}}<0$ whenever $\rho^{\mathrm L}=\rho^{\mathrm R}=0$
and $\pi^{\mathrm L}+\pi^{\mathrm R}>0$. Accordingly, $\mathcal O_{\mathrm{e}}(X)\subseteq\mathcal O_0(X)$ if $\tilde\tau^{\mathrm{m}}\le0$
and $\mathcal O_{\mathrm{e}}(X)\supseteq\mathcal O_0(X)$ if $\tilde\tau^{\mathrm{m}}\ge0$; on
$\{C^{\mathrm L}<q^{\mathrm{lo}}(X)\le q^{\mathrm{up}}(X)<C^{\mathrm R}\}$ the conditional coverage
\[
  \mathbb P(Y\in\mathcal O_{\mathrm{e}}(X)\mid X)\le1-\alpha \iff \tilde\tau^{\mathrm{m}}\le0 \qquad (\text{strictly if } \tilde\tau^{\mathrm{m}}<0).
\]
\item $\mathcal O_{\mathrm{e}}=\mathcal O_{\mathrm{m}}$ iff $p_{\mathrm{slack}}^\star(\tau^{\mathrm{m}})=0$; and
$\mathcal O_{\mathrm{e}}=\mathcal O_0=\mathcal O_{\mathrm{m}}=\mathcal O_{\mathrm{uncens}}$ iff
Assumption~\ref{ass:observability} holds.
\end{enumerate}
\end{proposition}
\noindent The proof is deferred to Appendix~\ref{sec:appendix-proofs}.

Finally, we establish the asymptotic consistency of eCQR with $\mathcal O_{\mathrm{e}}$, under
assumptions close to those of Theorem~\ref{thm:ccqr-oracle}. The only addition is a continuity
requirement at the oracle calibration level. In Theorem~\ref{thm:ccqr-oracle} the analogous property is
supplied by Assumption~\ref{ass:regularity}, but that assumption no longer suffices here: the oracle
score $\tilde S^{\mathrm{CQR},\star}$ in~\eqref{eq:cqr-oracle-score} is a function of the censored
outcome $\tilde Y$, and therefore inherits the point masses that $\tilde Y$ may carry at $C^{\mathrm L}$
and $C^{\mathrm R}$. We thus assume directly that $\tilde\tau^{\mathrm{m}}(\alpha)$ is the unique
$(1-\alpha)$-quantile of $\tilde S^{\mathrm{CQR},\star}$ and that its law places no point mass at
$\tilde\tau^{\mathrm{m}}(\alpha)$. The condition is mild; for example, it holds both in Example~\ref{ex:ecqr-neg} 
(where the score distribution has a point mass at $-2$ but $\tilde\tau^{\mathrm{m}}=-1$) 
and in Example~\ref{ex:ecqr-pos} (where the score distribution has no point masses).

\begin{theorem}\label{thm:ecqr-oracle}
Assume $(X_i, Y_i)_{i=1}^{n+1}$ are i.i.d.
Suppose Assumptions~\ref{ass:regularity} and~\ref{ass:full-consistency} hold, and that
$\tilde\tau^{\mathrm{m}}(\alpha)$ is the unique $(1-\alpha)$-quantile of $\tilde S^{\mathrm{CQR},\star}$
and a continuity point of its law. 
Let $\tilde\tau_n(\alpha)$ denote the CQR threshold at sample size $n$, 
with base predictors $\hat{f}_n^{\mathrm{lo}}, \hat{f}_n^{\mathrm{up}}$. 
Then:
\begin{enumerate}
\item[(i)] $\tilde\tau_n(\alpha)\xrightarrow{p}\tilde\tau^{\mathrm{m}}(\alpha)$ and
$\widehat{\mathcal C}^{\,\mathrm{eCQR}}_n\to\mathcal O_{\mathrm{e}}$;
\item[(ii)] $\mathbb P\bigl(Y\in\widehat{\mathcal C}^{\,\mathrm{eCQR}}_n(X)\bigr)
\to\mathbb P\bigl(Y\in\mathcal O_{\mathrm{e}}(X)\bigr)=1-\alpha+p_{\mathrm{slack}}^\star(\tilde\tau^{\mathrm{m}})$;
\item[(iii)] $\widehat{\mathcal C}^{\,\mathrm{eCQR}}_n$ has asymptotic conditional coverage matching
$\mathcal O_{\mathrm{e}}$ (Definition~\ref{def:cond-cov}).
\end{enumerate}
\end{theorem}
\noindent The proof is deferred to Appendix~\ref{sec:appendix-proofs}.

% \subsection{Summary}

% Table~\ref{tab:method-summary} summarizes these results and provides an overall comparison of the four methods.

\subsection{Analysis of peCQR} 
We show that peCQR is equivalent to ClipCQR applied with the \emph{projected} base predictor
(Lemma~\ref{lem:pecqr-is-projected-clipcqr}). Interestingly, Theorem~\ref{thm:pecqr-oracle} shows this projection suffices to make peCQR asymptotically target the snapped conditional
oracle $\mathcal{O}_0$, similar to ClipCQR+, rather than a contracted marginal oracle like ClipCQR and eCQR.

\begin{lemma}\label{lem:pecqr-is-projected-clipcqr}
With the same data, level $\alpha$, and no jittering, peCQR (Algorithm~\ref{alg:pecqr}) produces exactly
the calibration scores, threshold, and prediction band of ClipCQR (Algorithm~\ref{alg:clipcqr}) applied to
the projected base predictor $\Pi\hat f$.
\end{lemma}
\noindent The proof is deferred to Appendix~\ref{sec:appendix-proofs}.

\begin{theorem}\label{thm:pecqr-oracle}
Assume $(X_i, Y_i)_{i=1}^{n+1}$ are i.i.d. 
Let $\tilde\tau^{\mathrm{pe}}_n(\alpha)$ denote the peCQR threshold at sample size $n$, 
with base predictors $\hat{f}_n^{\mathrm{lo}}, \hat{f}_n^{\mathrm{up}}$. 
Under Assumptions~\ref{ass:regularity},~\ref{ass:consistency},
and~\ref{ass:base-regularity}:
\begin{enumerate}
\item[(i)] $\tilde\tau^{\mathrm{pe}}_n(\alpha)\xrightarrow{p}0$ and $\widehat{\mathcal C}^{\,\mathrm{peCQR}}_n\to\mathcal{O}_0$;
\item[(ii)] $\mathbb{P}(Y \in \widehat{\mathcal{C}}^{\,\mathrm{peCQR}}_n(X)) \to \mathbb{P}(Y \in \mathcal{O}_0)$;
\item[(iii)] $\widehat{\mathcal{C}}^{\,\mathrm{peCQR}}_n$ has asymptotic conditional coverage matching $\mathcal{O}_0$
(Definition~\ref{def:cond-cov}).
\end{enumerate}
\end{theorem}
\noindent The proof is deferred to Appendix~\ref{sec:appendix-proofs}.

\section{Relation to \cite{liu2025prediction}} \label{sec:app-relation-double}

\citet{liu2025prediction} study the more general problem of constructing conformal prediction
intervals for continuous outcomes under arbitrary double censoring. In their setting, one observes
data points $(X_i, Y_i^{\mathrm{L}}, Y_i^{\mathrm{U}})$ with $Y_i \in [Y_i^{\mathrm{L}}, Y_i^{\mathrm{U}}]
\subseteq \mathbb{R}$, assuming nothing about the censoring mechanism, and aims to predict $Y_{n+1}$
given $X_{n+1}$, assuming exchangeability of $\{(X_i, Y_i)\}_{i=1}^{n+1}$. Our clipped-outcome problem
is the special case in which the brackets come from clipping: $[Y_i^{\mathrm{L}},Y_i^{\mathrm{U}}] =
[Y_i,Y_i]$ when uncensored, $[Y_i^{\mathrm{L}},Y_i^{\mathrm{U}}] = [Y^{\min}, C^{\mathrm{L}}]$ when
left-censored, and $[Y_i^{\mathrm{L}},Y_i^{\mathrm{U}}] = [C^{\mathrm{R}}, Y^{\max}]$ when right-censored.

The \emph{LdPT} method proposed by \citet{liu2025prediction} applies CQR using the score function
\begin{align*}
    s_{\mathrm{LdPT}}(x, y^{\mathrm{L}}, y^{\mathrm{U}})
  & = \max\{\hat{f}^{\mathrm{lo}}(x) - y^{\mathrm{L}},\ y^{\mathrm{U}} - \hat{f}^{\mathrm{up}}(x)\} \\
  & = \begin{cases}
        \max\bigl\{\hat{f}^{\mathrm{lo}}(x) - Y^{\min},\, C^{\mathrm{L}} - \hat{f}^{\mathrm{up}}(x)\bigr\}, & \text{if } y \leq C^{\mathrm{L}}, \\
        \max\bigl\{\hat{f}^{\mathrm{lo}}(x) - y,\, y - \hat{f}^{\mathrm{up}}(x)\bigr\}, & \text{if } C^{\mathrm{L}} < y < C^{\mathrm{R}}, \\
        \max\bigl\{\hat{f}^{\mathrm{lo}}(x) - C^{\mathrm{R}},\, Y^{\max} - \hat{f}^{\mathrm{up}}(x)\bigr\}, & \text{if } y \geq C^{\mathrm{R}},
    \end{cases}
\end{align*}
which is negative exactly when $[y^{\mathrm{L}}, y^{\mathrm{U}}] \subseteq [\hat{f}^{\mathrm{lo}}(x), \hat{f}^{\mathrm{up}}(x)]$.

Their approach tends to produce unnecessarily large prediction intervals in our clipped setting,
especially when $Y^{\min} \ll C^{\mathrm{L}}$ or $Y^{\max} \gg C^{\mathrm{R}}$. Indeed, a positive score
$s_{\mathrm{LdPT}}(x, y^{\mathrm{L}},y^{\mathrm{U}})$ requires the base interval to contain the
\emph{entire} bracket $[y^{\mathrm{L}},y^{\mathrm{U}}]$: a left-censored point requires the base lower
endpoint to reach all the way down to $Y^{\min}$, incurring $\hat{f}^{\mathrm{lo}}(x) - Y^{\min}$.
Moreover, the width of every LdPT prediction interval, including the interior ones, depends on
$Y^{\min}$ and $Y^{\max}$, inflating as these move farther from $C^{\mathrm{L}}$ and $C^{\mathrm{R}}$,
which is wasteful. By contrast, the ClipCQR score charges a left-censored point only the distance to
the clipping threshold, $\hat{f}^{\mathrm{lo}}(x) - C^{\mathrm{L}}$, floored at the negative half-width
$\tfrac{1}{2}(\hat{f}^{\mathrm{lo}}(x) - \hat{f}^{\mathrm{up}}(x))$, so the ClipCQR scores are typically
smaller, leading to more informative prediction intervals.

Moreover, LdPT does not target conditional coverage. Consider the degenerate
case $Y^{\min} = C^{\mathrm{L}}$ and $Y^{\max} = C^{\mathrm{R}}$, in which the observable range spans the
entire support and no snapping is needed. In this case the bracket collapses to the clipped point and the
score reduces to $s_{\mathrm{LdPT}}(x, y^{\mathrm{L}}, y^{\mathrm{U}}) = s^{\mathrm{CQR}}(x, \Pi(y))$;
that is, LdPT coincides exactly with eCQR, which converges to the expanded oracle $\mathcal{O}_{\mathrm{e}}$
and therefore does not guarantee conditional coverage in general (Theorem~\ref{thm:ecqr-oracle}).

\clearpage

\section{Auxiliary Theoretical Results and Mathematical Proofs}\label{sec:appendix-proofs}

% \subsection{Helpful Notation}

% Define $R := Y^{\max}-Y^{\min}$, so that all scores lie in $[-R,R]$.
% Write the snapping function~\eqref{eq:psi-def} coordinatewise: for any $a \leq b$,
% $\psi(a,b) = [\underline{\psi}(a), \overline{\psi}(b)]$, with
% \[
%     \underline{\psi}(a) = \begin{cases} Y^{\min}, & a \leq C^{\mathrm{L}}, \\
%         \min\{a, C^{\mathrm{R}}\}, & a > C^{\mathrm{L}}, \end{cases}
%     \qquad
%     \overline{\psi}(b) = \begin{cases} \max\{b, C^{\mathrm{L}}\}, & b < C^{\mathrm{R}}, \\
%         Y^{\max}, & b \geq C^{\mathrm{R}}. \end{cases}
% \]

\subsection{Auxiliary Theoretical Results} \label{sec:appendix-auxiliary}

The proofs of results presented here are deferred to Appendix~\ref{app:proofs-auxiliary}.

\begin{assumption}[Snapped consistency]\label{ass:snapped-consistency}
With $\underline\psi$ and $\overline\psi$ denoting the upper and lower coordinates of $\psi$ in~\eqref{eq:psi-def},
$\mathbb{E} \bigl( |\underline{\psi}(\hat f^{\mathrm{lo}}_n(X)) - \underline{\psi}(q^{\mathrm{lo}}(X))| \bigr)\to 0$ and
$\mathbb{E} \bigl( |\overline{\psi}(\hat f^{\mathrm{up}}_n(X)) - \overline{\psi}(q^{\mathrm{up}}(X))| \bigr) \to 0$.
\end{assumption}

\begin{lemma}\label{lem:snapped-from-projected}
Assumptions~\ref{ass:regularity}--\ref{ass:base-regularity} imply Assumption~\ref{ass:snapped-consistency}.
\end{lemma}

\begin{proposition} \label{prop:consistency-weaker}
Under Assumption~\ref{ass:regularity}(i), Assumption~\ref{ass:full-consistency} implies
Assumption~\ref{ass:consistency}. Moreover, the implication is strict whenever
$\mathbb P(q^{\mathrm{lo}}(X)<C^{\mathrm L})>0$ or $\mathbb P(q^{\mathrm{up}}(X)>C^{\mathrm R})>0$.
\end{proposition}

\begin{lemma}\label{lem:uncensored}
Under Assumption~\ref{ass:regularity}, suppose that $\pi^{\mathrm L} = 0 = \pi^{\mathrm R}$, and that the base estimates take values in $[Y^{\min},Y^{\max}]$.
Then:
\begin{enumerate}
\item[(i)] $\tau^{\mathrm{m}}=0$ and $\mathcal O_{\mathrm{m}}=\mathcal O_0$;
\item[(ii)] Assumption~\ref{ass:full-consistency} and Assumption~\ref{ass:consistency} are equivalent.
\end{enumerate}
\end{lemma}

\begin{lemma}\label{lem:score-deficit}
Under Assumption~\ref{ass:regularity}, with $s^{\mathrm{pe},{\mathrm{m}}}(x,\tilde y):=\max\{\Pi q^{\mathrm{lo}}(x)-\tilde y,\,\tilde y-\Pi q^{\mathrm{up}}(x)\}$
\[
  \mathbb P\bigl(s^{\mathrm{pe},{\mathrm{m}}}(X,\tilde Y)<0\bigr) \le 1-\alpha,
\]
with strict inequality if and only if $\mathbb{P}(q^{\mathrm{lo}}(X) \le C^{\mathrm L}) + \mathbb{P}(q^{\mathrm{up}}(X) \ge C^{\mathrm R})  > 0$.
\end{lemma}

\begin{lemma}\label{lem:tau-nonneg} Assume $(X_i, Y_i)_{i=1}^{n+1}$ are i.i.d. Let $s^{\mathrm{pe},{\mathrm{m}}}(x,\tilde y):=\max\{\Pi q^{\mathrm{lo}}(x)-\tilde y,\,\tilde y-\Pi q^{\mathrm{up}}(x)\}$, $s^{\mathrm{pe}}_n(x, \tilde y) = \max\{\Pi \hat{f}^{\mathrm{lo}}_{n}(x) - \tilde y,\; \tilde y - \Pi \hat{f}^{\mathrm{up}}(x)\}$, and let $\tilde\tau^{\mathrm{pe}}_n$ denote the $\lceil (n+1)(1-\alpha)\rceil$-th smallest value in $\{s^{\mathrm{pe}}_n(X_i,\tilde Y_i)\}_{i=1}^{n} \cup \{R\}$. 
With $\pi^{\mathrm L} := \mathbb{P}(q^{\mathrm{lo}}(X) \le C^{\mathrm L})$ and
$\pi^{\mathrm R} := \mathbb{P}(q^{\mathrm{up}}(X) \ge C^{\mathrm R})$, assume $\pi^{\mathrm L} + \pi^{\mathrm R} > 0$.
Then, under Assumptions~\ref{ass:regularity} and~\ref{ass:consistency},  $\mathbb P(\tilde\tau^{\mathrm{pe}}_n(\alpha)<0)\to0$ as $n \to \infty$.
\end{lemma}

\subsection{Proofs for Section~\ref{sec:eCQR}}

% \begin{proof}[of Theorem~\ref{thm:cqr-validity}]
% By the definition of the widening operator \eqref{eq:widen-op} and the
% score function \eqref{eq:cqr-score},
% \[
%     Y_{n+1} \in \widehat{\mathcal{C}}(X_{n+1})
%     \quad \iff \quad
%     s^{\mathrm{CQR}}(X_{n+1}, Y_{n+1}) \leq \tilde\tau(\alpha).
% \]
% The bounds \eqref{eq:cqr-coverage-lb}–\eqref{eq:cqr-coverage-ub} then follow
% from the standard split-conformal argument applied to the exchangeable scores
% $\{s^{\mathrm{CQR}}(X_i, Y_i)\}_{i=1}^{n+1}$ \citep{romano2019cqr}: with
% $k = \lceil(n+1)(1-\alpha)\rceil$, then $\mathbb{P}(s^{\mathrm{CQR}}(X_{n+1}, Y_{n+1}) \leq \tilde\tau(\alpha))
% \geq k/(n+1) \geq 1 - \alpha$, with the matching upper bound $k/(n+1) \leq
% 1 - \alpha + 1/(n+1)$ when the scores are almost-surely distinct.
% \end{proof}

\begin{proof}[of Theorem~\ref{thm:ecqr-validity}]
Define $\tilde{L}(X_{n+1})$ and $\tilde{U}(X_{n+1})$ such that
\begin{align*}
\tilde{\mathcal{C}}(X_{n+1}) & = [\tilde{L}(X_{n+1}),\; \tilde{U}(X_{n+1})], \\
 \widehat{\mathcal{C}}^{\mathrm{eCQR}}(X_{n+1}) & = \psi\bigl(\tilde{L}(X_{n+1}),\; \tilde{U}(X_{n+1})\bigr),
\end{align*}
and the events
\begin{align*}
    \tilde{E} \;&:=\; \bigl\{\tilde{Y}_{n+1} \in [\tilde{L}(X_{n+1}),\; \tilde{U}(X_{n+1})]\bigr\}, \\
    E         \;&:=\; \bigl\{Y_{n+1}  \in \psi\bigl(\tilde{L}(X_{n+1}),\; \tilde{U}(X_{n+1})\bigr) \bigr\},
\end{align*}
and abbreviate $\tilde{L} = \tilde{L}(X_{n+1})$,
$\tilde{U} = \tilde{U}(X_{n+1})$.

\textit{Step 1: $\tilde{E} \subseteq E$.}
Assume $\tilde{E}$ holds, so $\tilde{L} \leq \tilde{U}$.
\begin{enumerate}[label=(\roman*)]
\item If $Y_{n+1} \in [C^{\mathrm{L}}, C^{\mathrm{R}}]$, then $\tilde{Y}_{n+1} = Y_{n+1} \in
[\tilde{L}, \tilde{U}] \cap [C^{\mathrm{L}}, C^{\mathrm{R}}]$. Inspection of \eqref{eq:psi-def} shows
$\psi(\tilde{L}, \tilde{U}) \supseteq [\tilde{L}, \tilde{U}] \cap [C^{\mathrm{L}}, C^{\mathrm{R}}]$ in every
case, so $Y_{n+1} \in \psi(\tilde{L}, \tilde{U})$ and $E$ holds.
 
\item If $Y_{n+1} < C^{\mathrm{L}}$, then $\tilde{Y}_{n+1} = C^{\mathrm{L}}$, so $\tilde{L} \leq C^{\mathrm{L}}
\leq \tilde{U}$. By \eqref{eq:psi-def}, $\tilde{L} \leq C^{\mathrm{L}}$ sends the lower endpoint
of $\psi(\tilde{L}, \tilde{U})$ to $Y^{\min}$, while its upper endpoint is at least
$C^{\mathrm{L}} > Y_{n+1}$; hence $Y_{n+1} \in \psi(\tilde{L}, \tilde{U})$ and $E$ holds.
 
\item If $Y_{n+1} > C^{\mathrm{R}}$, the argument is symmetric.
\end{enumerate}

\textit{Step 2: decomposition.}
Since $\tilde{E} \subseteq E$, $\mathbb{P}(E, \tilde{E}) = \mathbb{P}(\tilde{E})$ and
\begin{equation}\label{eq:proof-decomp}
    \mathbb{P}(E) \;=\; \mathbb{P}(\tilde{E}) + \mathbb{P}(E, \tilde{E}^c).
\end{equation}

\textit{Step 3: computing $\mathbb{P}(E, \tilde{E}^c)$.}
Since both $[\tilde{L}, \tilde{U}]$ and
$\psi(\tilde{L}, \tilde{U})$ are empty when $\tilde{L} > \tilde{U}$, we have
$E \subseteq \{\tilde{L} \leq \tilde{U}\}$ and $\tilde{E} \subseteq \{\tilde{L} \leq \tilde{U}\}$;
in particular $\{E, \tilde{E}^c\} \subseteq \{\tilde{L} \leq \tilde{U}\}$, so we may
restrict attention to nonempty intervals throughout.
Consider the three possible cases depending on the position of $Y_{n+1}$, partitioning
the support into $(C^{\mathrm{L}}, C^{\mathrm{R}})$, $\{Y_{n+1} \leq C^{\mathrm{L}}\}$, and $\{Y_{n+1} \geq C^{\mathrm{R}}\}$.
\begin{enumerate}[label=(\roman*)]
\item Suppose $Y_{n+1} \in (C^{\mathrm{L}}, C^{\mathrm{R}})$, so $\tilde{Y}_{n+1} = Y_{n+1}$ and, on
$\tilde{E}^c$, $Y_{n+1} \notin [\tilde{L}, \tilde{U}]$. If $Y_{n+1} < \tilde{L}$ then
$C^{\mathrm{L}} < Y_{n+1} < \tilde{L}$, so $\tilde{L} > C^{\mathrm{L}}$ and, by \eqref{eq:psi-def}, the
lower endpoint of $\psi(\tilde{L}, \tilde{U})$ is $\min\{C^{\mathrm{R}}, \tilde{L}\}$, which
exceeds $Y_{n+1}$ since $Y_{n+1} < \tilde{L}$ and $Y_{n+1} < C^{\mathrm{R}}$; thus
$Y_{n+1} \notin \psi(\tilde{L}, \tilde{U})$ and $E$ fails. The case
$Y_{n+1} > \tilde{U}$ is symmetric. Hence
$\{E, \tilde{E}^c, Y_{n+1} \in (C^{\mathrm{L}}, C^{\mathrm{R}})\} = \emptyset$.

\item Suppose $Y_{n+1} \leq C^{\mathrm{L}}$, so $\tilde{Y}_{n+1} = C^{\mathrm{L}}$. On $\tilde{E}^c$,
$C^{\mathrm{L}} \notin [\tilde{L}, \tilde{U}]$, so $\tilde{U} < C^{\mathrm{L}}$ or $\tilde{L} > C^{\mathrm{L}}$. If
$\tilde{L} > C^{\mathrm{L}}$, the lower endpoint of $\psi(\tilde{L}, \tilde{U})$ is at least
$\min\{C^{\mathrm{R}}, \tilde{L}\} > C^{\mathrm{L}} \geq Y_{n+1}$, so $E$ fails. Hence $\tilde{U} < C^{\mathrm{L}}$, in
which case $\tilde{L} \leq \tilde{U} < C^{\mathrm{L}} < C^{\mathrm{R}}$ and \eqref{eq:psi-def} gives
$\psi(\tilde{L}, \tilde{U}) = [Y^{\min}, C^{\mathrm{L}}]$, so $E$ holds for every
$Y_{n+1} \leq C^{\mathrm{L}}$. Conversely, any point with $\tilde{L} \leq \tilde{U} < C^{\mathrm{L}}$ and
$Y_{n+1} \leq C^{\mathrm{L}}$ has $\tilde{Y}_{n+1} = C^{\mathrm{L}} \notin [\tilde{L}, \tilde{U}]$ (so
$\tilde{E}^c$) and lies in $[Y^{\min}, C^{\mathrm{L}}]$ (so $E$). Therefore
\[
    \{E, \tilde{E}^c, Y_{n+1} \leq C^{\mathrm{L}}\}
    \;=\; \{\tilde{L} \leq \tilde{U} < C^{\mathrm{L}},\; Y_{n+1} \leq C^{\mathrm{L}}\}.
\]

\item Suppose $Y_{n+1} \geq C^{\mathrm{R}}$. Symmetric to (ii):
\[
    \{E, \tilde{E}^c, Y_{n+1} \geq C^{\mathrm{R}}\}
    \;=\; \{C^{\mathrm{R}} < \tilde{L} \leq \tilde{U},\; Y_{n+1} \geq C^{\mathrm{R}}\}.
\]
\end{enumerate}
The three cases are disjoint and exhaust the support, so
\[
    \mathbb{P}(E, \tilde{E}^c)
    \;=\; \underbrace{\mathbb{P}\bigl(\tilde{L} \leq \tilde{U} < C^{\mathrm{L}},\; Y_{n+1} \leq C^{\mathrm{L}}\bigr)
          + \mathbb{P}\bigl(C^{\mathrm{R}} < \tilde{L} \leq \tilde{U},\; Y_{n+1} \geq C^{\mathrm{R}}\bigr)}_{=:\, p_{\mathrm{slack}}}.
\]
and $\mathbb{P}(E) = \mathbb{P}(\tilde{E}) + p_{\mathrm{slack}}$. Finally, the inequality follows from Theorem~\ref{thm:cqr-validity}.
\end{proof}

\subsection{Proofs for Section~\ref{sec:method-clipcqr}}

\begin{proof}[of Lemma~\ref{lemma:score-formula}]
Fix $x$ and abbreviate $a := \hat{f}^{\mathrm{lo}}(x)$ and $b := \hat{f}^{\mathrm{up}}(x)$,
so $a \leq b$. By~\eqref{eq:phi-def}, $\phi(x;\tau) = \psi(a - \tau,\, b + \tau)$,
whose raw endpoints $\ell(\tau) = a - \tau$ and $u(\tau) = b + \tau$ are
respectively decreasing and increasing in $\tau$; the family is therefore nested,
and $s(x,y)$ is the smallest $\tau$ with $y \in \phi(x;\tau)$. Note that
$\ell(\tau) \leq u(\tau)$ iff $\tau \geq - \tfrac{1}{2}(b - a)$; for
$\tau < - \tfrac{1}{2}(b - a)$ the last case of~\eqref{eq:psi-def} gives
$\phi(x;\tau) = \emptyset$, so no $y$ is captured.

We consider three separate cases based on the position of $y$ relative to $[C^{\mathrm{L}},C^{\mathrm{R}}]$.
\begin{enumerate}[label=(\roman*)]
\item \emph{Uncensored range, $C^{\mathrm{L}} < y < C^{\mathrm{R}}$.} By~\eqref{eq:psi-def}, $y \in \psi(\ell, u) \iff \ell \leq y
\leq u$, for any pair $\ell \leq u$. Hence
$y \in \phi(x;\tau) \iff \tau \geq \max\{a - y,\, y - b\}$. Since
$\max\{a - y,\, y - b\} \geq \tfrac{1}{2}\bigl((a - y) + (y - b)\bigr) =
\tfrac{1}{2}(a - b)$, the non-emptiness constraint is automatically met, and so in this case
\begin{align*}
  s(x,y)
  & = \inf\bigl\{\tau \in \mathbb{R} : \tau \geq \max\{a - y,\, y - b\}  \geq \tfrac{1}{2}(a - b) \bigr\} = \max\{a - y,\, y - b\}.
\end{align*}
 
\item \emph{Left-censored range, $y \leq C^{\mathrm{L}}$.} Such a $y$ belongs to $\psi(\ell, u)$
(with $\ell \leq u$) if and only if its left endpoint is pushed to $Y^{\min}$ (i.e.\ if $\ell \leq C^{\mathrm{L}}$) and $u \geq \ell$.
Therefore, $y \in \phi(x;\tau) \iff a - \tau \leq C^{\mathrm{L}}$ and $\tau \geq \tfrac{1}{2}(a - b)$,
i.e.\ $\tau \geq \max\{a - C^{\mathrm{L}},\, \tfrac{1}{2}(a - b)\}$, giving
\begin{align*}
  s(x,y)
  & = \inf\bigl\{\tau \in \mathbb{R} : \tau \geq \max\{a - C^{\mathrm{L}},\, \tfrac{1}{2}(a - b)\} \bigr\} = \max\{a - C^{\mathrm{L}},\, \tfrac{1}{2}(a - b)\}.
\end{align*}
 
\item \emph{Right-censored range, $y \geq C^{\mathrm{R}}$.} Symmetrically, 
\begin{align*}
  s(x,y)
  & = \inf\bigl\{\tau \in \mathbb{R} : \tau \geq \max\{C^{\mathrm{R}} - b, \, \tfrac{1}{2}(a - b)\} \bigr\} = \max\{C^{\mathrm{R}}-b,\, \tfrac{1}{2}(a - b)\}.
\end{align*}
\end{enumerate}

Substituting $a = \hat{f}^{\mathrm{lo}}(x)$ and $b = \hat{f}^{\mathrm{up}}(x)$
yields~\eqref{eq:score-formula}.
\end{proof}

\begin{proof}[of Theorem~\ref{thm:ccqr-validity}]
To simplify notation, assume the random jitters $\eta_1, \dots, \eta_{n+1}$
have been applied to the base predictor and are thus implicitly absorbed into
the scores $S_i = s(X_i, \tilde{Y}_i)$ for $i \in [n]$ and
$S_{n+1} = s(X_{n+1}, Y_{n+1})$; since jittering sends $S_{i} \to S_{i} - \eta_i$,
the scores are almost surely distinct when $\delta > 0$.

Since $\phi(x; \cdot)$ is nested and closed-valued, the infimum defining
$s(x, \cdot)$ is attained, so $Y_{n+1} \in \phi(X_{n+1}; \tau)$ if and only if
$S_{n+1} = s(X_{n+1}, Y_{n+1}) \leq \tau$. Hence
\[
    \mathbb{P}\bigl[Y_{n+1} \in \widehat{\mathcal{C}}(X_{n+1})\bigr]
    = \mathbb{P}\bigl(s(X_{n+1}, Y_{n+1}) \leq \hat{\tau}(\alpha)\bigr).
\]
Next, we show $S_1, \ldots, S_{n+1}$ are exchangeable. This is not immediate, because the
calibration scores are evaluated at the observed responses,
$S_i = s(X_i, \tilde{Y}_i)$, whereas the test score is evaluated at the latent
response, $S_{n+1} = s(X_{n+1}, Y_{n+1})$, and in general $\tilde{Y} \neq Y$.

However, the key property of the ClipCQR score is that it depends on the response only
through its clipped (observed) value:
\begin{equation}\label{eq:score-agreement}
    s(x, y) \;=\; s\bigl(x,\, \Pi(y)\bigr),
    \qquad \forall\,(x, y) \in \mathcal{X} \times [Y^{\min}, Y^{\max}].
\end{equation}
This is easy to prove because $s(x, \cdot)$ is constant on each
censored region $\{y \leq C^{\mathrm{L}}\}$ and $\{y \geq C^{\mathrm{R}}\}$, while
$\Pi(y)$ sends every $y < C^{\mathrm{L}}$ to $C^{\mathrm{L}}$ and every
$y > C^{\mathrm{R}}$ to $C^{\mathrm{R}}$; the equality holds trivially for every $y \in [C^{\mathrm{L}}, C^{\mathrm{R}}]$.
Consequently $S_i = s(X_i, \tilde{Y}_i) = s(X_i, Y_i)$ almost surely for all
$i \in [n]$, so all $n+1$ scores coincide almost surely with the latent scores
$s(X_j, Y_j)$, $j \in [n+1]$. As these are a fixed function of the exchangeable
tuples $(X_j, Y_j, \eta_j)$, the scores $S_1, \ldots, S_{n+1}$ are exchangeable,
and the standard lower and upper coverage bounds for split-conformal prediction
follow, as in the proof of Theorem~\ref{thm:cqr-validity}; e.g., see \citet{romano2019cqr}.
\end{proof}

\begin{proof}[of Theorem~\ref{thm:ccqr_plus-validity}]
On the event $\{X_{n+1} \in\mathrm{int}(\hat f)\}$, we have $\hat\tau^+(\alpha, X_{n+1}) = \hat\tau(\alpha;\widehat{\mathcal I})$, 
and $(X_{n+1},Y_{n+1})$ is exchangeable with the calibration data points in $\widehat{\mathcal I}$. Therefore, $\mathbb{P}\bigl[Y_{n+1} \in \widehat{\mathcal{C}}^+(X_{n+1}) \mid X_{n+1} \in \mathrm{int}(\hat{f}) \bigr] \geq 1 - \alpha$ by the usual argument (e.g., Theorem~\ref{thm:cqr-validity}).
On the event $\{X_{n+1} \notin\mathrm{int}(\hat f)\}$, we have $\hat\tau^+(\alpha, X_{n+1}) \geq \hat\tau(\alpha;\widehat{\mathcal I}^{\mathrm{c}})$, and therefore $\mathbb{P}\bigl[Y_{n+1} \in \widehat{\mathcal{C}}^+(X_{n+1}) \mid X_{n+1} \notin \mathrm{int}(\hat{f}) \bigr] \geq 1 - \alpha$ by the same argument. Marginalizing leads to $\mathbb{P}\bigl[Y_{n+1} \in \widehat{\mathcal{C}}^+(X_{n+1}) \bigr] \geq 1-\alpha$.
\end{proof}

\subsection{Proofs for Section~\ref{sec:theory}}

\begin{proof}[of Proposition~\ref{prop:O0-coverage}]
This is implied by Proposition~\ref{prop:oracle-coverage}, stated in Appendix~\ref{app:theory-clipcqr}.
\end{proof}

\begin{proof}[of Theorem~\ref{thm:ccqr-plus-oracle}]
By Lemma~\ref{lem:snapped-from-projected} we can use Assumption~\ref{ass:snapped-consistency}, which follows from Assumptions~\ref{ass:regularity}, \ref{ass:consistency}, and~\ref{ass:base-regularity}.
Note that $\underline\psi(\cdot)$ and $\overline\psi(\cdot)$ are both nondecreasing and $1$-Lipschitz, apart from a single jump ($\underline\psi$ at $C^{\mathrm L}$ and $\overline\psi$ at $C^{\mathrm R}$). 
This representation of $\psi(\ell,u)$ is useful to bound the symmetric difference:
\[
  \lambda \bigl( \phi_n(X;0)\triangle\mathcal O_0(X) \bigr)
  \le|\underline\psi(\hat f^{\mathrm{lo}}_n)-\underline\psi(q^{\mathrm{lo}})| +|\overline\psi(\hat f^{\mathrm{up}}_n)-\overline\psi(q^{\mathrm{up}})|,
\]
which by Assumption~\ref{ass:snapped-consistency} implies $\mathbb{E} [ \lambda ( \phi_n(X;0)\triangle\mathcal O_0(X) ) ] \to 0$.

Write $\phi_n(\cdot;t) := \psi(\hat f^{\mathrm{lo}}_n(\cdot)-t,\hat f^{\mathrm{up}}_n(\cdot)+t)$
for the estimated prediction band at level $t$, so that
$\widehat{\mathcal C}^+_n(X_{n+1}) =\phi_n(X_{n+1};\hat\tau^+(\alpha,X_{n+1}))$ with
$\hat\tau^+(\alpha,X_{n+1})$ given by~\eqref{eq:-threshold-clipcqrt+},
and $Y\in\phi_n(X;t)\iff s_n(X,Y)\le t$, where $s_n$ is the non-conformity score function~\eqref{eq:score-formula}.

\begin{enumerate}
\item[(i-a)] We begin by proving $\hat\tau^+(\alpha, X_{n+1})\xrightarrow{p}0$.

Let $\mathcal U:=\{x \in \mathcal{X}: C^{\mathrm L} < q^{\mathrm{lo}}(x) \le q^{\mathrm{up}}(x) < C^{\mathrm R}\}$ be the population-level internal region,
and let $\mathcal U^{\mathrm{c}} := \mathcal{X} \setminus \mathcal U$ denote its complement.
By assumption, $\mathcal U$ carries positive mass $\pi_{\mathrm{int}}:=\mathbb P(X\in\mathcal U)>0$.
Let $\pi_{\mathrm{ext}} := 1 - \pi_{\mathrm{int}}$ denote the probability mass of $\mathcal U^{\mathrm{c}}$.

\emph{1) Convergence of $\hat\tau(\alpha; \widehat{\mathcal I}) \xrightarrow{p} 0$.}
Let $\hat\tau(\alpha; \widehat{\mathcal I})$ be the ClipCQR threshold calibrated on the $n_{\mathrm{int}}:=|\{i:X_i\in\mathcal U\}|$ points in $\mathcal U$. Since $n_{\mathrm{int}}/n\to\pi_{\mathrm{int}}>0$ a.s., $n_{\mathrm{int}}\to\infty$ and $\hat\tau(\alpha; \widehat{\mathcal I})$ is a conditional conformal quantile for the distribution of $(X,Y)$ given $X\in\mathcal U$. Define the conditional oracle threshold
\[
  \tau^{\mathrm{m}}_{\mathrm{int}}:=\inf\bigl\{t\in\mathbb R:\ \mathbb P\bigl(Y\in\phi^{\mathrm{m}}(X;t)\mid X\in\mathcal U\bigr)\ge1-\alpha\bigr\},
\]
the analogue of $\tau^{\mathrm{m}}$ for the distribution of $(X,Y)$ given $X\in\mathcal U$. For $x\in\mathcal U$ the oracle band is unsnapped, $\phi^\star(x;0)=\mathcal O_0(x)=[q^{\mathrm{lo}}(x),q^{\mathrm{up}}(x)]$, and $\mathbb P(Y\in\mathcal O_0(X)\mid X\in\mathcal U)=1-\alpha$ (Proposition~\ref{prop:oracle-coverage}); since $t\mapsto\mathbb P(Y\in\phi^\star(X;t)\mid X\in\mathcal U)$ increases strictly through $1-\alpha$ at $t=0$ (Assumption~\ref{ass:regularity}(ii)), it follows that $\tau^{\mathrm{m}}_{\mathrm{int}}=0$. As there is no censoring within $\mathcal U$, Assumption~\ref{ass:snapped-consistency} restricted to $\mathcal U$ coincides with 
$L^1$-consistency on $\mathcal U$ (Assumption~\ref{ass:full-consistency}) by Lemma~\ref{lem:uncensored}; applying Theorem~\ref{thm:ccqr-oracle}(i) to the distribution of $(X,Y)$ given $X\in\mathcal U$ gives $\hat\tau(\alpha; \widehat{\mathcal I})\xrightarrow{p}\tau^{\mathrm{m}}_{\mathrm{int}} = 0$.

\emph{2) Convergence of $\max\{0,\hat\tau(\alpha;\widehat{\mathcal I}^{\mathrm c})\}\xrightarrow{p}0$ under $\pi_{\mathrm{ext}}>0$.}
Unlike the internal case, conditioning on $\mathcal U^{\mathrm c}$ violates Assumption~\ref{ass:non-extreme-cens}
($\mathbb P(X\in\mathcal U\mid X\in\mathcal U^{\mathrm c})=0$), so Theorem~\ref{thm:ccqr-oracle}(i) cannot be invoked.
The truncation $\max\{0,\cdot\}$ means it suffices to bound the threshold from above, which we obtain from the snapped
oracle strictly over-covering on $\mathcal U^{\mathrm c}$: every $x\in\mathcal U^{\mathrm c}$ satisfies
$q^{\mathrm{lo}}(x)\le C^{\mathrm L}$ or $q^{\mathrm{up}}(x)\ge C^{\mathrm R}$, so Proposition~\ref{prop:oracle-coverage}
gives $\mathbb P(Y\in\mathcal O_0(X)\mid X=x)\ge1-\tfrac\alpha2$ in every case, and
\begin{equation}\label{eq:ext-overcover}
  G^{\mathrm{c}}:=\mathbb P\bigl(Y\in\mathcal O_0(X)\mid X\in\mathcal U^{\mathrm c}\bigr)\ \ge\ 1-\tfrac\alpha2\ >\ 1-\alpha.
\end{equation}

Consider the external coverage of the estimated band at level $0$, written as a ratio:
\[
  G_n^{\mathrm c}:=\mathbb P\bigl(Y\in\phi_n(X;0)\mid\hat f_n,\,X\in\mathrm{int}(\hat f_n)^{\mathrm c}\bigr)
  =\frac{N_n}{D_n},
\]
with $N_n:=\mathbb P\bigl(Y\in\phi_n(X;0),\,X\in\mathrm{int}(\hat f_n)^{\mathrm c}\mid\hat f_n\bigr)$ and
$D_n:=\mathbb P\bigl(X\in\mathrm{int}(\hat f_n)^{\mathrm c}\mid\hat f_n\bigr)$.
We show $N_n\xrightarrow{p}\mathbb P(Y\in\mathcal O_0(X),\,X\in\mathcal U^{\mathrm c})$ and $D_n\xrightarrow{p}\pi_{\mathrm{ext}}$
via two mismatch bounds, hence $\xrightarrow{p}0$:
\begin{itemize}
\item \emph{Band mismatch.}
$\mathbb P\bigl(Y\in\phi_n(X;0)\triangle\mathcal O_0(X)\mid\hat f_n\bigr)\le M\,\mathbb E[\lambda(\phi_n(X;0)\triangle\mathcal O_0(X))\mid\hat f_n]$
by the density bound (Assumption~\ref{ass:regularity}(ii)), and $\mathbb E[\lambda(\phi_n(X;0)\triangle\mathcal O_0(X))]\to0$ was shown above.
\item \emph{Region mismatch.}
$\mathbb P\bigl(\mathrm{int}(\hat f_n)\triangle\mathcal U\mid\hat f_n\bigr)\to0$. Since
$\I\{x\in\mathcal U\}=\I\{q^{\mathrm{lo}}(x)>C^{\mathrm L}\}\,\I\{q^{\mathrm{up}}(x)<C^{\mathrm R}\}$ and likewise for
$\mathrm{int}(\hat f_n)$, we have $\mathrm{int}(\hat f_n)\triangle\mathcal U\subseteq A_n^{\mathrm L}\cup A_n^{\mathrm R}$ with
$A_n^{\mathrm L}:=\{\I\{\hat f^{\mathrm{lo}}_n(X)>C^{\mathrm L}\}\ne\I\{q^{\mathrm{lo}}(X)>C^{\mathrm L}\}\}$ and
$A_n^{\mathrm R}$ defined symmetrically. On $A_n^{\mathrm L}$ the values $\hat f^{\mathrm{lo}}_n(X)$ and $q^{\mathrm{lo}}(X)$
straddle $C^{\mathrm L}$, so for any $\eta\in(0,C^{\mathrm R}-C^{\mathrm L})$,
\[
  A_n^{\mathrm L}\subseteq\{|q^{\mathrm{lo}}(X)-C^{\mathrm L}|\le\eta\}
   \cup\{|\hat f^{\mathrm{lo}}_n(X)-C^{\mathrm L}|\le\eta\}
   \cup\{|\Pi\hat f^{\mathrm{lo}}_n(X)-\Pi q^{\mathrm{lo}}(X)|>\eta\}.
\]
The last set has probability $\le\eta^{-1}\mathbb E|\Pi\hat f^{\mathrm{lo}}_n(X)-\Pi q^{\mathrm{lo}}(X)|\to0$
(Assumption~\ref{ass:consistency}); taking $\limsup_n$ then $\eta\downarrow0$, the first vanishes by
Assumption~\ref{ass:regularity}(i) and the second by Assumption~\ref{ass:base-regularity}. Thus
$\mathbb P(A_n^{\mathrm L})\to0$, and symmetrically for $A_n^{\mathrm R}$.
\end{itemize}
The region mismatch bound gives $|D_n-\pi_{\mathrm{ext}}|\xrightarrow{p}0$ and, combined with the band mismatch,
$|N_n-\mathbb P(Y\in\mathcal O_0(X),\,X\in\mathcal U^{\mathrm c})|\xrightarrow{p}0$. Since the denominator limit
$\pi_{\mathrm{ext}}>0$ is strictly positive, this implies
$G_n^{\mathrm c}\xrightarrow{p}G^{\mathrm c}\ge 1-\tfrac\alpha2$ by~\eqref{eq:ext-overcover}.

For $t\in\mathbb R$, let $\widehat{G}_n^{\mathrm c}(t):=n_{\mathrm{ext}}^{-1}\sum_{i\in\widehat{\mathcal I}^{\mathrm c}}
\mathbf 1\{s_n(X_i,Y_i)\le t\}$ with $n_{\mathrm{ext}}=|\widehat{\mathcal I}^{\mathrm c}|$, so that
$\hat\tau(\alpha;\widehat{\mathcal I}^{\mathrm c})=\inf\{t:\widehat{G}_n^{\mathrm c}(t)\ge p_{n,\mathrm{ext}}\}$ with
$p_{n,\mathrm{ext}}=\lceil(n_{\mathrm{ext}}+1)(1-\alpha)\rceil/n_{\mathrm{ext}}\to1-\alpha$. Conditional on $\hat f_n$
and the membership pattern, $\widehat{G}_n^{\mathrm c}(0)$ is an average of $n_{\mathrm{ext}}$ i.i.d.\
Bernoulli$(G_n^{\mathrm c})$ variables, so
$\mathbb E[(\widehat{G}_n^{\mathrm c}(0)-G_n^{\mathrm c})^2\mid\hat f_n]\le1/(4n_{\mathrm{ext}})$. Since
$n_{\mathrm{ext}}\to\infty$ a.s.\ (as $\pi_{\mathrm{ext}}>0$), this gives
$\widehat{G}_n^{\mathrm c}(0)-G_n^{\mathrm c}\xrightarrow{p}0$, and combined with
$G_n^{\mathrm c}\xrightarrow{p}G^{\mathrm c}$ above,
$\widehat{G}_n^{\mathrm c}(0)\xrightarrow{p}G^{\mathrm c}>1-\alpha=\lim p_{n,\mathrm{ext}}$. Hence
$\mathbb P[\widehat{G}_n^{\mathrm c}(0)\ge p_{n,\mathrm{ext}}]\to1$, and on that event
$\hat\tau(\alpha;\widehat{\mathcal I}^{\mathrm c})\le0$. As $\max\{0,\hat\tau(\alpha;\widehat{\mathcal I}^{\mathrm c})\}\ge0$, under $\pi_{\mathrm{ext}}>0$,
\[
  \max\{0,\hat\tau(\alpha;\widehat{\mathcal I}^{\mathrm c})\}\ \xrightarrow{p}\ 0.
\]

\emph{3) Combination.} 
The threshold applied to $X_{n+1}$ is
\[
  \hat\tau^+(\alpha, X_{n+1})
  =\hat\tau(\alpha;\widehat{\mathcal I})\,\I\{X_{n+1}\in\mathrm{int}(\hat f)\}
   +\max\{0,\hat\tau(\alpha;\widehat{\mathcal I}^{\mathrm c})\}\,\I\{X_{n+1}\notin\mathrm{int}(\hat f)\},
\]
If $\pi_{\mathrm{ext}}>0$, then $\hat\tau(\alpha; \widehat{\mathcal I}) \xrightarrow{p} 0$ and $\max\{0,\hat\tau(\alpha;\widehat{\mathcal I}^{\mathrm c})\} \xrightarrow{p}0$, as shown above.
Otherwise, if $\pi_{\mathrm{ext}}=0$, then $\hat\tau(\alpha; \widehat{\mathcal I}) \xrightarrow{p} 0$ and the external branch is asymptotically vacuous:
$\mathbb P(X_{n+1}\notin\mathrm{int}(\hat f))\to\pi_{\mathrm{ext}}=0$ (by the region mismatch bound), so the
indicator $\I\{X_{n+1}\notin\mathrm{int}(\hat f)\}\xrightarrow{p}0$ and the truncation
$\max\{0,\hat\tau(\alpha;\widehat{\mathcal I}^{\mathrm c})\}\le R$ is multiplied by a vanishing indicator; thus $\hat\tau^+(\alpha, X_{n+1}) \xrightarrow{p}0$ in either case.

\item[(i-b)] We now prove $\widehat{\mathcal{C}}^+_n \to \mathcal{O}_0$. Write
$\hat\tau_n^+(x):=\hat\tau^+(\alpha,x)$ for the applied threshold and set
\[
  \Delta_n:=|\hat\tau(\alpha;\widehat{\mathcal I})|
   +\max\{0,\hat\tau(\alpha;\widehat{\mathcal I}^{\mathrm c})\}\,\I\{X_{n+1}\notin\mathrm{int}(\hat f)\}.
\]
Since $\hat\tau_n^+(x)$ equals one of its two summands according to whether $x\in\mathrm{int}(\hat f)$,
$\sup_x|\hat\tau_n^+(x)|\le\Delta_n$, and $\Delta_n\le R$. Moreover $\Delta_n\xrightarrow{p}0$: the first summand
vanishes by (i-a), and the second does too — if $\pi_{\mathrm{ext}}>0$ because
$\max\{0,\hat\tau(\alpha;\widehat{\mathcal I}^{\mathrm c})\}\xrightarrow{p}0$, and if $\pi_{\mathrm{ext}}=0$ because the
bounded factor is multiplied by $\I\{X_{n+1}\notin\mathrm{int}(\hat f)\}\xrightarrow{p}0$. Define
\[
    \delta^{\mathrm{lo}}_n := \Delta_n + |\underline\psi(\hat f^{\mathrm{lo}}_n)-\underline\psi(q^{\mathrm{lo}})|,
    \qquad
    \delta^{\mathrm{up}}_n := \Delta_n + |\overline\psi(\hat f^{\mathrm{up}}_n)-\overline\psi(q^{\mathrm{up}})|.
\]
Define the open interval $\mathcal I_n^{\mathrm L}$ between $C^{\mathrm L}$ and $C^{\mathrm L}+\hat\tau_n^+(X)$
(possibly empty, as $\hat\tau_n^+(X)$ may be negative). Note that
\[
    |\underline\psi(\hat f^{\mathrm{lo}}_n-\hat\tau_n^+(X))-\underline\psi(q^{\mathrm{lo}})|
    \le \delta^{\mathrm{lo}}_n + R\,\mathbf 1\{\hat f^{\mathrm{lo}}_n\in\mathcal I_n^{\mathrm L}\},
\]
because $\underline\psi$ is $1$-Lipschitz except for its jump at $C^{\mathrm L}$, and that jump is crossed iff $C^{\mathrm L}$ lies strictly between $\hat f^{\mathrm{lo}}_n$ and $\hat f^{\mathrm{lo}}_n-\hat\tau_n^+(X)$, i.e.\ iff $\hat f^{\mathrm{lo}}_n\in\mathcal I_n^{\mathrm L}$.
Similarly, with $\mathcal I_n^{\mathrm R}$ denoting the corresponding open interval between $C^{\mathrm R}-\hat\tau_n^+(X)$ and $C^{\mathrm R}$,
\[
    |\overline\psi(\hat f^{\mathrm{up}}_n-\hat\tau_n^+(X))-\overline\psi(q^{\mathrm{up}})|
    \le \delta^{\mathrm{up}}_n + R\,\mathbf 1\{\hat f^{\mathrm{up}}_n\in\mathcal I_n^{\mathrm R}\}.
\]
Therefore,
\begin{align*}
    \mathbb E\bigl[\lambda(\widehat{\mathcal C}^+_n(X)\triangle\mathcal O_0(X))\bigr]
    \le \mathbb E[\delta^{\mathrm{lo}}_n] + \mathbb E[\delta^{\mathrm{up}}_n]
       + R\,\mathbb P(\hat f^{\mathrm{lo}}_n\in\mathcal I_n^{\mathrm L})
       + R\,\mathbb P(\hat f^{\mathrm{up}}_n\in\mathcal I_n^{\mathrm R}).
\end{align*}
The first two terms vanish by Assumption~\ref{ass:snapped-consistency} and $\Delta_n\xrightarrow{p}0$ (with $\Delta_n\le R$, so $\mathbb E\Delta_n\to0$).

For the lower straddle term, note that $\mathcal I_n^{\mathrm L}$, the open interval with endpoints
$C^{\mathrm L}$ and $C^{\mathrm L}+\hat\tau_n^+(X)$, has length $|\hat\tau_n^+(X)|\le\Delta_n$, so on $\{\Delta_n\le\xi\}$ it is
contained in $[C^{\mathrm L}-\xi,C^{\mathrm L}+\xi]$. Hence for every $\xi>0$,
\[
  \mathbb P\bigl(\hat f^{\mathrm{lo}}_n\in\mathcal I_n^{\mathrm L}\bigr)
  \le \mathbb P\bigl(\Delta_n>\xi\bigr)+\mathbb P\bigl(|\hat f^{\mathrm{lo}}_n-C^{\mathrm L}|\le\xi\bigr).
\]
Taking $\limsup_n$ (the first term vanishes since $\Delta_n\xrightarrow{p}0$, by (i-a)) and then $\xi\downarrow0$,
\[
  \limsup_n\mathbb P\bigl(\hat f^{\mathrm{lo}}_n\in\mathcal I_n^{\mathrm L}\bigr)
  \le \lim_{\xi\downarrow0}\limsup_n\mathbb P\bigl(|\hat f^{\mathrm{lo}}_n-C^{\mathrm L}|\le\xi\bigr)=0
\]
by Assumption~\ref{ass:base-regularity}; the left side does not depend on $\xi$, so
$\mathbb P(\hat f^{\mathrm{lo}}_n\in\mathcal I_n^{\mathrm L})\to0$. The upper term is symmetric, with $C^{\mathrm R}$,
$\hat f^{\mathrm{up}}_n$, and Assumption~\ref{ass:base-regularity}.

Therefore
\[
    \mathbb E\bigl[\lambda(\widehat{\mathcal C}^+_n(X)\triangle\mathcal O_0(X))\bigr] \to 0,
\]
and by non-negativity $\widehat{\mathcal{C}}^+_n\to\mathcal{O}_0$ in the sense of Definition~\ref{def:band-convergence}.

\item[(ii)] To prove the convergence of marginal coverage:
  \begin{align*}
    \bigl|\mathbb{P}(Y \in \widehat{\mathcal{C}}^+_n(X))
    - \mathbb{P}(Y \in \mathcal{O}_0(X))\bigr|
    & \leq \mathbb{P}(Y \in \widehat{\mathcal{C}}^+_n(X) \triangle \mathcal{O}_0(X)) \\
    & \leq M\,\mathbb{E} \bigl( \lambda( \widehat{\mathcal{C}}^+_n(X) \triangle \mathcal{O}_0(X) ) \bigr) \to 0,
  \end{align*}
where the second inequality uses $f_{Y\mid X} \leq M$ (Assumption~\ref{ass:regularity}(ii)) and the limit is from Step~i-b.

\item[(iii)] To establish asymptotic conditional coverage, let $\mathcal{D}_{n}$ denote the calibration data set of size $n$ and define
$V_n := \mathbb{E}[\lambda(\widehat{\mathcal{C}}^+_n(X) \triangle \mathcal{O}_0(X)) \mid \mathcal{D}_{n}]$; 
then $\mathbb{E}(V_n) \to 0$ (from Step~i-b), so $V_n \xrightarrow{p} 0$ by non-negativity.
On $\{V_n > 0\}$ define $\epsilon'_n := V_n^{1/2}$ and
$\Lambda_n := \{x : \lambda(\widehat{\mathcal{C}}^+_n(x) \triangle
\mathcal{O}_0(x)) \leq \epsilon'_n\}$ (and $\Lambda_n := \mathcal{X}$ otherwise). By the
conditional Markov inequality,
$\mathbb{P}(X \notin \Lambda_n \mid \mathcal{D}_{n}) \leq V_n/\epsilon'_n = V_n^{1/2}
\xrightarrow{p} 0$, so $\mathbb{P}(X \in \Lambda_n \mid \mathcal{D}_{n}) = 1 - o_P(1)$, and on
$\Lambda_n$ the density bound gives
\[
    \sup_{x \in \Lambda_n}
    \bigl|\mathbb{P}(Y \in \widehat{\mathcal{C}}^+_n(x) \mid X=x)
        - \mathbb{P}(Y \in \mathcal{O}_0(x) \mid X=x)\bigr|
    \leq M\epsilon'_n = M V_n^{1/2} \xrightarrow{p} 0.
\]
This is Definition~\ref{def:cond-cov} with limit $\mathcal{O}_0$, whose conditional
coverage is given by Proposition~\ref{prop:oracle-coverage}.
\end{enumerate}
\end{proof}

\subsection{Proofs for Appendix~\ref{sec:app-ecqr}}

\begin{proof}[of Corollary~\ref{cor:pexp-loose-bound}]
Abbreviate $\tilde{L} = \tilde{L}(X_{n+1})$ and $\tilde{U} = \tilde{U}(X_{n+1})$.
By definition,
\begin{align*}
    p_{\mathrm{slack}}
    \;=\; & \mathbb{P}\bigl(\tilde{L} \leq \tilde{U} < C^{\mathrm{L}},\; Y_{n+1} \leq C^{\mathrm{L}}\bigr)
          + \mathbb{P}\bigl(\tilde{U} \geq \tilde{L} > C^{\mathrm{R}},\; Y_{n+1} \geq C^{\mathrm{R}}\bigr) \geq 0.
\end{align*}
Dropping the constraint $\tilde{L} \leq \tilde{U} < C^{\mathrm{L}}$ from
the first summand yields the bound $\mathbb{P}(Y_{n+1} < C^{\mathrm{L}})$; dropping
$\tilde{U} \geq \tilde{L} > C^{\mathrm{R}}$ from the second yields
$\mathbb{P}(Y_{n+1} > C^{\mathrm{R}})$. Together,
\[
    p_{\mathrm{slack}} \leq \mathbb{P}(Y_{n+1} < C^{\mathrm{L}}) + \mathbb{P}(Y_{n+1} > C^{\mathrm{R}}) \;=\; p_{\mathrm{cens}}.
\]
Symmetrically, dropping $Y_{n+1} \leq \tilde{U}$
from the first summand yields $\mathbb{P}(\tilde{L} \leq \tilde{U} < C^{\mathrm{L}})$, and dropping $C^{\mathrm{R}} < \tilde{L}$ from the second
yields $\mathbb{P}(\tilde{U} \geq \tilde{L} > C^{\mathrm{R}})$. Thus,
$p_{\mathrm{slack}} \leq p_{\mathrm{ext}}$.

Taking the minimum of the two bounds gives $p_{\mathrm{slack}} \leq
\min\{p_{\mathrm{cens}}, p_{\mathrm{ext}}\}$, and substituting into
the upper bound of Theorem~\ref{thm:ecqr-validity} gives the second inequality.
\end{proof}

\begin{proof}[of Lemma~\ref{lem:pecqr-no-slack}]
The projected base interval satisfies
$C^{\mathrm{L}} \leq \Pi\hat{f}^{\mathrm{lo}}(X) \leq \Pi\hat{f}^{\mathrm{up}}(X) \leq C^{\mathrm{R}}$
almost surely, so its midpoint
$\hat{m}_{\mathrm{proj}}(X) := (\Pi\hat{f}^{\mathrm{lo}}(X) + \Pi\hat{f}^{\mathrm{up}}(X))/2$
lies in $[C^{\mathrm{L}}, C^{\mathrm{R}}]$. A nonempty raw interval
$[\tilde{L}(X), \tilde{U}(X)]
= [\Pi\hat{f}^{\mathrm{lo}}(X) - \hat\tau,\, \Pi\hat{f}^{\mathrm{up}}(X) + \hat\tau]$
has $\hat\tau \geq \tfrac12(\Pi\hat{f}^{\mathrm{lo}}(X) - \Pi\hat{f}^{\mathrm{up}}(X))$,
and therefore retains this midpoint:
\[
    \tilde{U}(X) = \Pi\hat{f}^{\mathrm{up}}(X) + \hat\tau
        \;\geq\; \hat{m}_{\mathrm{proj}}(X) \;\geq\; C^{\mathrm{L}},
    \qquad
    \tilde{L}(X) = \Pi\hat{f}^{\mathrm{lo}}(X) - \hat\tau
        \leq \hat{m}_{\mathrm{proj}}(X) \leq C^{\mathrm{R}}.
\]
Thus a nonempty raw interval never lies entirely
outside $[C^{\mathrm{L}}, C^{\mathrm{R}}]$, so both events defining $p_{\mathrm{ext}}$,
\[
    \{\tilde{L}(X) \leq \tilde{U}(X) < C^{\mathrm{L}}\}
    \quad\text{and}\quad
    \{C^{\mathrm{R}} < \tilde{L}(X) \leq \tilde{U}(X)\},
\]
have probability zero. Hence $p_{\mathrm{ext}} = 0$, and
Corollary~\ref{cor:pexp-loose-bound} gives
$p_{\mathrm{slack}} \leq \min\{p_{\mathrm{cens}}, p_{\mathrm{ext}}\} = 0$.
\end{proof}

\begin{proof}[of Theorem~\ref{thm:pecqr-validity}]
The lower bound is from Theorem~\ref{thm:ecqr-validity}, which in this case applies with $p_{\mathrm{slack}} = 0$ as shown in Section~\ref{sec:eCQR}.
For the upper bound, Theorem~\ref{thm:ecqr-validity} with $p_{\mathrm{slack}} = 0$ gives
\begin{align*}
    \mathbb{P}\bigl[Y_{n+1} \in \widehat{\mathcal{C}}^{\mathrm{peCQR}}(X_{n+1})\bigr]
    & = \mathbb{P}\bigl[\tilde{Y}_{n+1} \in \phi^{\mathrm{CQR}}\bigl(X_{n+1};\, \hat\tau^{\mathrm{eCQR}}(\alpha)\bigr)\bigr] \\
    & = \mathbb{P}\bigl(\tilde{S}_{n+1} \leq \hat\tau^{\mathrm{eCQR}}\bigr),
\end{align*}
where $\tilde S_i := s^{\mathrm{CQR}}(X_i,\tilde Y_i)$ is the score computed using the projected
base quantiles, for $i \in [n+1]$, and $\hat\tau^{\mathrm{eCQR}}(\alpha)$ is the $k$-th smallest value
of $\{\tilde S_i\}_{i=1}^n \cup \{R\}$, $k := \lceil(n+1)(1-\alpha)\rceil$.
The standard \emph{quantile inflation lemma}
\citep[e.g.,][Lemma 3.4]{angelopoulos2024conformal} identifies this event with a
comparison against the order statistics of the full score vector:
\[
    \mathbb{P}\bigl(\tilde S_{n+1} \leq \hat\tau^{\mathrm{eCQR}}(\alpha)\bigr)
    = \mathbb{P}\bigl(\tilde S_{n+1} \leq \tilde S'_{(k)}\bigr),
\]
where $\tilde S'_{(k)}$ is the $k$-th order statistic of $\{\tilde S_i\}_{i=1}^{n+1}$.
We decompose
\[
    \mathbb{P}\bigl(\tilde S_{n+1} \leq \tilde S'_{(k)}\bigr)
    = \mathbb{P}\bigl(\tilde S_{n+1} < \tilde S'_{(k)}\bigr)
    + \mathbb{P}\bigl(\tilde S_{n+1} = \tilde S'_{(k)}\bigr).
\]
For the inequality term, exchangeability of $\{\tilde S_i\}_{i=1}^{n+1}$ gives
\citep[Fact 2.15(i)]{angelopoulos2024conformal}
\[
    \mathbb{P}\bigl(\tilde S_{n+1} < \tilde S'_{(k)}\bigr) \leq \frac{k-1}{n+1}.
\]
For the equality term, $\tilde S_{n+1} = \tilde S'_{(k)}$ holds either when
$\tilde S_{n+1}$ is the (tie-free) $k$-th order statistic, an event of probability
at most $\tfrac{1}{n+1}$ by exchangeability, or when $\tilde S_{n+1}$ ties another
score, which by assumption can occur only at the value $0$; hence
\[
    \mathbb{P}\bigl(\tilde S_{n+1} = \tilde S'_{(k)}\bigr)
    \leq \frac{1}{n+1} + \mathbb{P}(\tilde S_{n+1} = 0)
    = \frac{1}{n+1} + p_{\mathrm{zero}}(\hat{f}),
\]
since the event of $p_{\mathrm{zero}}(\hat{f})$ is precisely the event under which scores are zero.
Combining these results, we arrive at:
\[
    \mathbb{P}\bigl(\tilde S_{n+1} \leq \tilde S'_{(k)}\bigr)
    \leq \frac{k-1}{n+1} + \frac{1}{n+1} + p_{\mathrm{zero}}(\hat{f})
    \leq 1 - \alpha + \frac{1}{n+1} + p_{\mathrm{zero}}(\hat{f}).
\]

\end{proof}

\subsection{Proofs for Appendix~\ref{app:theory}} \label{app:proofs-theory-app}

\begin{proof}[of Proposition~\ref{prop:ecqr-ccqr}]
Since both methods predict using a function $\phi(\cdot; \tau)$ in the nested family~\eqref{eq:phi-def}, it suffices to compare how they calibrate $\tau$. 
The eCQR and ClipCQR scores are $\tilde{S}^{\mathrm{CQR}}_i = s^{\mathrm{CQR}}(X_i, \tilde{Y}_i)$ from~\eqref{eq:cqr-score} and $S_i = s(X_i, \tilde{Y}_i)$ from~\eqref{eq:score-formula}, respectively. These can be directly compared, distinguishing three possible cases based on the position of each $\tilde{Y}_i$ for $i \in [n]$.
\begin{enumerate}[label=(\roman*)]
\item If $\tilde{Y}_i \in (C^{\mathrm{L}}, C^{\mathrm{R}})$, the middle case of~\eqref{eq:score-formula}
coincides with~\eqref{eq:cqr-score-closed}, so $S_i = \tilde{S}^{\mathrm{CQR}}_i$.
\item If $\tilde{Y}_i = C^{\mathrm{L}}$, write $A = \hat{f}^{\mathrm{lo}}(X_i) - C^{\mathrm{L}}$ and
$B = C^{\mathrm{L}} - \hat{f}^{\mathrm{up}}(X_i)$, so $\tilde{S}^{\mathrm{CQR}}_i = \max\{A, B\}$;
since $\tfrac12(\hat{f}^{\mathrm{lo}}(X_i) - \hat{f}^{\mathrm{up}}(X_i)) = \tfrac12(A + B)$,
\[
    S_i = \max\bigl\{A,\, \tfrac12(A + B)\bigr\} \leq \max\{A, B\} = \tilde{S}^{\mathrm{CQR}}_i,
\]
with equality iff $B \leq A$, i.e.\ $\hat{m}(X_i) \geq C^{\mathrm{L}}$.
\item the case $\tilde{Y}_i = C^{\mathrm{R}}$ is symmetric: $S_i \leq \tilde{S}^{\mathrm{CQR}}_i$, with equality
iff $\hat{m}(X_i) \leq C^{\mathrm{R}}$.
\end{enumerate}
Hence $S_i \leq \tilde{S}^{\mathrm{CQR}}_i$ for every $i \in [n]$, which implies $\hat\tau(\alpha) \leq \hat\tau^{\mathrm{eCQR}}(\alpha)$ almost surely,
and thus $\widehat{\mathcal{C}}(X_{n+1}) \subseteq \widehat{\mathcal{C}}^{\mathrm{eCQR}}(X_{n+1})$. 
 
On the event $\mathcal{E}$ every calibration midpoint satisfies $\hat{m}(X_i) \in [C^{\mathrm{L}}, C^{\mathrm{R}}]$,
so $S_i = \tilde{S}^{\mathrm{CQR}}_i$ at each $i \in [n]$ and the two prediction sets are equal.
\end{proof}

\begin{proof}[of Corollary~\ref{cor:ecqr-pecqr-ccqr}]
On $\mathcal{E}'$, both base quantiles lie in $[C^{\mathrm{L}}, C^{\mathrm{R}}]$ for every
$i \in [n] \cup \{n+1\}$, so peCQR effectively does not project $\hat{f}$, and 
$\widehat{\mathcal{C}}^{\mathrm{peCQR}}(X_{n+1}) = \widehat{\mathcal{C}}^{\mathrm{eCQR}}(X_{n+1})$ almost surely.
Moreover $\mathcal{E}' \subseteq \mathcal{E}$, since
$C^{\mathrm{L}} < \hat{f}^{\mathrm{lo}}(X_i) \le \hat{f}^{\mathrm{up}}(X_i) < C^{\mathrm{R}}$ forces
$\hat{m}(X_i) \in [C^{\mathrm{L}}, C^{\mathrm{R}}]$, so Proposition~\ref{prop:ecqr-ccqr} gives
$\widehat{\mathcal{C}}(X_{n+1}) = \widehat{\mathcal{C}}^{\mathrm{eCQR}}(X_{n+1})$.
Finally, on $\mathcal{E}'$ every calibration index $i \in [n]$ satisfies $X_i \in \widehat{\mathcal U}$, so the
internal region contains all $n$ calibration points and $ \hat\tau^+(\alpha, X_{n+1})  = \hat\tau(\alpha;\widehat{\mathcal I}) = \hat{\tau}(\alpha)$ almost surely, 
so $\widehat{\mathcal{C}}^+(X_{n+1}) = \widehat{\mathcal{C}}(X_{n+1})$.
\end{proof}

\begin{proof}[of Proposition~\ref{prop:oracle-coverage}]
$~$
\begin{enumerate}
\item[(i)] By definition of $\psi$,
$\mathcal{O}_0(x) = \psi(q^{\mathrm{lo}}(x), q^{\mathrm{up}}(x)) \supseteq
[q^{\mathrm{lo}}(x), q^{\mathrm{up}}(x)]$, so 
$\mathbb{P}(Y \in \mathcal{O}_0(X) \mid X) \geq F_{Y\mid X}(q^{\mathrm{up}}(X)) -
F_{Y\mid X}(q^{\mathrm{lo}}(X)^-) \geq (1-\tfrac\alpha2) - \tfrac\alpha2 = 1-\alpha$
a.s. Hence $\mathbb{P}(s^\star(X, Y) \leq 0) = \mathbb{P}(Y \in \mathcal{O}_0(X)) \geq 1-\alpha$
by the identity $\{Y \in \phi^\star(X; \tau)\} = \{s^\star(X, Y) \leq \tau\}$, so
$\tau^{\mathrm{m}} \leq 0$ and, by monotonicity of $\phi^\star(x; \cdot)$,
$\mathcal{O}_{\mathrm{m}}(x) \subseteq \mathcal{O}_0(x)$. The same identity gives
$\mathbb{P}(Y \in \mathcal{O}_{\mathrm{m}}(X)) = \mathbb{P}(s^\star(X, Y) \leq \tau^{\mathrm{m}}) \geq 1-\alpha$.
Since $\tau^{\mathrm{m}}$ is the smallest $\tau$ satisfying $\mathbb{P}(s^\star(X, Y) \leq \tau) \geq 1-\alpha$, 
and the family is nested, $\mathcal{O}_{\mathrm{m}}$ the tightest family member with
marginal coverage at least $1-\alpha$.

\item[(ii)] $\mathcal{O}_0(x) = [\underline\psi(q^{\mathrm{lo}}(x)),\, \overline\psi(q^{\mathrm{up}}(x))]$, and under
Assumption~\ref{ass:regularity} $F_{Y\mid X}(q^{\mathrm{lo}-} \mid x) = \tfrac\alpha2$ and
$F_{Y\mid X}(q^{\mathrm{up}} \mid x) = 1-\tfrac\alpha2$. Whenever the band meets the observable window
($q^{\mathrm{up}}(x)\ge C^{\mathrm L}$ and $q^{\mathrm{lo}}(x)\le C^{\mathrm R}$), snapping moves only the censored endpoint, so
$F_{Y\mid X}(\underline\psi(q^{\mathrm{lo}})^- \mid x)=\tfrac\alpha2\,\I\{q^{\mathrm{lo}}(x)>C^{\mathrm L}\}$ and
$F_{Y\mid X}(\overline\psi(q^{\mathrm{up}}) \mid x)=1-\tfrac\alpha2\,\I\{q^{\mathrm{up}}(x)<C^{\mathrm R}\}$, hence
\[
    \mathbb{P}(Y \in \mathcal{O}_0(x) \mid X = x)
    = 1-\alpha + \tfrac\alpha2\bigl(\I\{q^{\mathrm{lo}}(x)\le C^{\mathrm L}\}+\I\{q^{\mathrm{up}}(x)\ge C^{\mathrm R}\}\bigr).
\]
The remaining bands lie entirely outside the window: if $q^{\mathrm{up}}(x)<C^{\mathrm L}$ then
$\mathcal{O}_0(x)=[Y^{\min},C^{\mathrm L}]$ with coverage $F_{Y\mid X}(C^{\mathrm L}\mid x)>1-\tfrac\alpha2$, and if
$q^{\mathrm{lo}}(x)>C^{\mathrm R}$ then $\mathcal{O}_0(x)=[C^{\mathrm R},Y^{\max}]$ with coverage
$1-F_{Y\mid X}(C^{\mathrm R}\mid x)>1-\tfrac\alpha2$.

\item[(iii)] By the score identity, $\mathbb{P}(Y \in \mathcal{O}_{\mathrm{m}}(X)) = \mathbb{P}(s^\star(X,Y) \leq \tau^{\mathrm{m}}) = 1-\alpha$; indeed,
from~\eqref{eq:oracle-score} the only candidate atoms of $s^\star(X,Y)$ are the censored-region values
$c^{\mathrm L}(X), c^{\mathrm R}(X)$, which are continuous functions of the atomless quantiles
(Assumption~\ref{ass:regularity}(i)) and hence carry no mass. For conditional coverage, fix $x$ with
$C^{\mathrm{L}} < q^{\mathrm{lo}}(x) \leq q^{\mathrm{up}}(x) < C^{\mathrm{R}}$. Since $\tau^{\mathrm{m}} \leq 0$,
$\mathcal{O}_{\mathrm{m}}(x) = [q^{\mathrm{lo}}(x) - \tau^{\mathrm{m}}, q^{\mathrm{up}}(x) + \tau^{\mathrm{m}}] \subseteq [q^{\mathrm{lo}}(x), q^{\mathrm{up}}(x)]$,
so $\mathbb{P}(Y \in \mathcal{O}_{\mathrm{m}}(X) \mid x) \leq 1-\alpha$, with equality iff $\tau^{\mathrm{m}} = 0$, which by part~(iv) holds iff
Assumption~\ref{ass:observability} holds.

\item[(iv)] ($\Leftarrow$) Under Assumption~\ref{ass:observability}, $\pi^{\mathrm L} = \pi^{\mathrm R} = 0$, so
$\mathbb{P}(s^\star \leq 0) = 1-\alpha$ by (ii); for $t < 0$ the band $\phi^\star(X; t)$ is an interior strict
contraction of $[q^{\mathrm{lo}}(X), q^{\mathrm{up}}(X)]$ with coverage below $1-\alpha$ (positive density), so
$\tau^{\mathrm{m}} = 0$ and $\mathcal{O}_{\mathrm{m}} = \mathcal{O}_0 = [q^{\mathrm{lo}}, q^{\mathrm{up}}] = \mathcal{O}_{\mathrm{uncens}}$.
($\Rightarrow$) If Assumption~\ref{ass:observability} fails, then $\pi^{\mathrm L} + \pi^{\mathrm R} > 0$, so
$\mathbb{P}(s^\star \leq 0) = 1-\alpha + \tfrac\alpha2(\pi^{\mathrm L} + \pi^{\mathrm R}) > 1-\alpha$ by (ii). Since
$s^\star(X, Y)$ is atomless under Assumption~\ref{ass:regularity}, $\mathbb{P}(s^\star < 0) = \mathbb{P}(s^\star \leq 0) > 1-\alpha$,
so $\tau^{\mathrm{m}} < 0$ and $\mathcal{O}_{\mathrm{m}} \subset \mathcal{O}_0$.

\end{enumerate}
\end{proof}

\begin{proof}[of Theorem~\ref{thm:ccqr-oracle}]
Define also $\Delta_n(x) := |\hat{f}_n^{\mathrm{lo}}(x)-q^{\mathrm{lo}}(x)| +
|\hat{f}_n^{\mathrm{up}}(x)-q^{\mathrm{up}}(x)|$, and $a_n := \mathbb{E}[\Delta_n(X)]$, so
$a_n \to 0$ by Assumption~\ref{ass:full-consistency}. 

The map $(f^{\mathrm{lo}}, f^{\mathrm{up}}) \mapsto s(x, y; f^{\mathrm{lo}}, f^{\mathrm{up}})$
from~\eqref{eq:score-formula} is $1$-Lipschitz in the $\ell_1$ norm, uniformly in $(x,y)$:
each branch is a maximum of affine functions whose gradients have $\ell_\infty$-norm at most
$1$, and the branch is fixed by $y$ independently of the predictor.
Therefore, 
\begin{equation}\label{eq:lipschitz-scores}
    |s_n(x,y) - s^\star(x,y)| \leq \Delta_n(x) \qquad \text{for all } (x,y),
\end{equation}
where $s_n$ is the score function~\eqref{eq:score-formula} with base predictor $\hat{f}_n$.
Abbreviate $S_n = s_n(X, Y)$ and $S^\star := s^\star(X, Y)$, and define 
$F^\star(t) = \mathbb{P}(S^\star \leq t) = \mathbb{P}(Y \in \phi^\star(X;t))$.
Recall also the identity $s(x,\tilde{y})=s(x,y)$ with $\tilde{y} = \Pi(y)$ for all $(x,y)$.

\begin{enumerate}
\item[(i-a)] We begin by proving $\hat\tau_n(\alpha)\xrightarrow{p}\tau^{\mathrm{m}} \leq 0$.
Since $F^\star(0) = \mathbb{P}(Y \in \mathcal{O}_0) \geq 1-\alpha$ by
Proposition~\ref{prop:oracle-coverage}, definition~\eqref{eq:marginal-oracle} gives
$\tau^\star = \inf\{t : F^\star(t) \geq 1-\alpha\} \leq 0$. 
If $a_n = 0$ then $S_n = S^\star$ almost surely and the claim is immediate, so assume $a_n>0$ and define
$\epsilon_n := \sqrt{a_n} \to 0$. 
Define $B_n := \{x : \Delta_n(x) > \epsilon_n\}$ and $\mathcal{B}_n := \{i \leq n : X_i \in B_n\}$. 
Combining the definition of $B_n$ with \eqref{eq:lipschitz-scores} implies
$|s_n(X_i,Y_i) - s^\star(X_i,Y_i)| \leq \epsilon_n$ for every $i \notin \mathcal{B}_n$.

Let $\widehat{F}_n, \widehat{F}^\star_n$ be the empirical CDFs of
$\{s_n(X_i,Y_i)\}_{i=1}^{n}$ and $\{s^\star(X_i,Y_i)\}_{i=1}^{n}$, so that $\hat\tau_n(\alpha) =
\widehat{F}_n^{-1}(p_n)$ with $p_n = \lceil (n+1)(1-\alpha)\rceil/n \to 1-\alpha$. Perturbing
at most $|\mathcal{B}_n|$ of the $n$ scores and shifting the rest by at most $\epsilon_n$ moves
the $p_n$-quantile by at most a rank shift of $|\mathcal{B}_n|$ and a value shift of
$\epsilon_n$:
\[
    \widehat{F}_n^{\star -1}\!\Bigl(p_n - \tfrac{|\mathcal{B}_n|}{n}\Bigr) - \epsilon_n
    \ \leq\ \hat\tau_n(\alpha) \ \leq\
    \widehat{F}_n^{\star -1}\!\Bigl(p_n + \tfrac{|\mathcal{B}_n|}{n}\Bigr) + \epsilon_n.
\]
By Markov's inequality, $\beta_n := \mathbb{P}(X \in B_n) \leq a_n/\epsilon_n = \epsilon_n \to 0$. 
Similarly, since $\mathbb{E}\bigl( |\mathcal{B}_n| \bigr) = n\beta_n$ we have $|\mathcal{B}_n|/n \xrightarrow{p} 0$.
Since we know $\epsilon_n \to 0$, the sandwich yields $\hat\tau_n(\alpha) \xrightarrow{p} \tau^\star$.
By Assumption~\ref{ass:regularity}(ii) $\tau^{\mathrm{m}}$ is the unique $(1-\alpha)$-quantile of
$F^\star$, so the quantile functional is continuous at $1-\alpha$; with
$p_n \pm |\mathcal{B}_n|/n \xrightarrow{p} 1-\alpha$ and Glivenko--Cantelli, classical
sample-quantile asymptotics \citep[Ch.~21]{van2000asymptotic} give
$\widehat{F}_n^{\star -1}(p_n \pm |\mathcal{B}_n|/n) \xrightarrow{p} \tau^{\mathrm{m}}$. 

\item[(i-b)] We now prove $\widehat{\mathcal{C}}_n \to \mathcal{O}_{\mathrm{m}}$.
Write the pre-snap endpoints $\hat{a}_n(X) = \hat{f}_n^{\mathrm{lo}}(X) - \hat\tau$,
$a^{\mathrm{m}}(X) = q^{\mathrm{lo}}(X) - \tau^{\mathrm{m}}$, and $\hat{b}_n, b^{\mathrm{m}}$ analogously, so that
$\widehat{\mathcal{C}}_n(X) = \psi(\hat{a}_n,\hat{b}_n)$ and
$\mathcal{O}_{\mathrm{m}}(X) = \psi(a^{\mathrm{m}},b^{\mathrm{m}})$. Set
$\delta^{\mathrm{lo}}_n := |\hat{a}_n-a^{\mathrm{m}}| \leq
|\hat{f}_n^{\mathrm{lo}}(X)-q^{\mathrm{lo}}(X)| + |\hat\tau-\tau^{\mathrm{m}}|$ and likewise
$\delta^{\mathrm{up}}_n$. 

Write the snapping map element-wise as in~\eqref{eq:psi-def}, 
$\psi(a,b) = [\underline{\psi}(a), \overline{\psi}(b)]$.
Each coordinate is nondecreasing and $1$-Lipschitz apart from a single upward jump: the jump of $\underline{\psi}$
is at $a = C^{\mathrm{L}}$ (from $Y^{\min}$ to $C^{\mathrm{L}}$) and the jump of $\overline{\psi}$ is at
$b = C^{\mathrm{R}}$ (from $C^{\mathrm{R}}$ to $Y^{\max}$). Since the symmetric difference of
two intervals is at most the sum of the gaps between their endpoints,
\[
    \lambda\bigl(\widehat{\mathcal{C}}_n(X) \triangle \mathcal{O}_{\mathrm{m}}(X)\bigr)
    \leq |\underline{\psi}(\hat{a}_n) - \underline{\psi}(a^{\mathrm{m}})|
       + |\overline{\psi}(\hat{b}_n) - \overline{\psi}(b^{\mathrm{m}})|.
\]
We bound the lower term; the upper is symmetric. There are two cases.

If $\hat{a}_n$ and $a^{\mathrm{m}}$ lie on the same side of $C^{\mathrm{L}}$, then either both are at
most $C^{\mathrm{L}}$, where $\underline{\psi}(\hat{a}_n) = \underline{\psi}(a^{\mathrm{m}}) = Y^{\min}$
and the difference is $0$, or both exceed $C^{\mathrm{L}}$, where
$\underline{\psi} = \min\{\,\cdot\,, C^{\mathrm{R}}\}$ is non-expansive and 
$|\underline{\psi}(\hat{a}_n) - \underline{\psi}(a^{\mathrm{m}})| \leq |\hat{a}_n - a^{\mathrm{m}}| =
\delta^{\mathrm{lo}}_n$. In either case: $|\underline{\psi}(\hat{a}_n) - \underline{\psi}(a^{\mathrm{m}})| \leq \delta^{\mathrm{lo}}_n$.

If instead $\hat{a}_n$ and $a^{\mathrm{m}}$ straddle opposite sides of $C^{\mathrm{L}}$, then
$\underline{\psi}(\hat{a}_n), \underline{\psi}(a^{\mathrm{m}}) \in [Y^{\min}, Y^{\max}]$ gives the
crude bound $|\underline{\psi}(\hat{a}_n) - \underline{\psi}(a^{\mathrm{m}})| \leq R$; and since
$C^{\mathrm{L}}$ lies between $\hat{a}_n$ and $a^{\mathrm{m}}$, we have
$|a^{\mathrm{m}} - C^{\mathrm{L}}| \leq |\hat{a}_n - a^{\mathrm{m}}| = \delta^{\mathrm{lo}}_n$.

Combining the two cases,
\[
    |\underline{\psi}(\hat{a}_n) - \underline{\psi}(a^{\mathrm{m}})|
    \leq \delta^{\mathrm{lo}}_n
       + R\,\I\{|a^{\mathrm{m}} - C^{\mathrm{L}}| \leq \delta^{\mathrm{lo}}_n\},
\]
because the opposite-side case can occur only when the indicator equals $1$.

The other term is analogous:
\[
    |\overline{\psi}(\hat{b}_n) - \overline{\psi}(b^{\mathrm{m}})|
    \leq \delta^{\mathrm{up}}_n + R\,\I\{|b^{\mathrm{m}} - C^{\mathrm{R}}| \leq \delta^{\mathrm{up}}_n\}.
\]
Combining and taking expectations,
\begin{align*}
    \mathbb{E} \left[ \lambda\bigl(\widehat{\mathcal{C}}_n(X) \triangle \mathcal{O}_{\mathrm{m}}(X)\bigr) \right]
  & \leq \mathbb{E} ( \delta^{\mathrm{lo}}_n(X) )  + \mathbb{E} ( \delta^{\mathrm{up}}_n(X) ) 
      + R \; \mathbb{P}(|a^{\mathrm{m}}(X)-C^{\mathrm{L}}| \leq \delta^{\mathrm{lo}}_n(X)) \\
  & \qquad + R \; \mathbb{P}(|b^{\mathrm{m}}(X)-C^{\mathrm{R}}| \leq \delta^{\mathrm{up}}_n(X) ).
\end{align*}

Note that $\mathbb{E} ( \delta^{\mathrm{lo}}_n(X) ) \leq \mathbb{E} (|\hat{f}_n^{\mathrm{lo}}(X)-q^{\mathrm{lo}}(X)|) + \mathbb{E} (|\hat\tau-\tau^{\mathrm{m}}|) \to 0$, because the first term vanishes by Assumption~\ref{ass:full-consistency}, the second by bounded
convergence ($\hat\tau \xrightarrow{p} \tau^{\mathrm{m}}$, both in $[-R,R]$); analogously, $\mathbb{E} ( \delta^{\mathrm{up}}_n(X) ) \to 0$.

For the first straddle term, for any $\xi > 0$,
\[
    \mathbb{P}(|a^{\mathrm{m}}(X)-C^{\mathrm{L}}| \leq \delta^{\mathrm{lo}}_n(X))
    \leq \mathbb{P}(|a^{\mathrm{m}}(X)-C^{\mathrm{L}}| \leq \xi) + \mathbb{P}(\delta^{\mathrm{lo}}_n(X) > \xi).
\]
The second term on the right-hand-side above vanishes because $\delta^{\mathrm{lo}}_n(X) \xrightarrow{p} 0$.
Recall then that, by Assumption~\ref{ass:regularity}(i),
$a^{\mathrm{m}}(X) = q^{\mathrm{lo}}(X) - \tau^{\mathrm{m}}$ has no atom at $C^{\mathrm{L}}$;
therefore, the term $\mathbb{P}(|a^{\mathrm{m}}(X)-C^{\mathrm{L}}| \leq \xi)$ also vanishes 
by letting $n \to \infty$ then $\xi \downarrow 0$.
The second straddle term vanishes in the same way.

Therefore, $\mathbb{E} [ \lambda\bigl(\widehat{\mathcal{C}}_n(X) \triangle \mathcal{O}_{\mathrm{m}}(X)\bigr) ] \to 0$.
Finally, since the integrand is nonnegative, this implies convergence in probability: 
$\widehat{\mathcal{C}}_n \to \mathcal{O}_{\mathrm{m}}$.

\item[(ii)] To prove the convergence of marginal coverage:
  \begin{align*}
    \bigl|\mathbb{P}(Y \in \widehat{\mathcal{C}}_n(X))
    - \mathbb{P}(Y \in \mathcal{O}_{\mathrm{m}}(X))\bigr|
    & \leq \mathbb{P}(Y \in \widehat{\mathcal{C}}_n(X) \triangle \mathcal{O}_{\mathrm{m}}(X)) \\
    & \leq M\,\mathbb{E} \bigl( \lambda( \widehat{\mathcal{C}}_n(X) \triangle \mathcal{O}_{\mathrm{m}}(X) ) \bigr) \to 0,
  \end{align*}
where the second inequality uses $f_{Y\mid X} \leq M$ (Assumption~\ref{ass:regularity}(ii)) and the limit is from Step~i-b.

\item[(iii)] To establish asymptotic conditional coverage, let $\mathcal{D}_{n}$ denote the calibration data set of size $n$ and define
$V_n := \mathbb{E}[\lambda(\widehat{\mathcal{C}}_n(X) \triangle \mathcal{O}_{\mathrm{m}}(X)) \mid \mathcal{D}_{n}]$; 
then $\mathbb{E}(V_n) \to 0$ (from Step~i-b), so $V_n \xrightarrow{p} 0$ by non-negativity.
On $\{V_n > 0\}$ define $\epsilon'_n := V_n^{1/2}$ and
$\Lambda_n := \{x : \lambda(\widehat{\mathcal{C}}_n(x) \triangle
\mathcal{O}_{\mathrm{m}}(x)) \leq \epsilon'_n\}$ (and $\Lambda_n := \mathcal{X}$ otherwise). By the
conditional Markov inequality,
$\mathbb{P}(X \notin \Lambda_n \mid \mathcal{D}_{n}) \leq V_n/\epsilon'_n = V_n^{1/2}
\xrightarrow{p} 0$, so $\mathbb{P}(X \in \Lambda_n \mid \mathcal{D}_{n}) = 1 - o_P(1)$, and on
$\Lambda_n$ the density bound gives
\[
    \sup_{x \in \Lambda_n}
    \bigl|\mathbb{P}(Y \in \widehat{\mathcal{C}}_n(x) \mid X=x)
        - \mathbb{P}(Y \in \mathcal{O}_{\mathrm{m}}(x) \mid X=x)\bigr|
    \leq M\epsilon'_n = M V_n^{1/2} \xrightarrow{p} 0.
\]
This is Definition~\ref{def:cond-cov} with limit $\mathcal{O}_{\mathrm{m}}$, whose conditional
coverage is given by Proposition~\ref{prop:oracle-coverage}.
\end{enumerate}
\end{proof}

\begin{proof}[of Proposition~\ref{prop:ecqr-oracle}]
For any $t \in \mathbb{R}$, define the event
\[
E_t := \{q^{\mathrm{lo}}(X)-t\le q^{\mathrm{up}}(X)+t<C^{\mathrm L},\,Y \leq C^{\mathrm L}\}\cup\{q^{\mathrm{up}}(X)+t\ge q^{\mathrm{lo}}(X)-t>C^{\mathrm R},\,Y \geq C^{\mathrm R}\}.
\]
Then, for any $t$, the same argument used in the proof of Theorem~\ref{thm:ecqr-validity} gives, almost surely,
\begin{align} \label{eq:identity-Et}
  \I\{Y\in\phi^\star(X;t)\}
  = \I\{\tilde Y\in[q^{\mathrm{lo}}(X)-t,\,q^{\mathrm{up}}(X)+t]\} + \I\{E_t\}.
\end{align}
Taking expectations,
\begin{equation}\label{eq:Fstar-decomp}
    \underbrace{\mathbb P\bigl(Y\in\phi^\star(X;t)\bigr)}_{F^\star(t)}
    \;=\;
    \underbrace{\mathbb P\bigl(\tilde Y\in[q^{\mathrm{lo}}(X)-t,\,q^{\mathrm{up}}(X)+t]\bigr)}_{G^{*}(t)}
    \;+\;
    \underbrace{\mathbb P\bigl(E_t\bigr)}_{p_{\mathrm{slack}}^\star(t)}.
\end{equation}
\emph{(i)} By~\eqref{eq:Fstar-decomp}, $F^\star(t)\ge G^{*}(t)$ for every $t$, so
$\{t:G^{*}(t)\ge1-\alpha\}\subseteq\{t:F^\star(t)\ge1-\alpha\}$ and
\[
   \tilde\tau^{\mathrm{m}}:=\inf \{t:G^{*}(t)\ge1-\alpha\}\ \ge\ \tau^{\mathrm{m}}:=\inf \{t:F^\star(t)\ge1-\alpha\}.
\]
Since the oracle family $\phi^\star(x;\cdot)$ is nested,
$\mathcal O_{\mathrm{m}}(x)=\phi^\star(x;\tau^{\mathrm{m}})\subseteq\phi^\star(x;\tilde\tau^{\mathrm{m}})=\mathcal O_{\mathrm{e}}(x)$.

\emph{(ii)} $\mathbb P(Y\in\mathcal O_{\mathrm{e}}(X))=F^\star(\tilde\tau^{\mathrm{m}})$. Under Assumption~\ref{ass:regularity}
the conformal quantile is attained with equality, $G^{*}(\tilde\tau^{\mathrm{m}})=1-\alpha$, so~\eqref{eq:Fstar-decomp} gives
$F^\star(\tilde\tau^{\mathrm{m}})=1-\alpha+p_{\mathrm{slack}}^\star(\tilde\tau^{\mathrm{m}})\ge1-\alpha$.

\emph{(iii)} Under Assumption~\ref{ass:regularity}, $G^{*}$ is continuous and strictly increasing, so
$\tilde\tau^{\mathrm{m}}<0$ iff $G^{*}(0)>1-\alpha$, with $\tilde\tau^{\mathrm{m}}=0$ iff $G^{*}(0)=1-\alpha$. According to the position
of the band relative to $[C^{\mathrm L},C^{\mathrm R}]$, the conditional coverage of the clipped outcome is
\begin{align*}
  \mathbb P\bigl(\tilde Y\in[q^{\mathrm{lo}}(X),q^{\mathrm{up}}(X)]\mid X\bigr)
  & = (1-\alpha)\,\I[C^{\mathrm L}(X)<q^{\mathrm{lo}}(X)\le q^{\mathrm{up}}(X)<C^{\mathrm R}] \\
  & \quad + \bigl(1-\tfrac\alpha2\bigr)\,\I[q^{\mathrm{lo}}(X)\le C^{\mathrm L}\le q^{\mathrm{up}}(X)<C^{\mathrm R}] \\
  & \quad + \bigl(1-\tfrac\alpha2\bigr)\,\I[C^{\mathrm L}<q^{\mathrm{lo}}(X)\le C^{\mathrm R}\le q^{\mathrm{up}}(X)] \\
  & \quad + \I[q^{\mathrm{lo}}(X)\le C^{\mathrm L}\le C^{\mathrm R}\le q^{\mathrm{up}}(X)] \\
  & = 1-\alpha + \tfrac\alpha2\,\I[q^{\mathrm{lo}}(X)\le C^{\mathrm L}\le q^{\mathrm{up}}(X)] \\
  & \quad + \tfrac\alpha2\,\I[q^{\mathrm{lo}}(X)\le C^{\mathrm R}\le q^{\mathrm{up}}(X)] \\
  & \quad - (1-\alpha)\,\I[q^{\mathrm{up}}(X)<C^{\mathrm L}\ \text{or}\ q^{\mathrm{lo}}(X)>C^{\mathrm R}],
\end{align*}
with the entirely-outside bands ($q^{\mathrm{up}}(X)<C^{\mathrm L}$ or $q^{\mathrm{lo}}(X)>C^{\mathrm R}$) contributing $0$. 
Taking expectations, and using
$\mathbb P(q^{\mathrm{lo}}(X)\le C^{\mathrm L}\le q^{\mathrm{up}}(X))=\pi^{\mathrm L}-\rho^{\mathrm L}$,
$\mathbb P(q^{\mathrm{lo}}(X)\le C^{\mathrm R}\le q^{\mathrm{up}}(X))=\pi^{\mathrm R}-\rho^{\mathrm R}$, and
$\mathbb P(q^{\mathrm{up}}(X)<C^{\mathrm L}\ \text{or}\ q^{\mathrm{lo}}(X)>C^{\mathrm R})=\rho^{\mathrm L}+\rho^{\mathrm R}$,
\begin{align*}
  G^{*}(0)
  & = 1-\alpha + \tfrac\alpha2(\pi^{\mathrm L}-\rho^{\mathrm L}) + \tfrac\alpha2(\pi^{\mathrm R}-\rho^{\mathrm R}) - (1-\alpha)(\rho^{\mathrm L}+\rho^{\mathrm R}) \\
  & = 1-\alpha + \tfrac\alpha2(\pi^{\mathrm L}+\pi^{\mathrm R}) - (1-\tfrac\alpha2)(\rho^{\mathrm L}+\rho^{\mathrm R}).
\end{align*}
Hence $\tilde\tau^{\mathrm{m}}<0 \iff G^{*}(0) < 1-\alpha \iff \tfrac\alpha2(\pi^{\mathrm L}+\pi^{\mathrm R})>(1-\tfrac\alpha2)(\rho^{\mathrm L}+\rho^{\mathrm R})$,
with equality giving $\tilde\tau^{\mathrm{m}}=0$; in the special case $\rho^{\mathrm L}=\rho^{\mathrm R}=0$ this is
$\tilde\tau^{\mathrm{m}}<0\iff\pi^{\mathrm L}+\pi^{\mathrm R}>0$.
The set ordering follows from nestedness of $\phi^\star(x;\cdot)$.

\emph{(iv)} By strict nestedness of $\phi^\star(x;\cdot)$ under Assumption~\ref{ass:regularity},
$\mathcal O_{\mathrm{e}}=\mathcal O_{\mathrm{m}}$ iff $\tilde\tau^{\mathrm{m}}=\tau^{\mathrm{m}}$. If $p_{\mathrm{slack}}^\star(\tau^{\mathrm{m}})=\mathbb P(E_{\tau^{\mathrm{m}}})=0$,
then~\eqref{eq:Fstar-decomp} gives $G^{*}(\tau^{\mathrm{m}})=F^\star(\tau^{\mathrm{m}})=1-\alpha$ (the last equality by
Proposition~\ref{prop:oracle-coverage}(iii)), so $\tilde\tau^{\mathrm{m}}\le\tau^{\mathrm{m}}$, which with~(i) forces $\tilde\tau^{\mathrm{m}}=\tau^{\mathrm{m}}$.
Conversely, if $\tilde\tau^{\mathrm{m}}=\tau^{\mathrm{m}}$ then $G^{*}(\tau^{\mathrm{m}})=1-\alpha$ by~(ii), and~\eqref{eq:Fstar-decomp} gives
$p_{\mathrm{slack}}^\star(\tau^{\mathrm{m}})=F^\star(\tau^{\mathrm{m}})-G^{*}(\tau^{\mathrm{m}})=0$. Hence
$\mathcal O_{\mathrm{e}}=\mathcal O_{\mathrm{m}}$ iff $p_{\mathrm{slack}}^\star(\tau^{\mathrm{m}})=0$. For the full collapse: if
Assumption~\ref{ass:observability} holds then $\pi^{\mathrm L}=\pi^{\mathrm R}=0$, the conditional quantiles are
almost-surely interior, $p_{\mathrm{slack}}^\star(0)=0$, with $\mathcal O_{\mathrm{m}}=\mathcal O_0=\mathcal O_{\mathrm{uncens}}$ and
$\tau^{\mathrm{m}}=0$ (Proposition~\ref{prop:oracle-coverage}(iv)); then $\tilde\tau^{\mathrm{m}}=\tau^{\mathrm{m}}=0$ as just shown, so all four
oracles coincide. Conversely, if Assumption~\ref{ass:observability} fails then $\mathcal O_{\mathrm{m}}\subsetneq\mathcal O_0$
(Proposition~\ref{prop:oracle-coverage}(iv)), so the four cannot all coincide.
\end{proof}

\begin{proof}[of Theorem~\ref{thm:ecqr-oracle}]
The argument parallels the proof of Theorem~\ref{thm:ccqr-oracle}, with the ClipCQR score replaced by the
plug-in CQR score evaluated at the clipped outcome and $\mathcal O_{\mathrm{m}}$ replaced by $\mathcal O_{\mathrm{e}}$.
Let
$\Delta_n(x) := |\hat f^{\mathrm{lo}}_n(x)-q^{\mathrm{lo}}(x)| + |\hat f^{\mathrm{up}}_n(x)-q^{\mathrm{up}}(x)|$,
$a_n := \mathbb E[\Delta_n(X)]\to0$ by Assumption~\ref{ass:full-consistency}. The plug-in CQR score
$s^{\mathrm{CQR}}_n(x,y)=\max\{\hat f^{\mathrm{lo}}_n(x)-y,\,y-\hat f^{\mathrm{up}}_n(x)\}$ is $1$-Lipschitz in $(\hat f^{\mathrm{lo}}_n,\hat f^{\mathrm{up}}_n)$ in the
$\ell_1$ norm uniformly in $(x,y)$, so
\begin{equation}\label{eq:lipschitz-cqr-scores}
    |s^{\mathrm{CQR}}_n(x,y)-s^{\mathrm{CQR},\star}(x,y)|\le\Delta_n(x)\qquad\text{for all }(x,y),
\end{equation} 
and in particular at $y=\tilde y$. Recall that $\tilde S^{\mathrm{CQR},\star}=s^{\mathrm{CQR},\star}(X,\tilde Y)$
from~\eqref{eq:cqr-oracle-score}, write $\tilde S_n := s^{\mathrm{CQR}}_n(X,\tilde Y)$, and let
$G^{*}(t):=\mathbb P(\tilde S^{\mathrm{CQR},\star}\le t)=\mathbb P\bigl(\tilde Y\in[q^{\mathrm{lo}}(X)-t,\,q^{\mathrm{up}}(X)+t]\bigr)$,
so that $\tilde\tau^{\mathrm{m}}=\inf\{t:G^{*}(t)\ge1-\alpha\}$.

\begin{enumerate}
\item[(i-a)] We first show $\tilde\tau_n(\alpha)\xrightarrow{p}\tilde\tau^{\mathrm{m}}(\alpha)$.
If $a_n=0$ then $\tilde S_n=\tilde S^{\mathrm{CQR},\star}$ almost surely and the claim is immediate, so assume
$a_n>0$ and set $\epsilon_n:=\sqrt{a_n}\to0$. Define $B_n:=\{x:\Delta_n(x)>\epsilon_n\}$ and
$\mathcal B_n:=\{i\le n:X_i\in B_n\}$; by~\eqref{eq:lipschitz-cqr-scores},
$|s^{\mathrm{CQR}}_n(X_i,\tilde Y_i)-s^{\mathrm{CQR},\star}(X_i,\tilde Y_i)|\le\epsilon_n$ for every $i\notin\mathcal B_n$.

Let $\widehat G_n,\widehat G^{\mathrm{m}}_n$ be the empirical CDFs of $\{s^{\mathrm{CQR}}_n(X_i,\tilde Y_i)\}_{i=1}^n$
and $\{s^{\mathrm{CQR},\star}(X_i,\tilde Y_i)\}_{i=1}^n$, so that
$\tilde\tau_n(\alpha)=\widehat G_n^{-1}(p_n)$ with $p_n=\lceil(n+1)(1-\alpha)\rceil/n\to1-\alpha$.
Perturbing at most $|\mathcal B_n|$ of the $n$ scores and shifting the rest by at most $\epsilon_n$ moves the
$p_n$-quantile by a rank shift of $|\mathcal B_n|$ and a value shift of $\epsilon_n$:
\[
    \widehat G_n^{{\mathrm{m}}-1}\!\Bigl(p_n-\tfrac{|\mathcal B_n|}{n}\Bigr)-\epsilon_n
    \ \le\ \tilde\tau_n(\alpha)\ \le\
    \widehat G_n^{{\mathrm{m}}-1}\!\Bigl(p_n+\tfrac{|\mathcal B_n|}{n}\Bigr)+\epsilon_n.
\]
By Markov's inequality $\beta_n:=\mathbb P(X\in B_n)\le a_n/\epsilon_n=\epsilon_n\to0$, and since
$\mathbb E(|\mathcal B_n|)=n\beta_n$ we have $|\mathcal B_n|/n\xrightarrow{p}0$. By assumption $\tilde\tau^{\mathrm{m}}$
is the unique $(1-\alpha)$-quantile of $\tilde F^{\mathrm{m}}$ and a continuity point of its law, so the quantile
functional is continuous at $1-\alpha$; with $p_n\pm|\mathcal B_n|/n\xrightarrow{p}1-\alpha$, Glivenko--Cantelli
and classical sample-quantile asymptotics \citep[Ch.~21]{van2000asymptotic} give
$\widehat G_n^{{\mathrm{m}}-1}(p_n\pm|\mathcal B_n|/n)\xrightarrow{p}\tilde\tau^{\mathrm{m}}$, and the sandwich yields
$\tilde\tau_n(\alpha)\xrightarrow{p}\tilde\tau^{\mathrm{m}}$.

\item[(i-b)] The proof is identical to the proof of 
Theorem~\ref{thm:ccqr-oracle}(i-b) and thus omitted to avoid repetition.

\item[(ii)] As in the proof of Theorem~\ref{thm:ccqr-oracle}(ii), the density bound $f_{Y\mid X}\le M$
(Assumption~\ref{ass:regularity}(ii)) gives
\[
    \bigl|\mathbb P(Y\in\widehat{\mathcal C}^{\,\mathrm{eCQR}}_n(X))-\mathbb P(Y\in\mathcal O_{\mathrm{e}}(X))\bigr|
    \le M\,\mathbb E\bigl[\lambda(\widehat{\mathcal C}^{\,\mathrm{eCQR}}_n(X)\triangle\mathcal O_{\mathrm{e}}(X))\bigr]\to0,
\]
and the limit equals $\mathbb P(Y\in\mathcal O_{\mathrm{e}}(X))=1-\alpha+p_{\mathrm{slack}}^\star(\tilde\tau^{\mathrm{m}})$
by Proposition~\ref{prop:ecqr-oracle}(ii).

\item[(iii)] The proof is identical to the proof of 
Theorem~\ref{thm:ccqr-oracle}(iii) and thus omitted to avoid repetition.

\end{enumerate}
\end{proof}

\begin{proof}[of Lemma~\ref{lem:pecqr-is-projected-clipcqr}]
peCQR is eCQR applied to the projected base predictor $g = \Pi\hat f$, so it computes non-conformity scores
using the function $s^{\mathrm{CQR}}(x,\tilde y; g) = \max\{g^{\mathrm{lo}}(x) - \tilde y,\ \tilde y - g^{\mathrm{up}}(x)\}$ at the observed response $\tilde y = \Pi(y)$, takes the $\lceil(n+1)(1-\alpha)\rceil$ order
statistic as its threshold, and returns the snapped band. 
ClipCQR applied with the same base predictor $\Pi\hat f$ follows the same approach but uses the non-conformity score function
defined in~\eqref{eq:score-formula}, so it suffices to show the two score functions agree everywhere.

Write $g := \Pi\hat f$, so $C^{\mathrm L} \le g^{\mathrm{lo}}(x) \le g^{\mathrm{up}}(x) \le C^{\mathrm R}$ and $g^{\mathrm{lo}}(x)-C^{\mathrm L} \ge 0$
and $C^{\mathrm R}-g^{\mathrm{up}}(x) \ge 0$, while $\tfrac12(g^{\mathrm{lo}}(x)-g^{\mathrm{up}}(x)) \le 0$, for any $x$.
Therefore, for any $(x,\tilde{y})$, the score function defined in~\eqref{eq:score-formula} gives:
\begin{align*}
    s(x, \tilde{y}; g) 
    & = \begin{cases}
  \max\bigl\{g^{\mathrm{lo}}(x) - C^{\mathrm{L}},\, \tfrac{1}{2}(g^{\mathrm{lo}}(x) - g^{\mathrm{up}}(x))\bigr\}, & \text{if } \tilde{y} \leq C^{\mathrm{L}}, \\
  \max\bigl\{g^{\mathrm{lo}}(x) - \tilde{y},\, \tilde{y} - g^{\mathrm{up}}(x)\bigr\}, & \text{if } C^{\mathrm{L}} < \tilde{y} < C^{\mathrm{R}}, \\
        \max\bigl\{C^{\mathrm{R}} - g^{\mathrm{up}}(x),\, \tfrac{1}{2}(g^{\mathrm{lo}}(x) - g^{\mathrm{up}}(x))\bigr\}, & \text{if } \tilde{y} \geq C^{\mathrm{R}}.
    \end{cases} \\
    & = \begin{cases}
  g^{\mathrm{lo}}(x) - C^{\mathrm{L}}, & \text{if } \tilde{y} \leq C^{\mathrm{L}}, \\
  \max\bigl\{g^{\mathrm{lo}}(x) - \tilde{y},\, \tilde{y} - g^{\mathrm{up}}(x)\bigr\}, & \text{if } C^{\mathrm{L}} < \tilde{y} < C^{\mathrm{R}}, \\
        C^{\mathrm{R}} - g^{\mathrm{up}}(x), & \text{if } \tilde{y} \geq C^{\mathrm{R}}.
    \end{cases} \\
  & = \max\bigl\{g^{\mathrm{lo}}(x) - \tilde{y},\, \tilde{y} - g^{\mathrm{up}}(x)\bigr\} = s^{\mathrm{CQR}}(x,\tilde y; g).
\end{align*}
Hence the ClipCQR and eCQR scores coincide at every calibration point, giving identical thresholds and, from
the common base $\Pi\hat f$, the identical snapped band.
\end{proof}

\begin{proof}[of Theorem~\ref{thm:pecqr-oracle}]
Define the projected oracle family
\[
  \phi^{\mathrm{pe},\star}(x;\tau):=\psi(\Pi q^{\mathrm{lo}}(x)-\tau,\Pi q^{\mathrm{up}}(x)+\tau),
\]
with score function $s^{\mathrm{pe},\star}(x,\tilde y):=\max\{\Pi q^{\mathrm{lo}}(x)-\tilde y,\,\tilde y-\Pi q^{\mathrm{up}}(x)\}$, and the projected marginal
oracle
\begin{align*}
  & \mathcal O^{\mathrm{pe}}_{\mathrm{m}}:=\phi^{\mathrm{pe},\star}(\cdot;\tau^{\mathrm{pe},\star}),
  & \tau^{\mathrm{pe},\star}:=\inf\{t:\mathbb P(s^{\mathrm{pe},\star}(X,\tilde Y)\le t)\ge1-\alpha\}.
\end{align*}
Note that the score inversion identity gives
$\{\tilde Y\in\phi^{\mathrm{pe},\star}(X;\tau)\}=\{s^{\mathrm{pe},\star}(X,\tilde Y)\le\tau\}$, and (clipping identity)
$y \in \phi^{\mathrm{pe},\star}(x;\tau) \iff \tilde{y} = \Pi(y) \in \phi^{\mathrm{pe},\star}(x;\tau)$ by definition of $\psi$.

We begin by showing $\tau^{\mathrm{pe},\star} = 0$ and $\mathcal O^{\mathrm{pe}}_{\mathrm{m}}=\mathcal O_0$, in two steps.

\begin{itemize}

\item \emph{Step 1: $\tau^{\mathrm{pe},\star} \leq 0$}.
At $\tau=0$, $\phi^{\mathrm{pe},\star}(X;0)=\psi(\Pi q^{\mathrm{lo}}(X),\Pi q^{\mathrm{up}}(X))
=\psi(q^{\mathrm{lo}}(X),q^{\mathrm{up}}(X))=\mathcal O_0(X)$ since $\psi\circ\Pi=\psi$. By score inversion and clipping identity $\{Y\in\mathcal O_0\}=\{\tilde Y\in\mathcal O_0\}$,
\[
  \mathbb P(s^{\mathrm{pe},\star}(X,\tilde Y)\le0)=\mathbb P(Y\in\mathcal O_0)\ge\mathbb P(q^{\mathrm{lo}}\le Y\le q^{\mathrm{up}})\ge1-\alpha,
\]
since $\mathcal O_0\supseteq[q^{\mathrm{lo}},q^{\mathrm{up}}]$. Hence $\tau^{\mathrm{pe},\star}\le0$ with no assumptions.

\item \emph{Step 2: $\tau^{\mathrm{pe},\star} \geq 0$}.
Lemma~\ref{lem:score-deficit} in Appendix~\ref{sec:appendix-auxiliary} gives $\mathbb P(s^{\mathrm{pe},\star}(X,\tilde Y)<0 )\le1-\alpha$. 
Since only non-censored observations with $C^{\mathrm L}  < Y < C^{\mathrm R}$ and $\tilde{Y} = Y$ can produce a negative score $s^{\mathrm{pe},\star}(X,\tilde{Y}) < 0$, and 
the distribution of $Y\mid X$ is continuous (Assumption~\ref{ass:regularity}), it follows that the distribution of 
$s^{\mathrm{pe},\star}(X,\tilde Y)$ is continuous on $(-\infty,0)$.
Therefore, for every $t < 0$, $\mathbb P(s^{\mathrm{pe},\star}(X,\tilde Y)\le t)<\mathbb P(s^{\mathrm{pe},\star}(X,\tilde Y)<0)\le1-\alpha$, which implies
$\tau^{\mathrm{pe},\star} \geq 0$.

\end{itemize}
Combining the results from Steps 1 and 2 gives $\tau^{\mathrm{pe},\star} = 0$, and thus $\mathcal O^{\mathrm{pe}}_{\mathrm{m}}(\cdot) =\phi^{\mathrm{pe},\star}(\cdot;0)=\mathcal O_0(\cdot)$.

From here, we can prove (i)--(iii) by proceeding similarly to the proof of Theorem~\ref{thm:ccqr-oracle}.
Define also $\Delta_n(x) := |\hat{f}_n^{\mathrm{lo}}(x)-q^{\mathrm{lo}}(x)| +
|\hat{f}_n^{\mathrm{up}}(x)-q^{\mathrm{up}}(x)|$, and $a_n := \mathbb{E}[\Delta_n(X)]$, so
$a_n \to 0$ by Assumption~\ref{ass:consistency}. 
Since $\Pi$ is $1$-Lipschitz, $\Delta^{\mathrm{pe}}_n(x):=|\Pi\hat f^{\mathrm{lo}}_n(x)-\Pi q^{\mathrm{lo}}(x)|+|\Pi\hat f^{\mathrm{up}}_n(x)-\Pi q^{\mathrm{up}}(x)|\le\Delta_n(x)$,
so $a^{\mathrm{pe}}_n:=\mathbb E[\Delta^{\mathrm{pe}}_n(X)]\le a_n\to0$.
Moreover, the score function of peCQR is the same as that of ClipCQR (Lemma~\ref{lem:pecqr-is-projected-clipcqr}), which is 1-Lipschitz in the $\ell_1$ norm, and therefore $|s^{\mathrm{pe}}_n(x,y)-s^{\mathrm{pe},\star}(x,y)|\le\Delta^{\mathrm{pe}}_n(x)$ for all $(x,y)$, as in~\eqref{eq:lipschitz-scores}.
Abbreviate $S_n = s^{\mathrm{pe}}_n(X, \tilde Y)$ and $S^{\mathrm{m}} := s^{\mathrm{pe},\star}(X, \tilde Y)$, and write $F^{\mathrm{pe},\star}(t):=\mathbb P(S^{\mathrm{m}}\le t)$.

\begin{enumerate}

\item[(i-a)] We prove $\tilde\tau^{\mathrm{pe}}_n(\alpha)\xrightarrow{p}0$ by applying the sandwich argument used to prove Theorem~\ref{thm:ccqr-oracle}(i-a).
The only additional difficulty is showing that $\tau^{\mathrm{pe},\star}=0$ is the unique $(1-\alpha)$-quantile of $S^{\mathrm{m}}$; once that is established, the rest of the proof is the same as that of Theorem~\ref{thm:ccqr-oracle}(i-a).

\emph{Value of $F^{\mathrm{pe},\star}(t)$ at $t=0$.} Using the score identity and $\psi\circ\Pi=\psi$,
\begin{align*}
  \{S^{\mathrm{m}}\le0\}
  & = \{\tilde Y \in \psi\bigl(\Pi q^{\mathrm{lo}}(X), \Pi q^{\mathrm{up}}(X)\bigr) \}
   = \{\tilde Y \in \psi\bigl(q^{\mathrm{lo}}(X), q^{\mathrm{up}}(X)\bigr) \} \\
  & = \{\tilde Y \in \mathcal{O}_0 \} 
    = \{Y \in \mathcal{O}_0 \}.
\end{align*}
Therefore, by Proposition~\ref{prop:oracle-coverage}(ii),
$  F^{\mathrm{pe},\star}(0)=\mathbb P(Y\in\mathcal O_0)
  =1-\alpha+\tfrac\alpha2\bigl(\pi^{\mathrm L}+\pi^{\mathrm R}\bigr)
  \ \ge\ 1-\alpha$.

\emph{Strict inequalities for $F^{\mathrm{pe},\star}(t)$ at $t<0$ and $t>0$.} Step~2 above established that $F^{\mathrm{pe},\star}(t)<1-\alpha$ for every $t<0$.
For $t>0$ we claim $F^{\mathrm{pe},\star}(t)>1-\alpha$. If $\pi^{\mathrm L}+\pi^{\mathrm R}>0$,
then $F^{\mathrm{pe},\star}(0)>1-\alpha$ already.
If $\pi^{\mathrm L}+\pi^{\mathrm R}=0$, then $q^{\mathrm{lo}}(X)>C^{\mathrm L}$ and $q^{\mathrm{up}}(X)<C^{\mathrm R}$ a.s.,
so for each $t>0$ and a.e.\ $X$ the values
$Y\in\bigl(q^{\mathrm{lo}}(X)-\min\{t,\,q^{\mathrm{lo}}(X)-C^{\mathrm L}\},\,q^{\mathrm{lo}}(X)\bigr)$
are interior ($\tilde Y=Y$) and receive score $S^{\mathrm{m}}=q^{\mathrm{lo}}(X)-Y\in(0,t)$; this interval has
positive mass because $Y\mid X$ has a strictly positive density on $[Y^{\min},Y^{\max}]$
(Assumption~\ref{ass:regularity}), so $\mathbb P(0<S^{\mathrm{m}}\le t)>0$ and
$F^{\mathrm{pe},\star}(t)>F^{\mathrm{pe},\star}(0)=1-\alpha$.

Thus $F^{\mathrm{pe},\star}(t)<1-\alpha$ for $t<0$ and $F^{\mathrm{pe},\star}(t)>1-\alpha$ for $t>0$,
so $\tau^{\mathrm{pe},\star}=0$ is the unique $(1-\alpha)$-quantile of $S^{\mathrm{m}}$; the quantile
functional is therefore continuous at level $1-\alpha$.

The rest of the proof is the same as that of Theorem~\ref{thm:ccqr-oracle}(i-a).

\item[(i-b)] With $\tau^{\mathrm{pe},\star}=0$, write the pre-snap endpoints (omitting explicit dependence on $X$),
\[
  \hat a^{\mathrm{pe}}_n=\Pi\hat f^{\mathrm{lo}}_n-\tilde\tau^{\mathrm{pe}}_n,\quad a^{\mathrm{pe},\star}=\Pi q^{\mathrm{lo}},
  \qquad
  \hat b^{\mathrm{pe}}_n=\Pi\hat f^{\mathrm{up}}_n+\tilde\tau^{\mathrm{pe}}_n,\quad b^{\mathrm{pe},\star}=\Pi q^{\mathrm{up}}.
\]
Since $\hat a^{\mathrm{pe}}_n-a^{\mathrm{pe},\star}=(\Pi\hat f^{\mathrm{lo}}_n-\Pi q^{\mathrm{lo}})-\tilde\tau^{\mathrm{pe}}_n$,
\[
  \delta^{\mathrm{lo}}_n:=|\hat a^{\mathrm{pe}}_n-a^{\mathrm{pe},\star}|\le|\Pi\hat f^{\mathrm{lo}}_n-\Pi q^{\mathrm{lo}}|+|\tilde\tau^{\mathrm{pe}}_n|,
\]
and $\mathbb E[\delta^{\mathrm{lo}}_n]\to0$ by Assumption~\ref{ass:consistency} and $\tilde\tau^{\mathrm{pe}}_n\xrightarrow{p}0$ (with
$\tilde\tau^{\mathrm{pe}}_n$ bounded by $R$), as in Theorem~\ref{thm:ccqr-oracle}(i-b). As
$\widehat{\mathcal C}^{\,\mathrm{peCQR}}_n=[\underline\psi(\hat a^{\mathrm{pe}}_n),\overline\psi(\hat b^{\mathrm{pe}}_n)]$ and
$\mathcal O_0=[\underline\psi(a^{\mathrm{pe},\star}),\overline\psi(b^{\mathrm{pe},\star})]$,
\[
  \lambda(\widehat{\mathcal C}^{\,\mathrm{peCQR}}_n\triangle\mathcal O_0)
  \le|\underline\psi(\hat a^{\mathrm{pe}}_n)-\underline\psi(a^{\mathrm{pe},\star})|+|\overline\psi(\hat b^{\mathrm{pe}}_n)-\overline\psi(b^{\mathrm{pe},\star})|.
\]
We bound the lower term; the upper is symmetric.

The map $\underline\psi$ equals $Y^{\min}$ on $\{\cdot\le C^{\mathrm L}\}$ and $\min\{\cdot,C^{\mathrm R}\}$ on $\{\cdot>C^{\mathrm L}\}$;
it is $1$-Lipschitz everywhere except for a jump (of at most $R$) at $C^{\mathrm L}$, hence
\[
  |\underline\psi(x)-\underline\psi(y)|\le|x-y|+R\,\I\{x,y\text{ straddle }C^{\mathrm L}\}.
\]
Apply this with $(x,y)=(\hat a^{\mathrm{pe}}_n,a^{\mathrm{pe},\star})$. Because $a^{\mathrm{pe},\star}=\Pi q^{\mathrm{lo}}\ge C^{\mathrm L}$, a straddle is
either $\hat a^{\mathrm{pe}}_n\le C^{\mathrm L}<a^{\mathrm{pe},\star}$ (whence $a^{\mathrm{pe},\star}-C^{\mathrm L}\le\delta^{\mathrm{lo}}_n$) or
$a^{\mathrm{pe},\star}=C^{\mathrm L}<\hat a^{\mathrm{pe}}_n$, so
\[
  |\underline\psi(\hat a^{\mathrm{pe}}_n)-\underline\psi(a^{\mathrm{pe},\star})|
  \le \delta^{\mathrm{lo}}_n
   + R\,\I\{0<a^{\mathrm{pe},\star}-C^{\mathrm L}\le\delta^{\mathrm{lo}}_n\}
   + R\,\I\{a^{\mathrm{pe},\star}=C^{\mathrm L}<\hat a^{\mathrm{pe}}_n\}.
\]
Each term has vanishing expectation:
\begin{itemize}
\item $\mathbb E[\delta^{\mathrm{lo}}_n]\to0$, as above.
\item On $\{a^{\mathrm{pe},\star}>C^{\mathrm L}\}$, $a^{\mathrm{pe},\star}=\min\{q^{\mathrm{lo}},C^{\mathrm R}\}$ has no atom at $C^{\mathrm L}$
(Assumption~\ref{ass:regularity}(i)); since $\delta^{\mathrm{lo}}_n\xrightarrow{p}0$,
$\mathbb P(0<a^{\mathrm{pe},\star}-C^{\mathrm L}\le\delta^{\mathrm{lo}}_n)\to0$, as in Theorem~\ref{thm:ccqr-oracle}(i-b).
\item $a^{\mathrm{pe},\star}=C^{\mathrm L}$ iff $q^{\mathrm{lo}}\le C^{\mathrm L}$, and on $\{\tilde\tau^{\mathrm{pe}}_n\ge0\}$ the inequality
$\hat a^{\mathrm{pe}}_n>C^{\mathrm L}$ forces $\hat f^{\mathrm{lo}}_n\ge\Pi\hat f^{\mathrm{lo}}_n\ge\hat a^{\mathrm{pe}}_n>C^{\mathrm L}$, so
\[
  \{a^{\mathrm{pe},\star}=C^{\mathrm L}<\hat a^{\mathrm{pe}}_n\}\subseteq\{q^{\mathrm{lo}}\le C^{\mathrm L}<\hat f^{\mathrm{lo}}_n\}\cup\{\tilde\tau^{\mathrm{pe}}_n<0\}.
\]
On the first event $\Pi q^{\mathrm{lo}}=C^{\mathrm L}$ and $\Pi\hat f^{\mathrm{lo}}_n=\hat f^{\mathrm{lo}}_n$, so for every sufficiently small $\eta>0$,
\[
  \{q^{\mathrm{lo}}\le C^{\mathrm L}<\hat f^{\mathrm{lo}}_n\}
  \subseteq\{|\Pi\hat f^{\mathrm{lo}}_n-\Pi q^{\mathrm{lo}}|>\eta\}\cup\{C^{\mathrm L}<\hat f^{\mathrm{lo}}_n\le C^{\mathrm L}+\eta\}.
\]
By Markov and Assumption~\ref{ass:consistency}, $\mathbb P(|\Pi\hat f^{\mathrm{lo}}_n-\Pi q^{\mathrm{lo}}|>\eta) \to 0$ at fixed
$\eta$, so $\limsup_n\mathbb P(q^{\mathrm{lo}}\le C^{\mathrm L}<\hat f^{\mathrm{lo}}_n)\le\limsup_n\mathbb P(C^{\mathrm L}<\hat f^{\mathrm{lo}}_n\le C^{\mathrm L}+\eta)$,
which $\to0$ as $\eta\downarrow0$ by Assumption~\ref{ass:base-regularity}. 

It remains to handle the $\{\tilde\tau^{\mathrm{pe}}_n<0\}$ term, and here we split on $\pi^{\mathrm L}$.

If $\pi^{\mathrm L}=0$, then $\mathbb P(a^{\mathrm{pe},\star}=C^{\mathrm L})=\mathbb P(q^{\mathrm{lo}}\le C^{\mathrm L})=\pi^{\mathrm L}=0$,
so $\I\{a^{\mathrm{pe},\star}=C^{\mathrm L}<\hat a^{\mathrm{pe}}_n\}=0$ a.s.\ and the third term is identically $0$.

If $\pi^{\mathrm L}>0$, then Lemma~\ref{lem:tau-nonneg} in Appendix~\ref{sec:appendix-auxiliary} gives $P\bigl(\tau^{\mathrm{pe}}<0\bigr) \to 0$.

Therefore,
\[
  \mathbb P\bigl(a^{\mathrm{pe},\star}=C^{\mathrm L}<\hat a^{\mathrm{pe}}_n\bigr)
  \ \le\
  \mathbb P\bigl(q^{\mathrm{lo}}\le C^{\mathrm L}<\hat f^{\mathrm{lo}}_n\bigr)
  +\mathbb P\bigl(\tilde\tau^{\mathrm{pe}}_n<0\bigr)
   \to 0.
\]

\end{itemize}
Hence $\mathbb E \bigl( |\underline\psi(\hat a^{\mathrm{pe}}_n)-\underline\psi(a^{\mathrm{pe},\star})| \bigr) \to0$. 
The upper term vanishes by a symmetric argument.
Therefore $\mathbb E[\lambda(\widehat{\mathcal C}^{\,\mathrm{peCQR}}_n\triangle\mathcal O_0)]\to0$, and nonnegativity of the integrand
gives $\widehat{\mathcal C}^{\,\mathrm{peCQR}}_n\to\mathcal O_0$ in probability, proving (i).

\item[(ii)] Same as proof of Theorem~\ref{thm:ccqr-oracle}(ii), with limit $\mathcal{O}_0$ instead of $\mathcal{O}_{\mathrm{m}}$.

\item[(iii)] Same as proof of Theorem~\ref{thm:ccqr-oracle}(iii), with limit $\mathcal{O}_0$ instead of $\mathcal{O}_{\mathrm{m}}$.

\end{enumerate}

\end{proof}

\subsection{Proofs for Appendix~\ref{sec:appendix-auxiliary}} \label{app:proofs-auxiliary}

\begin{proof}[of Lemma~\ref{lem:snapped-from-projected}]
We prove the statement for the lower bounds; the statement for the upper bounds is symmetric. Abbreviate $\hat f:=\hat f^{\mathrm{lo}}_n(X)$ and $q:=q^{\mathrm{lo}}(X)$.

\emph{A pointwise identity.} On $\{\,\cdot>C^{\mathrm L}\}$ one has $\underline\psi(\cdot)=\min\{\cdot,C^{\mathrm R}\}=\Pi(\cdot)$, while on
$\{\,\cdot\le C^{\mathrm L}\}$ one has $\underline\psi(\cdot)=Y^{\min}$ and $\Pi(\cdot)=C^{\mathrm L}$. There are four possible cases:
\begin{itemize}
  \item If $\hat{f} > C^{\mathrm L}$ and $q > C^{\mathrm L}$, then $|\underline\psi(\hat f)-\underline\psi(q)|
  = \bigl|\Pi\hat f-\Pi q\bigr|$;
  \item If $\hat{f} \leq C^{\mathrm L}$ and $q \leq C^{\mathrm L}$, then $|\underline\psi(\hat f)-\underline\psi(q)| = 0 = |\Pi\hat f-\Pi q|$;
  \item If $\hat{f} \leq C^{\mathrm L}$ and $q > C^{\mathrm L}$, then $|\underline\psi(\hat f)-\underline\psi(q)| = \Pi q - Y^{\min} = |\Pi q-\Pi f| + C^{\mathrm L} - Y^{\min}$;
  \item If $\hat{f} > C^{\mathrm L}$ and $q \leq C^{\mathrm L}$, then $|\underline\psi(\hat f)-\underline\psi(q)| = |\Pi \hat f - Y^{\min}| = |\Pi q-\Pi f| + C^{\mathrm L} - Y^{\min}$.
\end{itemize}
Therefore,
\[
  \bigl|\underline\psi(\hat f)-\underline\psi(q)\bigr|
  = \bigl|\Pi\hat f-\Pi q\bigr|
  + (C^{\mathrm L}-Y^{\min})\,\I\{\hat f,q\ \text{straddle}\ C^{\mathrm L}\},
\]
where ``straddle'' means exactly one of $\hat f,q$ lies in $\{\,\cdot\le C^{\mathrm L}\}$. Taking expectations,
\[
  \mathbb E \bigl( \bigl|\underline\psi(\hat f)-\underline\psi(q)\bigr| \bigr)
  = \mathbb E\bigl( \bigl|\Pi\hat f-\Pi q\bigr| \bigr)
  + (C^{\mathrm L}-Y^{\min})\Bigl[\mathbb P(q\le C^{\mathrm L}<\hat f)+\mathbb P(\hat f\le C^{\mathrm L}<q)\Bigr].
\]
Therefore, it remains to show each straddle probability $\to0$. 

\begin{itemize}
  \item \emph{$\mathbb P(q\le C^{\mathrm L}<\hat f)$.} Fix $\eta\in(0,\,C^{\mathrm R}-C^{\mathrm L})$. On the event $\{q\le C^{\mathrm L}<\hat f\}$, we have $\Pi q=C^{\mathrm L}$, and if in addition $\hat f>C^{\mathrm L}+\eta$ then
$\Pi\hat f=\min\{\hat f,C^{\mathrm R}\}>C^{\mathrm L}+\eta$, whence $|\Pi\hat f-\Pi q|>\eta$. Therefore
\[
  \{q\le C^{\mathrm L}<\hat f\}\ \subseteq\ \{|\Pi\hat f-\Pi q|>\eta\}\ \cup\ \{C^{\mathrm L}<\hat f\le C^{\mathrm L}+\eta\}.
\]
By Markov, $\mathbb P(|\Pi\hat f-\Pi q|>\eta)\le\eta^{-1}\mathbb E|\Pi\hat f-\Pi q|\to0$ at the fixed $\eta$, so
\[
  \limsup_n\mathbb P(q\le C^{\mathrm L}<\hat f)\ \le\ \limsup_n\mathbb P\bigl(C^{\mathrm L}<\hat f^{\mathrm{lo}}_n(X)\le C^{\mathrm L}+\eta\bigr),
\]
and the right-hand side $\to0$ as $\eta\downarrow0$ by Assumption~\ref{ass:base-regularity}.

  \item \emph{$\mathbb P( \hat f\le C^{\mathrm L}<q )$.} On the event $\{\hat f\le C^{\mathrm L}<q\}$, $\Pi\hat f=C^{\mathrm L}$, and if in addition $q>C^{\mathrm L}+\eta$ then
$\Pi q=\min\{q,C^{\mathrm R}\}>C^{\mathrm L}+\eta$, whence $|\Pi\hat f-\Pi q|>\eta$. Therefore
\[
  \{\hat f\le C^{\mathrm L}<q\}\ \subseteq\ \{|\Pi\hat f-\Pi q|>\eta\}\ \cup\ \{C^{\mathrm L}<q\le C^{\mathrm L}+\eta\}.
\]
By Markov the first probability $\to0$, so
$\limsup_n\mathbb P(\hat f\le C^{\mathrm L}<q)\le\mathbb P\bigl(C^{\mathrm L}<q^{\mathrm{lo}}(X)\le C^{\mathrm L}+\eta\bigr)$, which $\to0$ as
$\eta\downarrow0$ by Assumption~\ref{ass:regularity}.
\end{itemize}

Both straddle probabilities thus vanish, giving $\mathbb E \bigl( |\underline\psi(\hat f)-\underline\psi(q)| \bigr) \to0$ if $\mathbb E\bigl( \bigl|\Pi\hat f-\Pi q\bigr| \bigr) \to 0$. The argument for the upper bounds is analogous.
\end{proof}

\begin{proof}[of Proposition~\ref{prop:consistency-weaker}]
We treat the lower coordinate; the upper is symmetric, with $\overline\psi$ and $C^{\mathrm R}$ in place of
$\underline\psi$ and $C^{\mathrm L}$. By~\eqref{eq:psi-def}, $\underline\psi$ is nondecreasing and $1$-Lipschitz
apart from a single upward jump of at most $R$ at $C^{\mathrm L}$ (from $Y^{\min}$ to $C^{\mathrm L}$): on
$(-\infty,C^{\mathrm L}]$ it is constant, and on $(C^{\mathrm L},\infty)$ it equals $\min\{\cdot,C^{\mathrm R}\}$,
which is $1$-Lipschitz. Hence for any $a,a'$,
\begin{align*}
    |\underline\psi(a)-\underline\psi(a')|
    & \le |a-a'| + R\,\mathbf 1\{a,a'\ \text{straddle}\ C^{\mathrm L}\} \\
    & \le |a-a'| + R\,\mathbf 1\{ |a' - C^{\mathrm L}| \leq |a-a'|\}.
\end{align*}
Applying this with $a=\hat f^{\mathrm{lo}}_n(X)$, $a'=q^{\mathrm{lo}}(X)$ gives:
\[
    |\underline\psi(\hat f^{\mathrm{lo}}_n(X))-\underline\psi(q^{\mathrm{lo}}(X))|
    \le |\hat f^{\mathrm{lo}}_n(X)-q^{\mathrm{lo}}(X)|
       + R\,\mathbf 1\{|q^{\mathrm{lo}}(X)-C^{\mathrm L}|\le|\hat f^{\mathrm{lo}}_n(X)-q^{\mathrm{lo}}(X)|\},
\]
and taking expectations,
\[
    \mathbb E\bigl( |\underline\psi(\hat f^{\mathrm{lo}}_n)-\underline\psi(q^{\mathrm{lo}})| \bigr)
    \le \mathbb E \bigl( |\hat f^{\mathrm{lo}}_n-q^{\mathrm{lo}} | \bigr)
       + R\,\mathbb P\bigl(|q^{\mathrm{lo}}-C^{\mathrm L}|\le|\hat f^{\mathrm{lo}}_n-q^{\mathrm{lo}}|\bigr).
\]
The first term vanishes by Assumption~\ref{ass:full-consistency}. For the second, fix $\xi>0$ and apply Markov's inequality:
\begin{align*}
  \mathbb P\bigl(|q^{\mathrm{lo}}-C^{\mathrm L}|\le|\hat f^{\mathrm{lo}}_n-q^{\mathrm{lo}}|\bigr)
    & \leq \mathbb P\bigl(|q^{\mathrm{lo}}-C^{\mathrm L}|\le\xi\bigr)
       + \mathbb P\bigl(|\hat f^{\mathrm{lo}}_n-q^{\mathrm{lo}}|>\xi\bigr) \\
    & \leq \mathbb P\bigl(|q^{\mathrm{lo}}-C^{\mathrm L}|\le\xi\bigr)
       + \frac{\mathbb{E} \bigl( |\hat f^{\mathrm{lo}}_n-q^{\mathrm{lo}}| \bigr) }{\xi},
\end{align*}
and the last term vanishes as $n \to \infty$ by Assumption~\ref{ass:full-consistency}. 
Therefore, for any $\xi > 0$,
\begin{align*}  
  \lim_{n \to \infty} \mathbb P\bigl(|q^{\mathrm{lo}}-C^{\mathrm L}|\le|\hat f^{\mathrm{lo}}_n-q^{\mathrm{lo}}|\bigr)
    & \leq P\bigl(|q^{\mathrm{lo}}-C^{\mathrm L}|\le\xi\bigr).
\end{align*}
Finally, letting $\xi \to 0$ makes the right-hand-side term above vanish by Assumption~\ref{ass:regularity}(i).
Hence
$\mathbb E\bigl( |\underline\psi(\hat f^{\mathrm{lo}}_n)-\underline\psi(q^{\mathrm{lo}})| \bigr) \to0$, and symmetrically for the
upper coordinate, which is Assumption~\ref{ass:consistency}.

For strictness, suppose $\mathbb P(q^{\mathrm{lo}}(X)<C^{\mathrm L})>0$ and take $\hat f^{\mathrm{lo}}_n\equiv
C^{\mathrm L}-1$ on $\{q^{\mathrm{lo}}<C^{\mathrm L}\}$ and $\hat f^{\mathrm{lo}}_n\equiv q^{\mathrm{lo}}$ otherwise:
the snapped error is identically $0$ while $\mathbb E(|\hat f^{\mathrm{lo}}_n-q^{\mathrm{lo}}|)$ is a fixed positive
constant, so Assumption~\ref{ass:consistency} holds but Assumption~\ref{ass:full-consistency} fails. If instead only
$\mathbb P(q^{\mathrm{up}}(X)>C^{\mathrm R})>0$, the symmetric construction on the upper coordinate gives the same
conclusion.
\end{proof}

\begin{proof}[of Lemma~\ref{lem:uncensored}]$~$
\begin{enumerate}
\item[(i)] By Proposition~\ref{prop:oracle-coverage}, $F^\star(0)=\mathbb P(Y\in\mathcal O_0(X))=1-\alpha+\gamma=1-\alpha$.
For $t<0$, $\phi^\star(X;t)=\psi(q^{\mathrm{lo}}-t,q^{\mathrm{up}}+t)$ is a strict subset of $\mathcal O_0(X)$ on a set
of positive probability (the endpoints $q^{\mathrm{lo}},q^{\mathrm{up}}$ are interior to $(C^{\mathrm L},C^{\mathrm R})$,
so for small $|t|$ no snapping occurs and the interval genuinely contracts), and the density is strictly positive
(Assumption~\ref{ass:regularity}(ii)); hence $F^\star(t)<1-\alpha$. Therefore
$\tau^{\mathrm{m}}=\inf\{t:F^\star(t)\ge1-\alpha\}=0$ and $\mathcal O_{\mathrm{m}}=\phi^\star(\cdot;0)=\mathcal O_0$.
\item[(ii)] Proposition~\ref{prop:consistency-weaker} says that Assumption~\ref{ass:full-consistency} implies Assumption~\ref{ass:consistency}.
For the converse, suppose Assumption~\ref{ass:consistency} holds; we show
$\mathbb E\bigl( |\hat f^{\mathrm{lo}}_n-q^{\mathrm{lo}}| \bigr) \to0$, the upper coordinate being symmetric. Write
$e^{\mathrm{lo}}_n:=|\underline\psi(\hat f^{\mathrm{lo}}_n)-\underline\psi(q^{\mathrm{lo}})|$, with
$\mathbb E(e^{\mathrm{lo}}_n)\to0$ by assumption. Since $q^{\mathrm{lo}}\in(C^{\mathrm L},C^{\mathrm R})$ a.s.,
$\underline\psi(q^{\mathrm{lo}})=q^{\mathrm{lo}}$, and we partition on the position of $\hat f^{\mathrm{lo}}_n$:
\begin{itemize}
\item On $\{C^{\mathrm L}<\hat f^{\mathrm{lo}}_n\le C^{\mathrm R}\}$, $\underline\psi(\hat f^{\mathrm{lo}}_n)=\hat f^{\mathrm{lo}}_n$,
so $|\hat f^{\mathrm{lo}}_n-q^{\mathrm{lo}}|=e^{\mathrm{lo}}_n$.
\item On $\{\hat f^{\mathrm{lo}}_n\le C^{\mathrm L}\}$, $\underline\psi(\hat f^{\mathrm{lo}}_n)=Y^{\min}$, so
$e^{\mathrm{lo}}_n=q^{\mathrm{lo}}-Y^{\min}$; as $\hat f^{\mathrm{lo}}_n\in[Y^{\min},C^{\mathrm L}]$,
$|\hat f^{\mathrm{lo}}_n-q^{\mathrm{lo}}|=q^{\mathrm{lo}}-\hat f^{\mathrm{lo}}_n\le q^{\mathrm{lo}}-Y^{\min}=e^{\mathrm{lo}}_n$.
\item On $\{\hat f^{\mathrm{lo}}_n>C^{\mathrm R}\}$, $|\hat f^{\mathrm{lo}}_n-q^{\mathrm{lo}}|\le R$ since estimates are in $[Y^{\min},Y^{\max}]$.
\end{itemize}
Combining the three cases,
\[
  \mathbb E\bigl( |\hat f^{\mathrm{lo}}_n-q^{\mathrm{lo}}| \bigr)
  \le \mathbb E \bigl( e^{\mathrm{lo}}_n \bigr) + R\,\mathbb P\bigl(\hat f^{\mathrm{lo}}_n>C^{\mathrm R}\bigr).
\]
It remains to show $\mathbb P(\hat f^{\mathrm{lo}}_n>C^{\mathrm R})\to0$. On $\{\hat f^{\mathrm{lo}}_n>C^{\mathrm R}\}$ we have
$\underline\psi(\hat f^{\mathrm{lo}}_n)=C^{\mathrm R}$ and $\underline\psi(q^{\mathrm{lo}})=q^{\mathrm{lo}}<C^{\mathrm R}$, so
$e^{\mathrm{lo}}_n=C^{\mathrm R}-q^{\mathrm{lo}}$. Hence for any $\eta>0$, by Markov,
\begin{align*}
  \mathbb P(\hat f^{\mathrm{lo}}_n>C^{\mathrm R})
  & = \mathbb P(\hat f^{\mathrm{lo}}_n>C^{\mathrm R}, e^{\mathrm{lo}}_n \geq \eta ) + \mathbb P(\hat f^{\mathrm{lo}}_n>C^{\mathrm R}, e^{\mathrm{lo}}_n < \eta ) \\
  & \leq \mathbb P(e^{\mathrm{lo}}_n \geq \eta ) + \mathbb P( q^{\mathrm{lo}} > C^{\mathrm R} - \eta ) \\
  & \leq \frac{ \mathbb{E} \bigl( e^{\mathrm{lo}}_n \bigr)}{\eta} + \mathbb P( q^{\mathrm{lo}} > C^{\mathrm R} - \eta ).
\end{align*}
Letting $n\to\infty$ then $\eta\downarrow0$, the first term vanishes by hypothesis and the second tends to
$\mathbb P(q^{\mathrm{lo}}\ge C^{\mathrm R})=0$ by Assumption~\ref{ass:regularity}.
Thus
$\mathbb P(\hat f^{\mathrm{lo}}_n>C^{\mathrm R})\to0$ and $\mathbb E \bigl( |\hat f^{\mathrm{lo}}_n-q^{\mathrm{lo}} \bigr) |\to 0$, which is
Assumption~\ref{ass:full-consistency}.
\end{enumerate}
\end{proof}

\begin{proof}[of Lemma~\ref{lem:score-deficit}]
Fix any $x \in \mathcal{X}$. By definition of the score function,
\[
  \{s^{\mathrm{pe},\star}(x, \tilde Y)<0\}
  = \{\Pi q^{\mathrm{lo}}(x)<\tilde Y<\Pi q^{\mathrm{up}}(x)\}
  = \{\Pi q^{\mathrm{lo}}(x)<Y<\Pi q^{\mathrm{up}}(x)\},
\]
because $\Pi q^{\mathrm{lo}}(x),\Pi q^{\mathrm{up}}(x)\in[C^{\mathrm L},C^{\mathrm R}]$. 
Therefore, by continuity of $F_{Y\mid x}$ (Assumption~\ref{ass:regularity}),
\begin{align*}
  g(x) 
  & := \mathbb P\bigl(s^{\mathrm{pe},\star}(X,\tilde Y)<0 \mid X=x\bigr) 
   = F_{Y\mid x}\bigl(\Pi q^{\mathrm{up}}(x)\bigr)-F_{Y\mid x}\bigl(\Pi q^{\mathrm{lo}}(x)\bigr).
\end{align*}
If $\Pi q^{\mathrm{lo}}(x)=\Pi q^{\mathrm{up}}(x)$ the band is degenerate and $g(x)=0\le 1-\alpha$. Otherwise
$\Pi q^{\mathrm{lo}}(x)<\Pi q^{\mathrm{up}}(x)$ forces $q^{\mathrm{lo}}(x)<C^{\mathrm R}$ and $q^{\mathrm{up}}(x)>C^{\mathrm L}$, so
$\Pi q^{\mathrm{lo}}(x)=\max\{q^{\mathrm{lo}}(x),C^{\mathrm L}\}$ and $\Pi q^{\mathrm{up}}(x)=\min\{q^{\mathrm{up}}(x),C^{\mathrm R}\}$; therefore
\begin{align*}
  g(x)
  & = F_{Y\mid x}\bigl(\min\{q^{\mathrm{up}}(x),C^{\mathrm R}\}\bigr)-F_{Y\mid x}\bigl(\max\{q^{\mathrm{lo}}(x),C^{\mathrm L}\}\bigr) \\
  & \leq F_{Y\mid x}\bigl(q^{\mathrm{up}}(x)\bigr)-F_{Y\mid x}\bigl(q^{\mathrm{lo}}(x) \bigr)
   = 1 - \alpha.
\end{align*}

For strictness, suppose $q^{\mathrm{lo}}(x) < C^{\mathrm L}$, which implies $q^{\mathrm{lo}}(x) \leq C^{\mathrm L} - \epsilon$ for some $\epsilon > 0$. If $q^{\mathrm{up}}(x)\le C^{\mathrm L}$ then $g(x)=0<1-\alpha$.
Otherwise $q^{\mathrm{lo}}(x) + \epsilon \leq C^{\mathrm L} = \Pi q^{\mathrm{lo}}(x)<\Pi q^{\mathrm{up}}(x)$.
By Assumption~\ref{ass:regularity}, the interval $(q^{\mathrm{lo}}(x),q^{\mathrm{lo}}(x)+\epsilon)$ has strictly positive probability mass under $F_{Y\mid x}$, so $F_{Y\mid x}\bigl(C^{\mathrm L}\bigr) > F_{Y\mid x}\bigl(q^{\mathrm{lo}}(x)\bigr)$ and
\begin{align*}
  g(x)
  & = F_{Y\mid x}\bigl(\min\{q^{\mathrm{up}}(x),C^{\mathrm R}\}\bigr)-F_{Y\mid x}\bigl(\max\{q^{\mathrm{lo}}(x),C^{\mathrm L}\}\bigr) \\
  & \leq F_{Y\mid x}\bigl(q^{\mathrm{up}}(x)\bigr)-F_{Y\mid x}\bigl(C^{\mathrm L}\bigr) \\
  & < F_{Y\mid x}\bigl(q^{\mathrm{up}}(x)\bigr)-F_{Y\mid x}\bigl(q^{\mathrm{lo}}(x)\bigr) = 1 - \alpha.
\end{align*}
The case $q^{\mathrm{up}}(x)>C^{\mathrm R}$ is symmetric.

Finally, $\{q^{\mathrm{lo}}(X)<C^{\mathrm L}\}$ and $\{q^{\mathrm{up}}(X)>C^{\mathrm R}\}$ coincide a.s.\ with the events $\{q^{\mathrm{lo}}(X)\le C^{\mathrm L}\}$
and $\{q^{\mathrm{up}}(X)\ge C^{\mathrm R}\}$ (no atoms at $C^{\mathrm L},C^{\mathrm R}$, Assumption~\ref{ass:regularity}(i)), whose probabilities
are $\pi^{\mathrm L} = \mathbb{P}(q^{\mathrm{lo}}(X) \le C^{\mathrm L})$ and $\pi^{\mathrm R} = \mathbb{P}(q^{\mathrm{up}}(X) \ge C^{\mathrm R})$, respectively. Since $g\le 1-\alpha$ everywhere and $g<1-\alpha$ on a set of probability
$\pi^{\mathrm L}+\pi^{\mathrm R}$, integrating gives $\mathbb E[g(X)]\le 1-\alpha$, strict iff $\pi^{\mathrm L}+\pi^{\mathrm R}>0$.
\end{proof}

\begin{proof}[of Lemma~\ref{lem:tau-nonneg}]
Write $S_i:=s^{\mathrm{pe}}_n(X_i,\tilde Y_i)$ and $S_i^{\mathrm{m}}:=s^{\mathrm{pe},\star}(X_i,\tilde Y_i)$ for $i\in[n]$, and set
$k:=\lceil(1-\alpha)(n+1)\rceil$, $p_n:=k/n\to1-\alpha$. As $\tilde\tau^{\mathrm{pe}}_n$ is the $k$-th order statistic of
$S_1,\dots,S_n$,
\[
  \{\tilde\tau^{\mathrm{pe}}_n<0\}=\Bigl\{\tfrac1n\textstyle\sum_i\I\{S_i<0\}\ge p_n\Bigr\}.
\]
Define $F^-:=\mathbb P(S^{\mathrm{m}}<0)$ and $\hat F^-_n:=\tfrac1n\sum_i\I\{S_i^{\mathrm{m}}<0\}$. For any $\eta>0$,
\begin{align*}
  \{\tilde\tau^{\mathrm{pe}}_n<0\}
  &=\Bigl\{\hat F^-_n-(p_n-\eta)+\tfrac1n\textstyle\sum_i\bigl(\I\{S_i<0\}-\I\{S_i^{\mathrm{m}}<0\}\bigr)\ge\eta\Bigr\}\\
  &\subseteq \bigl\{\hat F^-_n>p_n-\eta\bigr\}
   \ \cup\ \Bigl\{\tfrac1n\textstyle\sum_i\bigl(\I\{S_i<0\}-\I\{S_i^{\mathrm{m}}<0\}\bigr)\ge\eta\Bigr\} \\
  &\subseteq\underbrace{\bigl\{\hat F^-_n>p_n-\eta\bigr\}}_{=:E_n}
   \ \cup\ \underbrace{\Bigl\{\tfrac1n\textstyle\sum_i\I\{\I\{S_i<0\}\ne\I\{S_i^{\mathrm{m}}<0\}\} \ge\eta\Bigr\}}_{=:G_n} \\
\end{align*}
Since $\pi^{\mathrm L} + \pi^{\mathrm R} > 0$, Lemma~\ref{lem:score-deficit} gives $F^{-} < 1-\alpha$, which means we can pick $\eta > 0$ satisfying also $\eta < 1-\alpha-F^-$. With this choice of $\eta$:
\begin{itemize}
\item \emph{Term $E_n$.} Set $\epsilon :=1-\alpha-F^--\eta>0$ (positive by the choice of $\eta$). Since $p_n \to 1-\alpha$, we also have $p_n - \eta - F^{-} \to \epsilon > 0$. Therefore, $p_n - \eta - F^{-} > \epsilon / 2$ for all sufficiently large $n$. This implies that, for $n$ large enough,
\begin{align*}
  \mathbb{P}(E_n)
  & = \mathbb{P} \left[ \hat{F}^{-}_n - F^{-} > p_n - \eta - F^{-} \right]
  \leq \mathbb{P} \left[ \hat{F}^{-}_n - F^{-} > \epsilon/2 \right] \to 0,
\end{align*}
where the last limit follows because $\hat F^-_n\xrightarrow{a.s.}F^-$ by the strong law of large numbers, since the indicators $\I\{S_i^{\mathrm{m}}<0\}$ are i.i.d.\ with mean $F^-$.

\item \emph{Term $G_n$.} A sign disagreement between $S_i$ $S_i^{\mathrm{m}}$ puts them on opposite sides of $0$, forcing
$|S_i^{\mathrm{m}}|\le|S_i-S_i^{\mathrm{m}}|\le\Delta^{\mathrm{pe}}_n(X_i) := |\Pi\hat f^{\mathrm{lo}}_n(X_i)-\Pi q^{\mathrm{lo}}(X_i)|+|\Pi\hat f^{\mathrm{up}}_n(X_i)-\Pi q^{\mathrm{up}}(X_i)|$, and it cannot occur at a censored point, since $\tilde Y_i=C^{\mathrm L}$
gives $S_i\ge\Pi\hat f^{\mathrm{lo}}_n(X_i)-C^{\mathrm L}\ge0$ and $S_i^{\mathrm{m}}\ge\Pi q^{\mathrm{lo}}(X_i)-C^{\mathrm L}\ge0$ (symmetrically at
$C^{\mathrm R}$), so both indicators vanish.  At the remaining interior points $S_i^{\mathrm{m}}\ne0$ a.s.\
(Assumption~\ref{ass:regularity}); hence a sign disagreement implies $0<|S_i^{\mathrm{m}}|\le\Delta^{\mathrm{pe}}_n(X_i)$. Therefore,
\begin{align*}
  G_n
  & = \Bigl\{\tfrac1n\textstyle\sum_i \I\{\I\{S_i<0\}\ne\I\{S_i^{\mathrm{m}}<0\}\} \ge\eta\Bigr\}
   \subseteq \Bigl\{ \underbrace{ \tfrac1n\textstyle\sum_i\I\{ 0 < |S_i^{\mathrm{m}}| \leq \Delta^{\mathrm{pe}}_n(X_i)\} }_{D_n}  \ge\eta\Bigr\},
\end{align*}
Moreover, for any $\zeta>0$
\[
  D_n
  \le\tfrac1n\textstyle\sum_i\I\{\Delta^{\mathrm{pe}}_n(X_i)>\zeta\}+\tfrac1n\textstyle\sum_i\I\{0<|S_i^{\mathrm{m}}|\le\zeta\}.
\]
Therefore, by a union bound and Markov's inequality,
\begin{align*}
  \mathbb{P}[D_n \geq \eta]
  & \leq \mathbb{P} \left[ \tfrac1n\textstyle\sum_i\I\{\Delta^{\mathrm{pe}}_n(X_i)>\zeta\} \geq \tfrac\eta2 \right] + 
    \mathbb{P} \left[ \tfrac1n\textstyle\sum_i\I\{0<|S_i^{\mathrm{m}}|\le\zeta\}  \geq \tfrac\eta2 \right]  \\
  & \leq \tfrac2\eta \mathbb{P} \left[ \Delta^{\mathrm{pe}}_n(X)>\zeta \right] + 
    \tfrac2\eta \mathbb{P} \left[ 0<|S^{\mathrm{m}}|\le\zeta \right]  \\
  & \leq \tfrac2{\eta\zeta} \mathbb{E} \left[ \Delta^{\mathrm{pe}}_n(X) \right] + 
    \tfrac2\eta \mathbb{P} \left[ 0<|S^{\mathrm{m}}|\le\zeta \right].
\end{align*}
The first term vanishes by Assumption~\ref{ass:consistency}, since $\mathbb E[\Delta^{\mathrm{pe}}_n(X)]\to0$
while $\eta,\zeta$ are fixed. Hence for every $\zeta>0$,
\[
  \limsup_{n\to\infty}\mathbb P[D_n\ge\eta]\ \le\ \tfrac2\eta\,\mathbb P\bigl[0<|S^{\mathrm{m}}|\le\zeta\bigr].
\]
The left-hand side does not depend on $\zeta$; therefore, we can let $\zeta\downarrow0$ and in that limit $P\bigl[0<|S^{\mathrm{m}}|\le\zeta\bigr] \to 0$ by Assumption~\ref{ass:regularity}.
Therefore,
\[
  \mathbb P(G_n)\ \le\ \mathbb P[D_n\ge\eta]\ \to\ 0 .
\]
\end{itemize}
Combining the two terms, $\mathbb P(\tilde\tau^{\mathrm{pe}}_n<0)\le\mathbb P(E_n)+\mathbb P(G_n)\to0$.
\end{proof}

\clearpage

\section{Further Details on Empirical Study} \label{app:empirical}

\subsection{Evaluation of Conditional Coverage} \label{app:worstcase}
We estimate feature-conditional coverage empirically using the \emph{worst-slab}
approach \citep{cauchois2020knowing}: the empirical coverage over the
$q$-fraction of the feature space in which coverage is lowest (we use $q = 0.1$).
Since the worst region is unknown, we estimate it from the test data by sample
splitting, similar to \citet{romano2020classification}.

Let $C_i = \mathbf{1}\{L(X_i) \le Y_i \le U(X_i)\}$ be the per-point coverage
indicator. Then, worst-slab conditional coverage is evaluated as follows.
\begin{enumerate}
  \item Split the test indices uniformly at random into two halves,
        $\mathcal{I}_{\mathrm{fit}}$ and $\mathcal{I}_{\mathrm{eval}}$ (a fixed
        seed makes the split reproducible).
  \item On $\mathcal{I}_{\mathrm{fit}}$, fit a gradient-boosted classifier
        $\hat c(x) \approx \Pr(C = 1 \mid X = x)$ of the coverage indicator on
        the features, using the same hyperparameters as the base predictor.
  \item Score the held-out half, $\hat c(X_i)$ for $i \in
        \mathcal{I}_{\mathrm{eval}}$, and let $\hat q_q$ be the $q$-quantile of
        these scores. The estimated worst region is the lowest-coverage level
        set $S = \{\, i \in \mathcal{I}_{\mathrm{eval}} : \hat c(X_i) \le \hat
        q_q \,\}$.
  \item Report the empirical coverage on that region,
        $\widehat{\mathrm{WC}} = \frac{1}{|S|}\sum_{i \in S} C_i$.
\end{enumerate}
The estimate is returned as \texttt{NA} when fewer than $20$ test points are available.

\FloatBarrier

\subsection{Numerical Experiments with Synthetic Data} \label{app:experiments-synthetic}

We present here additional empirical results from numerical experiments similar
to those described in Section~\ref{sec:experiments}. To make the plots more readable, we 
omit reporting the performance of LdPT, which tends to produce very large prediction intervals.

Figure~\ref{fig:fig2_synthetic_cr} fixes $n = 10{,}000$ and instead varies the censoring
thresholds $C^{\mathrm{R}} = -C^{\mathrm{L}}$. As censoring relaxes (larger
$C^{\mathrm{R}}$), the methods become increasingly similar to one another.
Figure~\ref{fig:fig3_synthetic_dim} fixes $n = 10{,}000$ and
$C^{\mathrm{R}} = 0.5 = -C^{\mathrm{L}}$ and varies the number of features $p$.
As $p$ grows the base predictor becomes less accurate: all intervals widen and
the base model's conditional coverage drops. ClipCQR+ keeps conditional coverage
consistently close to the nominal level, and ClipCQR and eCQR also remain above
the base model.

\begin{figure}[!htb]
  \centering
  \includegraphics[width=\linewidth]{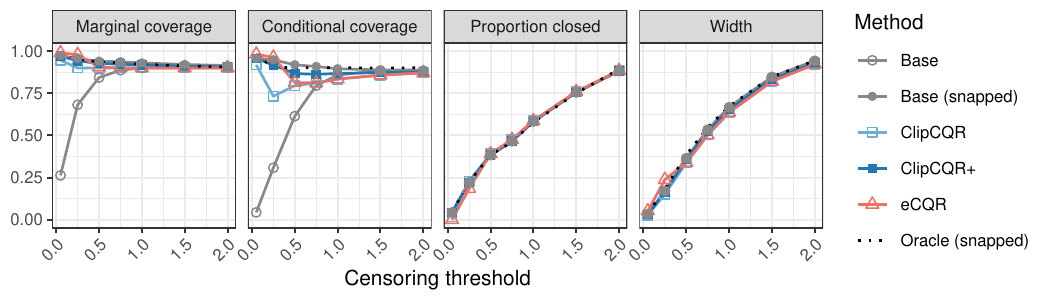}
  \caption{Performance of prediction intervals on synthetic data, as a function of the censoring threshold $C^{\mathrm{R}}=-C^{\mathrm{L}}$, with fixed $n=10{,}000$. Other details are as in Figure~\ref{fig:fig1_synthetic_ncal}.}
  \label{fig:fig2_synthetic_cr}
\end{figure}

\begin{figure}[!htb]
  \centering
  \includegraphics[width=\linewidth]{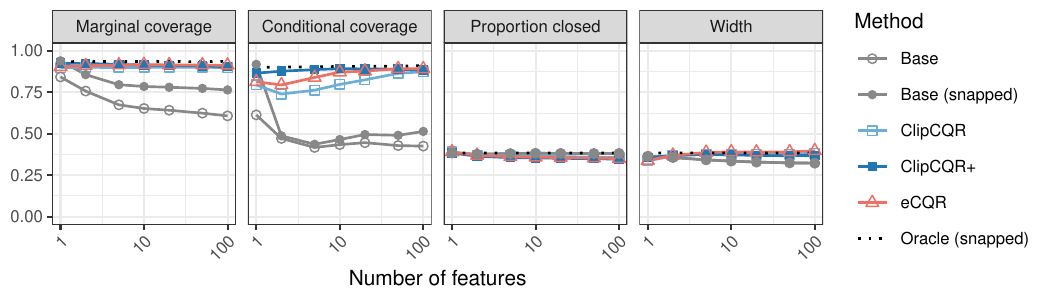}
  \caption{Performance of prediction intervals on synthetic data, as a function of the number of features, with fixed $C^{\mathrm{R}}=0.5=-C^{\mathrm{L}}$ and $n=10{,}000$. Other details are as in Figure~\ref{fig:fig1_synthetic_ncal}.}
  \label{fig:fig3_synthetic_dim}
\end{figure}

\FloatBarrier

Figures~\ref{fig:fig1_synthetic_ncal_large}--\ref{fig:figA3_synthetic_dim_large} report on the results of experiments similar to those of Figures~\ref{fig:fig1_synthetic_ncal}--\ref{fig:fig3_synthetic_dim}, respectively, using different base models. The models used are:
\begin{itemize}
\item \texttt{qrf}: quantile random forest, implemented using the \texttt{R} package \texttt{ranger}, trained to predict the observed outcomes while ignoring censoring;
\item \texttt{xgboost}: gradient-boosting homoscedastic regression model implemented using the \texttt{R} package \texttt{xgboost}, trained to predict the observed outcomes while ignoring censoring;
\item \texttt{xgboost\_tobit}: gradient-boosting homoscedastic regression model implemented using the \texttt{R} package \texttt{xgboost}, trained with a Tobit-like loss function to account for censoring;
\item \texttt{xgboost\_hsk}: gradient-boosting heteroscedastic regression model implemented using the \texttt{R} package \texttt{xgboost}, trained to predict the observed outcomes while ignoring censoring;
  \item \texttt{xgboost\_hsk\_tobit}: gradient-boosting heteroscedastic regression model implemented using the \texttt{R} package \texttt{xgboost}, trained with a Tobit-like loss function to account for censoring, as in Figure~\ref{fig:fig1_synthetic_ncal};
  \item \texttt{qmlp}: a two-layer quantile neural network implemented using the \texttt{python} package \texttt{PyTorch}, trained to predict the observed outcomes while ignoring censoring.
\end{itemize}
The results show that \texttt{xgboost\_hsk} and \texttt{xgboost\_hsk\_tobit} are the only two models with performance approaching that of the population oracle as the sample size grows. In all cases, however, ClipCQR+ tends to produce prediction intervals with higher conditional coverage than the other calibration methods.

Figures~\ref{fig:figA4_synthetic_ncal_peCQR}--\ref{fig:figA6_synthetic_dim}  report on the results of experiments similar to those of Figures~\ref{fig:fig1_synthetic_ncal}--\ref{fig:fig3_synthetic_dim}, respectively, including the peCQR approach in the comparison.
The results show peCQR performs similar to ClipCQR+ when the sample size is very large, and similar to eCQR in smaller samples.

\begin{figure}[!htb]
  \centering
  \includegraphics[width=\linewidth]{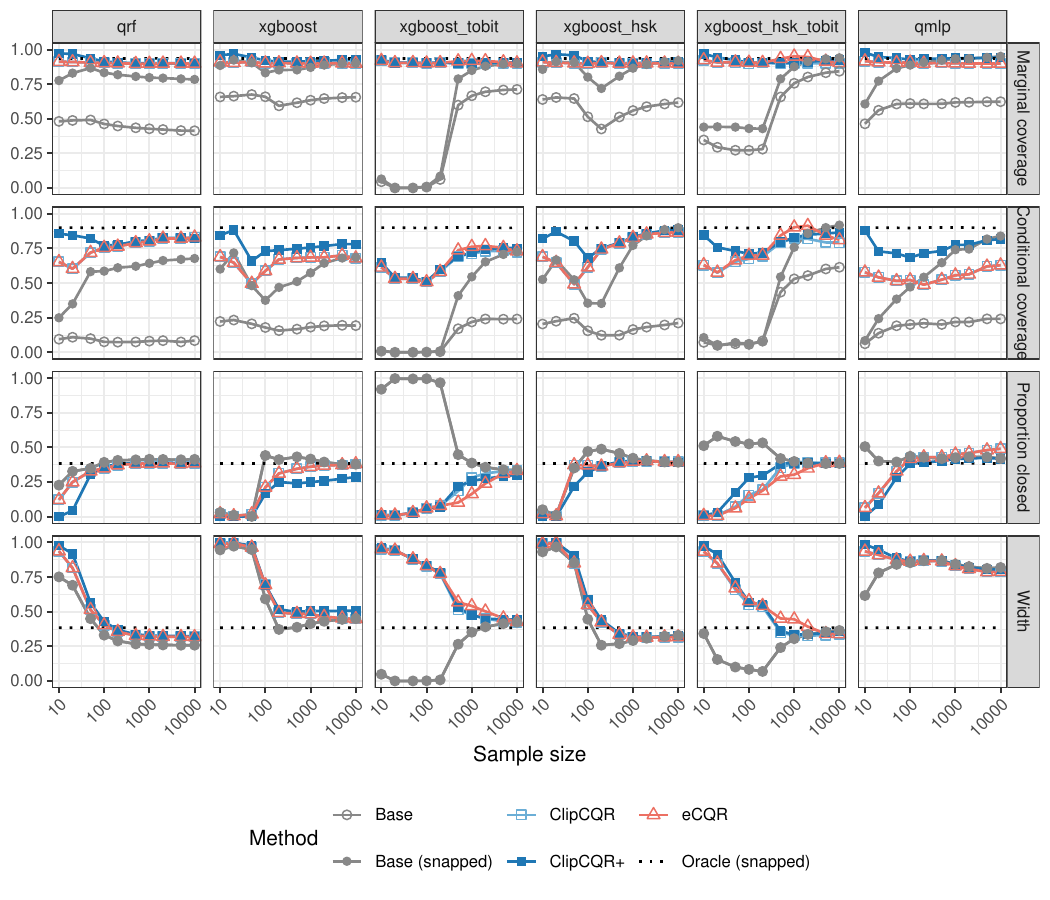}
  \caption{Performance of prediction intervals on doubly censored synthetic data with univariate features, as a function of the training and calibration sample size, for different base models. The censoring thresholds are $C^{\mathrm{R}}=0.5=-C^{\mathrm{L}}$.
The results for the heteroscedastic Tobit-like gradient boosting model (\texttt{xgboost\_hsk\_tobit}) correspond to Figure~\ref{fig:fig1_synthetic_ncal}.}
  \label{fig:fig1_synthetic_ncal_large}
\end{figure}

\begin{figure}[!htb]
  \centering
  \includegraphics[width=\linewidth]{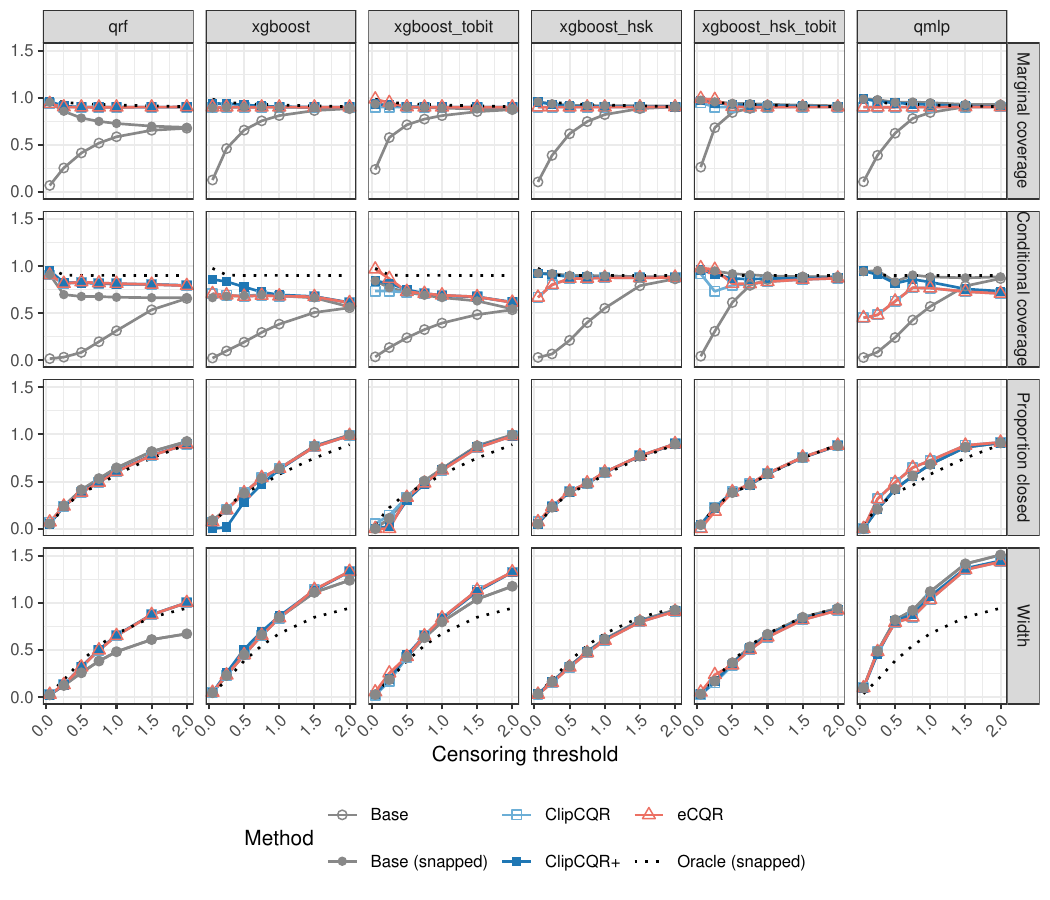}
  \caption{Performance of prediction intervals on doubly censored synthetic data, as a function of the censoring threshold $C^{\mathrm{R}}=-C^{\mathrm{L}}$, for different base models.
The training and calibration sample sizes are $n=10{,}000$. The results for the heteroscedastic Tobit-like gradient boosting model (\texttt{xgboost\_hsk\_tobit}) correspond to Figure~\ref{fig:fig2_synthetic_cr}.
}
  \label{fig:figA2_cr_large}
\end{figure}

\begin{figure}[!htb]
  \centering
  \includegraphics[width=\linewidth]{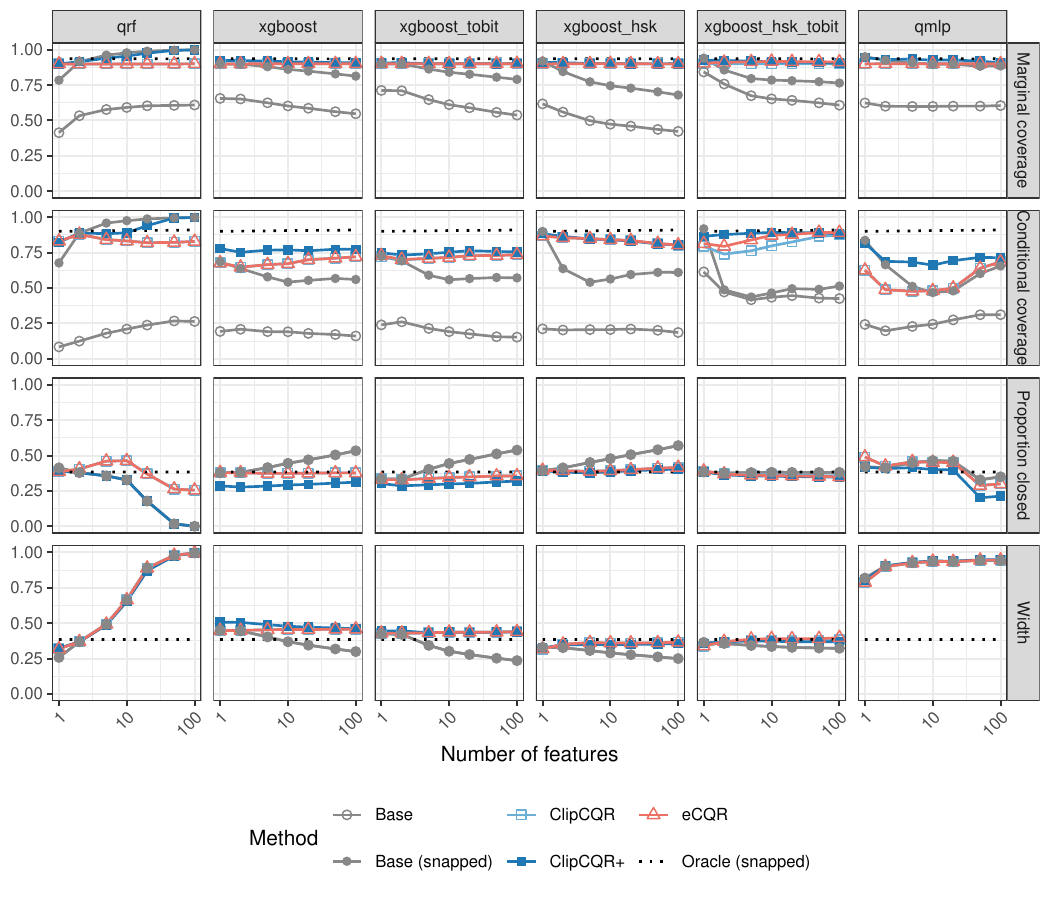}
  \caption{Performance of prediction intervals on doubly censored synthetic data, as a function of the number of features, for different base models.
The censoring thresholds are $C^{\mathrm{R}}=0.5=-C^{\mathrm{L}}$; the training and calibration sample sizes are $n=10{,}000$. The results for the heteroscedastic Tobit-like gradient boosting model (\texttt{xgboost\_hsk\_tobit}) correspond to Figure~\ref{fig:fig3_synthetic_dim}.}
  \label{fig:figA3_synthetic_dim_large}
\end{figure}

\begin{figure}[!htb]
  \centering
  \includegraphics[width=\linewidth]{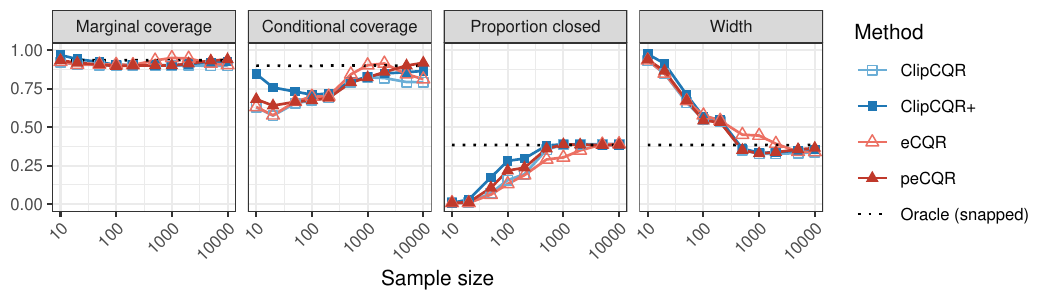}
  \caption{Performance of prediction intervals on doubly censored synthetic data with univariate features, as a function of the training and calibration sample size. The comparison includes peCQR prediction intervals. Other details are as in Figure~\ref{fig:fig1_synthetic_ncal}.}
  \label{fig:figA4_synthetic_ncal_peCQR}
\end{figure}

\begin{figure}[!htb]
  \centering
  \includegraphics[width=\linewidth]{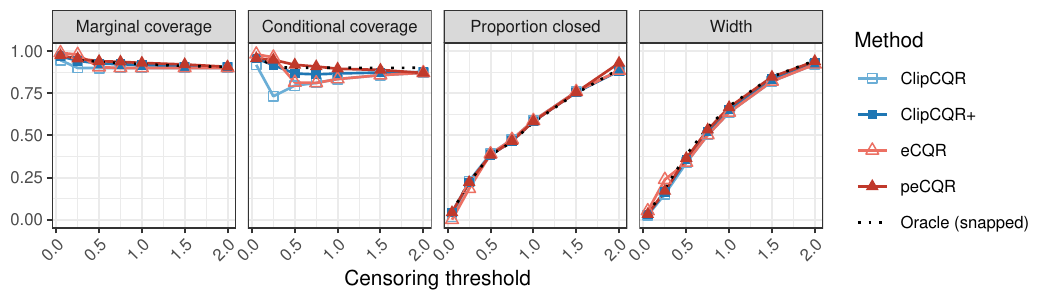}
  \caption{Performance of prediction intervals on doubly censored synthetic data, as a function of the censoring threshold $C^{\mathrm{R}}=-C^{\mathrm{L}}$. The training and calibration sample sizes are $n=10{,}000$. The comparison includes peCQR prediction intervals. Other details are as in Figure~\ref{fig:fig2_synthetic_cr}.}
  \label{fig:figA5_synthetic_cr_peCQR}
\end{figure}

\begin{figure}[!htb]
  \centering
  \includegraphics[width=\linewidth]{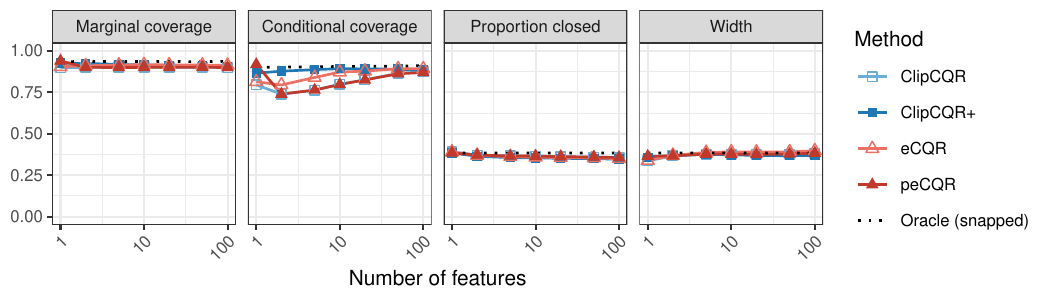}
  \caption{Performance of prediction intervals on doubly censored synthetic data, as a function of the number of features.
The censoring thresholds are $C^{\mathrm{R}}=0.5=-C^{\mathrm{L}}$; the training and calibration sample sizes are $n=10{,}000$. 
The comparison includes peCQR prediction intervals. Other details are as in Figure~\ref{fig:fig3_synthetic_dim}.}
  \label{fig:figA6_synthetic_dim}
\end{figure}

\FloatBarrier

\subsection{QSAR Data Analysis} \label{app:qsar}

We apply ClipCQR+ and the benchmark methods to a real quantitative
structure-activity relationship (QSAR) regression task using the
\emph{3A4} assay from the 2012 ``Merck Molecular Activity Challenge'' Kaggle competition
\citep{ma2015deep}, comprising $37{,}241$ molecules with $9{,}177$ features.
The outcome takes values in $[4.300, 8.134]$ and is left-censored:
about $55.9\%$ of the molecules report the minimum recorded value, which acts as the
left-censoring threshold ($C^{\mathrm{L}} = \min_i Y_i = 4.300$), and the remaining
$44.1\%$ are uncensored ($C^{\mathrm{R}} = \infty$). We chose this data set because it is
publicly available and substantially censored; censoring arises naturally in raw
internal QSAR data but is not always retained in public data sets
\citep{svensson2025enhancing}.

We randomly hold out $n_{\text{cal}} = 1000$ for
calibration and $n_{\text{test}} = 5000$ for evaluation, using one of five different base models
 (\texttt{xgboost}, \texttt{xgboost\_tobit}, \texttt{xgboost\_hsk}, \texttt{xgboost\_hsk\_tobit}, and \texttt{qmlp}, as in Appendix~\ref{app:experiments-synthetic}). 
The prediction intervals produced by each model are calibrated at nominal level $\alpha = 0.1$ using 
the following conformal methods: ClipCQR+, ClipCQR, eCQR, peCQR, and the LdPT method of
\citet{liu2025prediction}.
All methods are applied assuming $[Y^{\min}, Y^{\max}] = [-4, 4]$.
Although the latent
outcome is censored, the snapping construction makes its coverage identifiable: a
snapped band covers the latent $Y$ if and only if it covers the entire censoring
region in which $Y$ lies. We can therefore evaluate prediction intervals 
based on marginal coverage, worst-slab estimate of conditional coverage (computaed as detailed in
Appendix~\ref{app:worstcase}), the proportion of closed intervals (intervals not touching the censoring thresholds), the
average interval width, and the average width of closed intervals.

% \begin{table}[!htb]
% \centering
% \caption{Performance of conformal prediction intervals on the Kaggle 3A4 QSAR data, using a gradient boosting model and different calibration methods. Cells show mean $\pm$ 2\, standard errors(SE). Coverage shown in red when mean $+$ 2\,SE is below the target coverage $1-\alpha = 0.90$; bold marks the shortest interval among methods meeting the coverage threshold.} \label{tab:qsar-kaggle-results-small}
% \centering
% \fontsize{9}{11}\selectfont
% \input{tables/kaggle_3a4_small.tex}
% \end{table}

Table~\ref{tab:qsar-kaggle-results-full} summarizes the results obtained over 10 independent
experiments using different random splits of the data into training, calibration, and test sets.
ClipCQR always achieves marginal coverage very close to the nominal $90\%$ level, as predicted by the
theory, while the other conformal methods are marginally conservative for at least some base models.
ClipCQR+ typically comes closest to the nominal $90\%$ worst-slab conditional coverage, though it does not always
reach it exactly, whereas ClipCQR, eCQR, and peCQR tend to have lower conditional
coverage. The LdPT intervals tend to be wider yet still fall short on conditional coverage. Among the
five base models considered, heteroskedastic gradient boosting (\texttt{xgboost\_hsk}) performs
particularly well, yielding narrower intervals on average than its homoskedastic counterpart
(\texttt{xgboost}) while retaining valid conditional coverage in combination with ClipCQR+. Finally,
these results also show that using a censoring-aware model (\texttt{xgboost\_tobit} or \texttt{xgboost\_hsk\_tobit}) tends
to increase coverage of uncalibrated prediction intervals but does not seem to improve performance
of the conformal methods considered.

\begin{table}[!htb]
\centering
\caption{Performance of conformal prediction intervals on the Kaggle 3A4 QSAR data, for different base models and calibration methods. Cells show a 95\% confidence interval for the mean. The nominal coverage level is $1-\alpha = 0.9$.} \label{tab:qsar-kaggle-results-full}
\centering
\fontsize{9}{11}\selectfont
\begin{tabular}[t]{lrrrrr}
\toprule
Calibration & Coverage & Cond. cov. & Prop. closed & Width & Width $|$ closed\\
\midrule
\addlinespace[0.6em]
\multicolumn{6}{l}{\textbf{xgboost}}\\
\hspace{1em}Base & 0.343 $\pm$ 0.005 & 0.517 $\pm$ 0.039 & 0.225 $\pm$ 0.004 & 0.80 $\pm$ 0.00 & 1.08 $\pm$ 0.00\\
\hspace{1em}Base (snapped) & 0.880 $\pm$ 0.003 & 0.700 $\pm$ 0.024 & 0.225 $\pm$ 0.004 & 0.80 $\pm$ 0.00 & 1.08 $\pm$ 0.00\\
\hspace{1em}ClipCQR & 0.902 $\pm$ 0.006 & 0.743 $\pm$ 0.020 & 0.195 $\pm$ 0.006 & 0.87 $\pm$ 0.01 & 1.19 $\pm$ 0.02\\
\hspace{1em}ClipCQR+ & 0.936 $\pm$ 0.003 & 0.876 $\pm$ 0.009 & 0.103 $\pm$ 0.005 & 0.95 $\pm$ 0.01 & 1.84 $\pm$ 0.06\\
\hspace{1em}eCQR & 0.902 $\pm$ 0.006 & 0.743 $\pm$ 0.020 & 0.195 $\pm$ 0.006 & 0.87 $\pm$ 0.01 & 1.19 $\pm$ 0.02\\
\hspace{1em}peCQR & 0.902 $\pm$ 0.006 & 0.743 $\pm$ 0.020 & 0.195 $\pm$ 0.006 & 0.87 $\pm$ 0.01 & 1.19 $\pm$ 0.02\\
\hspace{1em}LdPT & 0.900 $\pm$ 0.005 & 0.831 $\pm$ 0.015 & 0.036 $\pm$ 0.002 & 1.97 $\pm$ 0.01 & 3.22 $\pm$ 0.02\\
\addlinespace[0.6em]
\multicolumn{6}{l}{\textbf{xgboost\_tobit}}\\
\hspace{1em}Base & 0.338 $\pm$ 0.005 & 0.459 $\pm$ 0.035 & 0.160 $\pm$ 0.003 & 0.68 $\pm$ 0.00 & 1.30 $\pm$ 0.00\\
\hspace{1em}Base (snapped) & 0.889 $\pm$ 0.004 & 0.805 $\pm$ 0.023 & 0.160 $\pm$ 0.003 & 0.68 $\pm$ 0.00 & 1.30 $\pm$ 0.00\\
\hspace{1em}ClipCQR & 0.901 $\pm$ 0.005 & 0.823 $\pm$ 0.020 & 0.149 $\pm$ 0.006 & 0.73 $\pm$ 0.01 & 1.38 $\pm$ 0.03\\
\hspace{1em}ClipCQR+ & 0.916 $\pm$ 0.004 & 0.878 $\pm$ 0.020 & 0.090 $\pm$ 0.004 & 0.78 $\pm$ 0.01 & 1.99 $\pm$ 0.05\\
\hspace{1em}eCQR & 0.919 $\pm$ 0.004 & 0.852 $\pm$ 0.021 & 0.133 $\pm$ 0.006 & 0.80 $\pm$ 0.02 & 1.50 $\pm$ 0.03\\
\hspace{1em}peCQR & 0.901 $\pm$ 0.005 & 0.823 $\pm$ 0.020 & 0.149 $\pm$ 0.006 & 0.73 $\pm$ 0.02 & 1.37 $\pm$ 0.03\\
\hspace{1em}LdPT & 0.900 $\pm$ 0.008 & 0.854 $\pm$ 0.016 & 0.045 $\pm$ 0.003 & 1.53 $\pm$ 0.02 & 2.87 $\pm$ 0.04\\
\addlinespace[0.6em]
\multicolumn{6}{l}{\textbf{xgboost\_hsk}}\\
\hspace{1em}Base & 0.356 $\pm$ 0.005 & 0.419 $\pm$ 0.034 & 0.194 $\pm$ 0.004 & 0.79 $\pm$ 0.01 & 1.18 $\pm$ 0.02\\
\hspace{1em}Base (snapped) & 0.886 $\pm$ 0.005 & 0.819 $\pm$ 0.013 & 0.194 $\pm$ 0.004 & 0.79 $\pm$ 0.01 & 1.18 $\pm$ 0.02\\
\hspace{1em}ClipCQR & 0.901 $\pm$ 0.004 & 0.843 $\pm$ 0.016 & 0.176 $\pm$ 0.006 & 0.81 $\pm$ 0.01 & 1.27 $\pm$ 0.02\\
\hspace{1em}ClipCQR+ & 0.927 $\pm$ 0.004 & 0.901 $\pm$ 0.012 & 0.112 $\pm$ 0.010 & 0.85 $\pm$ 0.01 & 1.70 $\pm$ 0.08\\
\hspace{1em}eCQR & 0.913 $\pm$ 0.003 & 0.861 $\pm$ 0.016 & 0.161 $\pm$ 0.004 & 0.84 $\pm$ 0.01 & 1.35 $\pm$ 0.02\\
\hspace{1em}peCQR & 0.901 $\pm$ 0.004 & 0.843 $\pm$ 0.017 & 0.176 $\pm$ 0.006 & 0.81 $\pm$ 0.01 & 1.27 $\pm$ 0.02\\
\hspace{1em}LdPT & 0.904 $\pm$ 0.008 & 0.855 $\pm$ 0.022 & 0.025 $\pm$ 0.001 & 2.06 $\pm$ 0.01 & 3.64 $\pm$ 0.04\\
\addlinespace[0.6em]
\multicolumn{6}{l}{\textbf{xgboost\_hsk\_tobit}}\\
\hspace{1em}Base & 0.642 $\pm$ 0.008 & 0.663 $\pm$ 0.013 & 0.110 $\pm$ 0.003 & 0.83 $\pm$ 0.01 & 1.43 $\pm$ 0.02\\
\hspace{1em}Base (snapped) & 0.929 $\pm$ 0.003 & 0.876 $\pm$ 0.009 & 0.110 $\pm$ 0.003 & 0.83 $\pm$ 0.01 & 1.43 $\pm$ 0.02\\
\hspace{1em}ClipCQR & 0.904 $\pm$ 0.005 & 0.841 $\pm$ 0.012 & 0.124 $\pm$ 0.005 & 0.74 $\pm$ 0.02 & 1.27 $\pm$ 0.03\\
\hspace{1em}ClipCQR+ & 0.938 $\pm$ 0.004 & 0.905 $\pm$ 0.009 & 0.089 $\pm$ 0.005 & 0.86 $\pm$ 0.01 & 1.72 $\pm$ 0.07\\
\hspace{1em}eCQR & 0.944 $\pm$ 0.004 & 0.897 $\pm$ 0.013 & 0.099 $\pm$ 0.006 & 0.91 $\pm$ 0.02 & 1.57 $\pm$ 0.04\\
\hspace{1em}peCQR & 0.929 $\pm$ 0.003 & 0.876 $\pm$ 0.009 & 0.110 $\pm$ 0.003 & 0.83 $\pm$ 0.01 & 1.43 $\pm$ 0.02\\
\hspace{1em}LdPT & 0.904 $\pm$ 0.007 & 0.893 $\pm$ 0.013 & 0.056 $\pm$ 0.003 & 1.31 $\pm$ 0.02 & 2.27 $\pm$ 0.05\\
\addlinespace[0.6em]
\multicolumn{6}{l}{\textbf{qmlp}}\\
\hspace{1em}Base & 0.356 $\pm$ 0.004 & 0.489 $\pm$ 0.033 & 0.267 $\pm$ 0.018 & 0.85 $\pm$ 0.01 & 1.57 $\pm$ 0.05\\
\hspace{1em}Base (snapped) & 0.880 $\pm$ 0.006 & 0.775 $\pm$ 0.017 & 0.267 $\pm$ 0.018 & 0.85 $\pm$ 0.01 & 1.57 $\pm$ 0.05\\
\hspace{1em}ClipCQR & 0.901 $\pm$ 0.005 & 0.815 $\pm$ 0.024 & 0.234 $\pm$ 0.017 & 0.90 $\pm$ 0.02 & 1.72 $\pm$ 0.04\\
\hspace{1em}ClipCQR+ & 0.918 $\pm$ 0.004 & 0.861 $\pm$ 0.014 & 0.160 $\pm$ 0.018 & 0.92 $\pm$ 0.02 & 2.11 $\pm$ 0.09\\
\hspace{1em}eCQR & 0.901 $\pm$ 0.005 & 0.815 $\pm$ 0.024 & 0.233 $\pm$ 0.018 & 0.90 $\pm$ 0.02 & 1.72 $\pm$ 0.04\\
\hspace{1em}peCQR & 0.901 $\pm$ 0.005 & 0.815 $\pm$ 0.024 & 0.234 $\pm$ 0.017 & 0.90 $\pm$ 0.02 & 1.72 $\pm$ 0.04\\
\hspace{1em}LdPT & 0.901 $\pm$ 0.009 & 0.875 $\pm$ 0.028 & 0.010 $\pm$ 0.001 & 2.19 $\pm$ 0.02 & 5.24 $\pm$ 0.04\\
\bottomrule
\end{tabular}

\end{table}

%%% Local Variables:
%%% mode: latex
%%% TeX-master: "supplement_biometrika"
%%% End:

\end{document}